%% file: main.tex
\documentclass[11pt]{article}

\usepackage[margin=1in]{geometry}
\usepackage{amsmath,amssymb,amsthm,mathtools,bm}
\usepackage{booktabs}
\usepackage{graphicx}
\usepackage{subcaption}
\usepackage{enumitem}
\usepackage{placeins}
\usepackage{natbib}
\usepackage[colorlinks=true,linkcolor=blue,citecolor=blue,urlcolor=blue]{hyperref}
\hypersetup{
  pdftitle={Network Realized GARCH-Ito Models: Volatility Spillovers with High-Frequency Identification},
  pdfauthor={Xinyu Song}
}

\newtheorem{assumption}{Assumption}
\newtheorem{proposition}{Proposition}
\newtheorem{theorem}{Theorem}
\newtheorem{lemma}{Lemma}
\newtheorem{corollary}{Corollary}
\theoremstyle{definition}
\newtheorem{remark}{Remark}

\DeclareMathOperator{\rank}{rank}
\DeclareMathOperator{\diag}{diag}
\DeclareMathOperator{\supp}{supp}
\DeclareMathOperator{\Var}{Var}
\DeclareMathOperator{\Cov}{Cov}
\DeclareMathOperator{\E}{E}
\DeclareMathOperator{\tr}{tr}

\newcommand{\R}{\mathbb R}
\newcommand{\F}{\mathcal F}
\newcommand{\A}{\mathcal A}
\newcommand{\T}{\mathcal T}
\newcommand{\OmegaS}{\Omega}
\newcommand{\op}{\mathrm{op}}
\newcommand{\p}{\mathbb P}

\title{Network Realized GARCH--It\^{o} Models:\\
Volatility Spillovers with High-Frequency Identification}
\author{Xinyu Song\\
\small School of Statistics and Data Science\\
\small Shanghai University of Finance and Economics}
\date{July 2026}

\begin{document}
\maketitle

\begin{abstract}
We introduce a network realized GARCH--It\^{o} model in which volatility
transmission is a dynamic relation among the latent daily integrated
volatilities of multiple assets.
An unknown directed and signed network is embedded in a continuous-time
variance process and appears in the resulting exponential daily recursion.
Intraday returns identify integrated volatility and hence the network that
governs its propagation. For a fixed network template, feasible estimation
based on realized volatility is first-order equivalent to estimation based on
latent integrated volatility. For an unknown network, regularization selects
the relevant structure, and, under oracle conditions, a fixed-rank local refit
provides conditional inference for model-implied response-based connectedness.
In an application to nine U.S. sector ETFs from 2007 to 2025, the model
attains the lowest average out-of-sample QLIKE among the reported recursive
forecasts, although its advantage over HAR is not statistically significant.
In the full sample the weighted LASSO selects an empty sparse support, so
inference concerns the rank-one projected factor component; this conditional
analysis identifies Energy as a net volatility transmitter. Aggregate
connectedness is more stable across volatility measures than individual
sparse channels.
\end{abstract}


\section{Introduction}

This paper studies the identification of, and inference on, cross-asset
volatility transmission in daily integrated volatility. The object of interest
is the directed effect of an asset's lagged integrated volatility on the
conditional integrated volatilities of other assets. Both the state variable
and its transmission law are latent: daily returns provide noisy proxies for
integrated volatility, and the network must be inferred from temporal
variation in the latent volatility process. The model must therefore link
continuous-time volatility, daily integrated-volatility dynamics, and
cross-asset propagation.

The connectedness literature provides influential reduced-form descriptions of
cross-market dependence. Forecast-error variance decompositions from vector
autoregressions measure directional spillovers in realized volatility
\citep{diebold2009,diebold2012,diebold2023}; related approaches measure
systemic dependence or estimate predictive and contemporaneous financial
networks \citep{billio2012,barigozzi2019}. These objects answer how forecast
uncertainty or statistical dependence is distributed across assets. They do
not, however, place a transmission matrix inside the conditional law of
integrated volatility itself. Consequently, their edges and connectedness
shares need not describe how a perturbation to one asset's volatility
propagates through the volatility process.

Structural volatility models approach the problem from the conditional
variance law. Multivariate GARCH and dynamic-correlation models parameterize
cross-asset volatility dependence \citep{engle2002,bauwens2006}, while network
autoregressions and network GARCH models use an observed weight matrix to
organize that dependence \citep{zhu2017,zhou2020}. However, two limitations remain. First,
the network is typically fixed before estimation rather than identified from
the volatility dynamics. Second, these models are usually fitted to daily
returns or squared returns. A daily squared return contains persistent
measurement error as a proxy for integrated volatility, and a discrete GARCH
model need not be asymptotically equivalent to its diffusion limit
\citep{wang2002}. Thus a network estimated from daily squared returns need
not coincide with the network governing continuous-time integrated
volatility.

High-frequency observations address the measurement problem.
Intraday returns consistently recover daily integrated volatility
\citep{barndorff2002,andersen2003}; truncation, pre-averaging, realized
kernels, and noise-robust QMLE address jumps and market microstructure noise
\citep{mancini2009,jacod2009,barndorff2008,xiu2010}. High-frequency
transforms and integrated moments further permit inference on nonlinear
volatility functionals and restrictions among latent volatility components
\citep{todorov2012,lixiu2016,litodorovtauchen2016}. Realized GARCH models use
realized measurements in daily volatility recursions \citep{hansen2012}, and
realized GARCH--It\^{o} models embed such recursions in continuous time
\citep{song2021,kim2024}. This literature establishes how to recover latent
volatility and how to connect intraday price variation to daily dynamics.
The known-network high-frequency specification of \citet{yuan2024}, however,
does not estimate an unknown directed network within the continuous-time
volatility law.

Our contribution is to formulate this latent transmission law as a parameter
of a continuous-time volatility model, identify it from intraday prices, and
conduct inference on the resulting response-based connectedness. To this end,
we introduce a network realized GARCH--It\^{o} model. Starting from a
multivariate continuous-time price process, we construct a positive
instantaneous variance process whose conditional daily integrated volatility
obeys an exponential network recursion. The off-diagonal matrix in this
recursion is unknown, directed, and signed. Its $(i,j)$ entry quantifies the
marginal effect of asset $j$'s lagged log integrated volatility on the
conditional log integrated volatility of asset $i$, net of own-volatility
feedback and persistence. The network is therefore a parameter of the
integrated-volatility dynamics, rather than a network imposed on observed
returns or constructed after fitting an unrestricted forecasting system.

High-frequency observations play an identification role in this model. As the
intraday mesh shrinks, realized-volatility measurement error vanishes and the
feasible recursion approaches the one based on latent integrated volatility.
By contrast, a daily squared return retains nonvanishing log measurement error.
Proposition~\ref{prop:nonequiv} shows the resulting attenuation in a
leverage-free projection benchmark, and Theorem~\ref{thm:asym-normal}
establishes first-order equivalence between feasible and latent-IV estimation
for a fixed and correctly specified network template. High-frequency data thus
identify the dynamic state on which the network is defined; they are not
merely an additional predictor.

For an unknown network, we use projected-factor and sparse regularization to
separate pervasive off-diagonal dependence from directed sparse channels
\citep{chandrasekaran2011,negahban2012}. This decomposition is an estimation
device for moderate-dimensional panels, not a separate economic definition of
connectedness. The fitted stable recursion instead defines the object of
interest: a perturbation to one asset's lagged log integrated volatility is
propagated through the model, and the cumulative responses yield to, from,
net, and total connectedness. These are response-based measures on the
log-volatility scale. They are neither forecast-error variance shares nor
responses to an externally identified causal shock.

The second econometric result concerns inference for these response
functionals. After the sparse support has been selected, a fixed-rank local
refit jointly re-estimates the dynamic and network parameters. Under the stated
oracle and local-identification conditions, its limit law transfers to
connectedness measures and Wald tests. The resulting inference is deliberately
conditional on support selection and fixed rank; the supporting separation,
recovery, and support-selection results establish when this transfer is
available. The result is pointwise and is not uniform over networks with
spillover coefficients approaching zero \citep{leeb2005}.

We apply the model to one-minute observations for the nine original Select
Sector SPDR ETFs from 2007 to 2025. The model delivers the lowest average
QLIKE among the reported recursive forecasts for 2017--2025, although its
advantage over HAR is not statistically significant. The fitted recursion also
produces a model-implied map of sectoral volatility propagation: the
full-sample weighted LASSO selects an empty sparse support, so conditional
inference concerns the rank-one projected factor component and identifies
Energy as a net transmitter. Aggregate connectedness is similar across the
main high-frequency measures, whereas individual sparse edges and annual
transmitter--receiver rankings are less stable and are therefore not
interpreted as an invariant sector hierarchy. The noisy-price estimators used
in the application are robust empirical implementations, whereas the
network-recovery theory is stated for a moderate-dimensional panel and a
generic log-IV proxy expansion.

The remainder of the paper is organized as follows. Section~2 introduces the model and its continuous-time embedding. Section~3 studies estimation with generated realized-volatility regressors. Section~4 develops network recovery and connectedness inference. Sections~5 and 6 report the simulation and empirical results, respectively. Section~7 concludes. The proofs are in the appendices.

\paragraph{Notation.}
For a vector $v$, $\|v\|_2$ is the Euclidean norm. For a matrix $A$, $\|A\|_2$ denotes the spectral norm, $\|A\|_F$ the Frobenius norm, $\|A\|_*$ the nuclear norm, $\|A\|_1=\sum_{i,j}|A_{ij}|$ and $\|A\|_\infty=\max_{i,j}|A_{ij}|$ the entrywise $\ell_1$ and sup norms ($\|A\|_{\max}$ is used interchangeably with $\|A\|_\infty$), and $\|A\|_{\op,\infty}=\max_i\sum_j|A_{ij}|$ the induced maximum-row-sum norm. The spectral radius is $\varrho(A)$, $\langle A,B\rangle=\tr(A^\top B)$, and $\diag(A)$ is the diagonal matrix formed from the diagonal of $A$, so $P_{\mathrm{off}}(A)=A-\diag(A)$ deletes that diagonal. For any norm $\mathcal N$, $\mathcal N^\circ$ denotes its dual norm. Constants $c,C$ are positive, universal unless indexed, and may change from line to line.

\section{Model Setup}

Let $X_t=(X_{1,t},\ldots,X_{N,t})^\top$ denote log prices of $N$ assets. For $i=1,\ldots,N$,
\begin{equation}
  dX_{i,t}=\mu_{i,t}\,dt+\sigma_{i,t}\,dB_{i,t}+dJ_{i,t},
  \label{eq:price}
\end{equation}
where $B_t$ is an $N$-dimensional Brownian motion, $J_{i,t}$ is a finite-activity jump process, and $\sigma_{i,t}^2$ is the instantaneous variance. For integer days $t=1,2,\ldots$, define the integrated volatility and conditional integrated volatility by
\begin{equation}
  IV_{i,t}=\int_{t-1}^{t}\sigma_{i,s}^2\,ds,\qquad
  h_{i,t}=\E\!\left(IV_{i,t}\mid \F_{t-1}\right).
\end{equation}
Let $z_{i,t}=\log IV_{i,t}$ and $y_{i,t}=\log h_{i,t}$.

\begin{proposition}
\label{prop:embedding}
Suppose $h_{i,0}>0$ and $IV_{i,0}>0$ are $\F_0$-measurable, the Brownian filtration can support independent within-day volatility innovations, $\omega$ is finite, and the spillover matrix $M$ has finite row sums. Define the daily log conditional integrated volatility by
\begin{equation}
  y_{i,t}
  =\omega_i+\gamma y_{i,t-1}+\beta z_{i,t-1}
  +\sum_{j=1}^N M_{ij}z_{j,t-1},
  \qquad i=1,\ldots,N,
  \label{eq:log-recursion}
\end{equation}
or, in vector form,
\begin{equation}
  y_t=\omega+\gamma y_{t-1}+(\beta I_N+M)z_{t-1}.
  \label{eq:vector-recursion}
\end{equation}
Then there exists an adapted positive semimartingale volatility process $\sigma_{i,s}^2$ such that $h_{i,t}=\E(IV_{i,t}\mid\F_{t-1})=\exp(y_{i,t})$ and the log recursion \eqref{eq:log-recursion} is generated by the continuous-time price process \eqref{eq:price}. In particular, $h_{i,t}>0$ and $IV_{i,t}<\infty$ almost surely, even when the directed spillover coefficients are signed.
\end{proposition}

Proposition~\ref{prop:embedding} establishes that the log specification is
compatible with a positive continuous-time variance process. It ensures
$h_{i,t}=\exp(y_{i,t})>0$ without restricting the signs of the spillover
coefficients. The parameter $\gamma$ measures persistence in conditional log
integrated volatility, $\beta$ measures own realized-volatility feedback, and
$M$ measures directed cross-asset feedback. We impose $\diag(M)=0$ because own
feedback is represented by $\beta I_N$. Thus $M_{ij}$ maps lagged log integrated
volatility of asset $j$ to the conditional log volatility of asset $i$, and an
edge is directed from $j$ to $i$.

The jump process in \eqref{eq:price} does not enter the volatility recursion directly. The recursion concerns continuous integrated volatility, $IV_{i,t}=\int_{t-1}^t\sigma_{i,s}^2\,ds$. The empirical proxy therefore combines jump truncation with a noise correction, as described in Section~\ref{sec:hf}.

We distinguish two conditional means:
\begin{equation}
  y_{i,t}=\log\E(IV_{i,t}\mid\F_{t-1})=\log h_{i,t},
  \qquad
  q_{i,t}=\E(\log IV_{i,t}\mid\F_{t-1})=\E(z_{i,t}\mid\F_{t-1}).
  \label{eq:two-objects}
\end{equation}
The embedding in Appendix~\ref{app:embedding} gives
$u_{i,t}:=z_{i,t}-y_{i,t}
=\log\!\int_0^1\mathcal E_{i,t}(u)\,du$, where $\mathcal E_{i,t}$ is a positive
martingale with unit conditional mean. Jensen's inequality implies
$\E(u_{i,t}\mid\F_{t-1})=c_{0,i}\le0$, where $c_{0,i}$ is finite and equals zero
only when there is no within-day variance innovation. Hence $q_t=y_t+c_0$,
where $c_0=(c_{0,1},\ldots,c_{0,N})^\top$. Substituting $y_t=q_t-c_0$ into
\eqref{eq:vector-recursion} gives the same slope recursion with a shifted
intercept:
\begin{equation}
  q_t=\widetilde\omega+\gamma q_{t-1}+(\beta I_N+M)z_{t-1},
  \qquad
  \widetilde\omega=\omega+(1-\gamma)c_0 .
  \label{eq:q-recursion}
\end{equation}
The quasi-likelihood criterion of Section~3 has population target $q_t$, the
conditional mean of $z_t$. The recentered innovation
$\widetilde u_{i,t}=z_{i,t}-q_{i,t}=u_{i,t}-c_{0,i}$ is a martingale
difference. Equations \eqref{eq:vector-recursion} and \eqref{eq:q-recursion}
differ only in the intercept: they share the same slope parameters and
therefore imply the same connectedness measures. A level forecast uses
$h_{i,t}=\exp(q_{i,t}-c_{0,i})$, with $c_{0,i}$ estimated by a residual
smearing correction.

For later network regression, iterate \eqref{eq:q-recursion} under $|\gamma|<1$. Let $\mu_z=\E z_t$, $K=\beta I_N+M$, and define the centered filtered design
\begin{equation}
  x_t=\sum_{k=0}^{\infty}\gamma^k(z_{t-1-k}-\mu_z),
  \qquad
  q_t=c_q+K x_t,
  \qquad
  c_q=\frac{\widetilde\omega+K\mu_z}{1-\gamma}.
  \label{eq:filtered-design}
\end{equation}
The staged network regression uses the filtered predictor $x_t$. Its feasible
version $\widehat x_t$ replaces $z_s$, $\mu_z$, and $\gamma$ by
$\widehat z_s$, the sample mean, and the first-stage estimate. All Gram
matrices, factor decompositions, curvature conditions, and defactoring
operations below refer to this filtered design.

The directed spillover matrix is decomposed as
\begin{equation}
  M=R+S,\qquad R=P_{\mathrm{off}}(L)\ \text{for some rank-}r\text{ completion }L,\qquad S \ \text{sparse and zero diagonal}.
  \label{eq:lplusS}
\end{equation}
Only the off-diagonal component $R$ is identified; the diagonal of a completion
$L$ is not observed through $M$. Since $\diag(M)=0$, own-volatility feedback is
identified through $\beta I_N$. We regularize $R$ by the projected nuclear norm
\begin{equation}
  \Omega_*(R)=\inf_{L:P_{\mathrm{off}}(L)=R}\|L\|_* .
  \label{eq:projected-nuclear}
\end{equation}
The component $R$ represents pervasive off-diagonal dependence, and $S$
represents sparse directed spillovers. If $M=\rho W$, the scale of $W$ is
fixed by a normalization such as unit absolute row sums or $\|W\|_F=1$.

Let
\begin{equation}
  A=(\gamma+\beta)I_N+M.
\end{equation}
A sufficient stability condition is
\begin{equation}
  \varrho(A)<1,
  \label{eq:stability}
\end{equation}
where $\varrho(\cdot)$ denotes spectral radius. Under this condition, the log-volatility recursion admits a stationary causal representation.

\section{Estimation and Asymptotic Properties}

\subsection{High-frequency realized-volatility proxies}
\label{sec:hf}

Suppose that $m$ intraday returns are observed each day. The latent quantity
\[
  IV_{i,t}=\int_{t-1}^{t}\sigma_{i,s}^2\,ds
\]
is replaced by a robust estimator $\widehat{IV}_{i,t}$. With jumps and
negligible microstructure noise, threshold realized variance removes returns
above a data-dependent threshold and targets continuous integrated variance
\citep{mancini2009}. With microstructure noise, pre-averaging and realized
kernels provide noise-robust estimators of quadratic variation
\citep{jacod2009,barndorff2008}. Estimator-specific modifications are required
when jumps and noise are both present.

We require the relative-error expansion \citep[see][]{jacod2012}
\begin{equation}
  \frac{\widehat{IV}_{i,t}-IV_{i,t}}{IV_{i,t}}
  =b_{i,t}/m+\eta^{IV}_{i,t},
  \qquad
  \eta^{IV}_{i,t}=O_p(m^{-1/2})
  \label{eq:iv-expansion}
\end{equation}
without microstructure noise. Under noise, the stochastic rate becomes
$O_p(m^{-1/4})$. The bias is negligible under the in-fill conditions below or
can be corrected. Define the feasible log realized volatility by
\begin{equation}
  \widehat z_{i,t}=\log\{\widehat{IV}_{i,t}\vee c_{N,T}\},
\end{equation}
where $c_{N,T}\downarrow0$ is deterministic. Under
Assumption~\ref{ass:price}, trimming is inactive with probability tending to
one. A first-order expansion gives
\begin{equation}
  \widehat z_{i,t}-z_{i,t}
  =
  \frac{\widehat{IV}_{i,t}-IV_{i,t}}{IV_{i,t}}
  +O_p\!\left(\left\{\frac{\widehat{IV}_{i,t}-IV_{i,t}}{IV_{i,t}}\right\}^2\right),
  \label{eq:log-iv-expansion}
\end{equation}
which is the generated-regressor error controlled in Theorem~\ref{thm:asym-normal}.

\subsection{Quasi-likelihood and penalized network estimation}

Let $\theta_f$ collect the fixed-dimensional dynamic parameters, such as
$(\widetilde\omega,\gamma,\beta,\rho)$ in the staged scalar-network model.
Throughout, $q_t(\theta)$ denotes the fitted conditional mean of $z_t$, while
$y_t=\log h_t$ denotes log conditional integrated volatility in the
continuous-time model. We use the identity-weighted Gaussian criterion
\begin{equation}
  Q_T(\theta_f)=\frac1T\sum_{t=1}^T
  \ell_t(\theta_f),\qquad
  \ell_t(\theta_f)=\|z_t-q_t(\theta_f)\|_2^2.
  \label{eq:infeasible-qml}
\end{equation}
The feasible criterion replaces $z_t$ by $\widehat z_t$. A fixed diagonal
weighting is equivalent after row standardization. We do not estimate a
non-diagonal innovation covariance because it couples the response equations
and changes the support Hessian to a Kronecker form that lies outside the
row-wise theory developed below. For network recovery, let
$(\widehat c_q,\widehat\gamma,\widehat\beta)$ be preliminary estimates and
construct $\widehat x_t$ from \eqref{eq:filtered-design}. The residualized
response is
\begin{equation}
  \widehat w_t^{\,\mathrm{res}}
  =\widehat z_t-\widehat c_q-\widehat\beta\widehat x_t.
  \label{eq:residualized-response}
\end{equation}
The preliminary estimates may come from the Stage-0 QMLE or another procedure
satisfying Assumption~\ref{ass:profile}. The regularized network estimator
selects a support for the subsequent refit.

For $M=R+S$, the theoretical network estimator is the fixed-first-stage residualized quadratic estimator
\begin{equation}
  (\widehat R,\widehat S)
  \in\arg\min_{R,S}
  \frac1T\sum_{t=1}^T
  \|\widehat w_t^{\,\mathrm{res}}-(R+S)\widehat x_t\|_2^2
  +\lambda_R\Omega_*(R)+\lambda_S\sum_{i\ne j}a_{ij}|S_{ij}|,
  \label{eq:penalized}
\end{equation}
subject to $\diag(R)=\diag(S)=0$, the stability region for $R+S$, and the explicit row-degree constraint
\begin{equation}
  \max_{1\le i\le N}\|S_{i\cdot}\|_0\le \bar d_N,
  \qquad \bar d_N=2d_N,
  \label{eq:row-degree}
\end{equation}
where $d_N$ is the maximum degree in
Assumption~\ref{ass:network}(iii). This row-degree restriction is a theoretical
localization used in the cone argument. Theorem~\ref{thm:error} applies to any
feasible estimator that satisfies the basic objective inequality relative to
$(R_0,S_0)$; a global minimizer is one example. The row-degree and stability
restrictions are nonconvex. Without them, the fixed-first-stage criterion is
convex, and the theorem applies on the event that the unconstrained solution
has row degree at most $\bar d_N$.

The weights $a_{ij}$ are obtained from a preliminary fit. Write
$\|S\|_{1,a}=\sum_{i\ne j}a_{ij}|S_{ij}|$ and
$\|G\|_{\infty,a}^*=\max_{i\ne j}|G_{ij}|/a_{ij}$. We assume
$0<c_a\le a_{ij}\le C_a<\infty$ on the local event. The implementation first
constructs $\widehat z_t$, estimates the preliminary dynamics, solves
\eqref{eq:penalized}, and refits the selected support without shrinkage.
Theorem~\ref{thm:error} concerns \eqref{eq:penalized}, not a later joint
numerical refinement.

\subsection{Identification and asymptotic properties}

We now state the conditions for identification and asymptotic inference.

\begin{assumption}
\label{ass:price}
The drift and volatility processes in \eqref{eq:price} are adapted, c\`adl\`ag, locally bounded, and satisfy the moment and smoothness conditions required for jump-robust realized-volatility estimation. Jumps have finite activity. For some $p>0$, $\sup_{i,t}\E(IV_{i,t}^{-p})<\infty$. The relative error $\delta_{i,t}=(\widehat{IV}_{i,t}-IV_{i,t})/IV_{i,t}$ admits \eqref{eq:iv-expansion}, with $\max_{i,t}|b_{i,t}|=O_p(1)$; conditionally on the volatility path, its centered part has sub-Gaussian scale $Cm^{-1/2}$ without microstructure noise and $Cm^{-1/4}$ under noise. Choose $c_{N,T}\downarrow0$ so that $NTc_{N,T}^p=o(1)$. Together with Assumption~\ref{ass:rates}, these conditions imply that trimming is inactive and $\max_{i,t}|\delta_{i,t}|\le1/2$ with probability tending to one.
\end{assumption}

\begin{assumption}
\label{ass:dep}
The process $\{z_t\}$ is strictly stationary and ergodic, with moments of order $2+\delta$ for some $\delta>0$. For the fixed-dimensional QMLE, there are $\epsilon>0$ and a neighborhood $\Theta_0$ of $\theta_{f,0}$ such that
\[
  \E\|\nabla q_t(\theta_{f,0})^\top\widetilde u_t\|_2^{2+\epsilon}<\infty,
  \qquad
  \E\sup_{\theta\in\Theta_0}\|\nabla^2\ell_t(\theta)\|_2^{1+\epsilon}<\infty,
\]
where $\widetilde u_t=z_t-q_t$. These conditions provide square integrability and a Lindeberg condition for the score and an integrable envelope for the Hessian; they may alternatively be obtained from stronger primitive moments on $z_t$ and the stable filter.
\end{assumption}

\begin{assumption}
\label{ass:hd-tail}
For high-dimensional network recovery, the filtered design $x_t$ in \eqref{eq:filtered-design} and the recentered innovation $\widetilde u_t=z_t-q_t$ are geometrically beta-mixing and uniformly coordinatewise sub-exponential:
\[
  \max_{j\le N}\|x_{j,t}\|_{\psi_1}\le\tau_x,
  \qquad
  \max_{i\le N}\|\widetilde u_{i,t}\|_{\psi_1}\le\tau_u,
\]
uniformly in $N$ and $t$. The feasible log-IV plug-in error satisfies Lemma~\ref{lem:logiv-filter}. These tail conditions provide primitive background for coordinatewise control and are compatible with pervasive factors because they do not bound every unit-vector projection of $x_t$ uniformly in $N$. The restricted curvature and score bounds actually used below are maintained separately in Assumption~\ref{ass:hd-concentration}.
\end{assumption}

\begin{assumption}
\label{ass:network}
The network decomposition and design satisfy the following conditions.
\begin{enumerate}[label=(\roman*)]
  \item $M_0=R_0+S_0$, where $R_0=P_{\mathrm{off}}(L_0)$ and $L_0$ is the unique minimum-nuclear-norm completion of $R_0$, with $\rank(L_0)=r$.
  When $r=0$, set $R_0=0$ and omit the projected-norm terms below.
  \item For $r\ge1$, the projected nuclear norm is regular at $R_0$.
  Specifically,
  $(\T,\T^\perp)$ is a decomposable pair for $\Omega_*$; there exists
  $E_0\in\partial\Omega_*(R_0)\cap\T$ such that
  $\Omega_*^\circ(E_0)=1$ and $\|E_0\|_\infty\le\xi(\T)$; and, for every
  off-diagonal perturbation $H$, there is an $F_H\in\T^\perp$ such that
  \[
    \Omega_*^\circ(F_H)\le1,\qquad
    \langle F_H,H\rangle=\Omega_*(P_{\T^\perp}H),\qquad
    E_0+F_H\in\partial\Omega_*(R_0).
  \]
  For a fixed constant $C_{\mathrm{tan}}<\infty$, independent of
  $(N,T,m)$,
  \[
    \Omega_*(H)\le C_{\mathrm{tan}}\sqrt{2r}\|H\|_F,\quad H\in\T,
    \qquad
    \|P_\T Z\|_2\le2\|Z\|_2 .
  \]
  For the standard rank-$r$ tangent space
  $\T_{\mathrm{std}}=\{U_0X^\top+YV_0^\top\}$, the spectral bound follows
  from the usual tangent projection formula \citep{chandrasekaran2011}.
  Lemma~\ref{lem:projected-primitives} derives injectivity and the
  effective-rank bound from a diagonal-leakage condition and constructs
  $E_0$ in the hollow case. The decomposability, projected spectral bound,
  and compatible-subgradient property remain maintained conditions.
  \item The sparse component $S_0$ is zero diagonal. With $\A_i=\{j:S_{0,ij}\ne0\}$, let
  \[
    d_N=\max\!\left\{\max_i|\A_i|,
    \max_j|\{i:S_{0,ij}\ne0\}|\right\}
  \]
  denote its maximum in- or out-degree. The separation transversality is
  controlled by $\xi(\T)\mu(\OmegaS)$ in Lemma~\ref{lem:trans}, while the
  RSC cone transversality is measured by $\kappa_{\mathrm{cone}}$ in
  Lemma~\ref{lem:cone-equivalence}; assume $\kappa_{\mathrm{cone}}<3/4$.
  \item The centered filtered design has the approximate factor representation
  \[
    x_t=\Lambda f_t^x+e_t^x,
    \qquad
    \Gamma_x=\E(x_tx_t^\top)
    =\Lambda\Sigma_f^x\Lambda^\top+\Sigma_{e,x}.
  \]
  The factor dimension $r_f$ is fixed or controlled. The eigenvalues of
  $N^{-1}\Lambda^\top\Lambda$ are bounded above and away from zero, and
  $\liminf_N\lambda_{\min}(\Sigma_{e,x})\ge\phi_0>0$. Thus
  $\lambda_{\max}(\Gamma_x)$ may be of order $N$, while the idiosyncratic
  curvature remains nondegenerate. Assumption~\ref{ass:hd-concentration}
  states the sample curvature condition used below.
\end{enumerate}
\end{assumption}

\begin{assumption}
\label{ass:rates}
$m,T\to\infty$. For finite-dimensional inference, $T=o(m^2)$ without noise, and $T=o(m)$ under noise. The current high-dimensional network recovery theory is stated for the no-microstructure-noise or bias-corrected case in which the log-IV plug-in error is $O_p(m^{-1/2})$. In that case the recovery rate restrictions are
\begin{align*}
  \frac{N^2(\log T)^2}{T}&=o(1), &
  \frac{d_N^2\log N(\log T)^2}{T}&=o(1),\\
  \frac{N^2}{m}&=o(1), &
  \frac{d_N^2}{m}&=o(1), &
  \frac{\log(NT)}{m}&=o(1).
\end{align*}
Equivalently, the sampling and plug-in terms vanish on the relevant
projected-factor/sparse cone.
\end{assumption}

\begin{assumption}
\label{ass:orth}
Let $\zeta_t$ be the log-realized-volatility plug-in error in
\eqref{eq:log-iv-expansion}. As in Lemma~\ref{lem:logiv-filter}, write it as
the sum of a conditional bias $b^z_t$, a centered leading term
$\mathring\zeta_t$, and a quadratic remainder $r^z_t$. Estimator-specific
higher-order discretization terms are absorbed into $b^z_t$ and $r^z_t$; the
exact centering below applies only to $\mathring\zeta_t$. Let $\mathcal G$
denote the $\sigma$-field generated by the drift, volatility, and jump paths
and by the daily innovation sequence $\{\widetilde u_t\}$. Thus $\mathcal G$
contains
everything except the within-day discretization randomness and, under noise,
the microstructure randomness. Conditionally on $\mathcal G$, on the localized
event of Lemma~\ref{lem:logiv-filter}:
\begin{enumerate}[label=(\roman*)]
  \item the vectors $\{\mathring\zeta_t\}_{t\le T}$ are independent across $t$ with $\E(\mathring\zeta_t\mid\mathcal G)=0$;
  \item conditionally on $\mathcal G$, they are uniformly vector sub-Gaussian:
  \[
    \sup_{\|v\|_2=1}\|v^\top\mathring\zeta_t\|_{\psi_2\mid\mathcal G}
    \le C m^{-1/2}
  \]
  in the no-noise case, with $Cm^{-1/4}$ under microstructure noise. In particular, each coordinate has conditional second moment at most $C/m$ (respectively $C/\sqrt m$).
\end{enumerate}
These high-level conditions strengthen estimator-specific stable limits
\citep{jacod2012,mancini2009,jacod2009,barndorff2008}. They impose the
uniform vector tails and conditional independence needed here. The
generated-realized-volatility framework follows \citet{song2021,kim2024}.
Set $\varsigma_m=m^{-1/2}$ without noise and
$\varsigma_m=m^{-1/4}$ under noise. The conditions imply that, for any
$\mathcal G$-measurable weight array satisfying
$T^{-1}\sum_t\|g_t\|_2^2=O_p(1)$,
\[
  \Var\!\Big(T^{-1/2}\textstyle\sum_t g_t^\top\mathring\zeta_t\;\Big|\;\mathcal G\Big)
  =T^{-1}\textstyle\sum_t\E\{(g_t^\top\mathring\zeta_t)^2\mid\mathcal G\}
  \le C\varsigma_m^2\,T^{-1}\textstyle\sum_t\|g_t\|_2^2
  =O_p(\varsigma_m^2),
\]
and therefore $T^{-1/2}\sum_tg_t^\top\mathring\zeta_t=o_p(1)$. In the
QMLE score, the weights are functions of the infeasible gradient and the
geometric filter and are $\mathcal G$-measurable. Terms involving the feasible
gradient are quadratic in $\zeta$. The bias terms vanish under the rate
conditions in Assumption~\ref{ass:rates}.
\end{assumption}

\begin{assumption}
\label{ass:hd-concentration}
At the true network and true first-stage parameter, let
$\widehat x_t(\theta_{f,0})$ and
$\widehat w_t^{\,\mathrm{res}}(\theta_{f,0})$ denote the feasible filtered
design and residualized response constructed from $\widehat z_t$ while holding
the dynamic parameters at $\theta_{f,0}$. Define
\[
  \widehat\varepsilon_t^{\,\mathrm f}
  =\widehat w_t^{\,\mathrm{res}}(\theta_{f,0})
   -M_0\widehat x_t(\theta_{f,0}),
  \qquad
  G_T^{\mathrm f}=T^{-1}\sum_{t=1}^T
  \widehat\varepsilon_t^{\,\mathrm f}
  \widehat x_t(\theta_{f,0})^\top .
\]
Define
\[
  a_{S,T}=\sqrt{\frac{(\log T)^2\log N}{T}}+\sqrt{\frac{\log N}{m}},
  \qquad
  a_{R,T}=\frac{N\log T}{\sqrt T}+\frac{N}{\sqrt m}.
\]
For the quadratic network loss in \eqref{eq:penalized},
$G_T^{\mathrm f}$ is minus one half of the gradient with respect to $M$ at
$M_0$, holding the preliminary quantities fixed.
With probability tending to one, the feasible quadratic loss has restricted curvature directly on the joint cone $\mathcal C_\oplus$ defined in \eqref{eq:joint-cone}:
\begin{equation}
  \inf_{\substack{(\Delta_R,\Delta_S)\in\mathcal C_\oplus\\
                  \Delta_R+\Delta_S\ne0}}
  \frac{T^{-1}\sum_{t=1}^T
  \|(\Delta_R+\Delta_S)\widehat x_t(\theta_{f,0})\|_2^2}
  {\|\Delta_R+\Delta_S\|_F^2}
  \ge \underline\phi>0.
  \label{eq:rsc}
\end{equation}
The network score satisfies
\[
  \|G_T^{\mathrm f}\|_{\infty,a}^*=O_p(a_{S,T}),
  \qquad
  \Omega_*^\circ(G_T^{\mathrm f})=O_p(a_{R,T}).
\]
These are high-level conditions for the growing-$N$ result. The order
$a_{R,T}$ allows $\lambda_{\max}(\Gamma_x)\asymp N$ under a pervasive
predictor factor. Proposition~SA.1 in the Online Appendix verifies the
curvature and score bounds for a Gaussian pervasive-factor array with
geometric mixing and an independent Gaussian log-IV perturbation.
\end{assumption}

Assumptions~\ref{ass:dep} and \ref{ass:rates} support the fixed-dimensional
QMLE. Assumptions~\ref{ass:hd-tail} and \ref{ass:hd-concentration} support
network recovery. The projected-norm conditions separate the projected factor
component from sparse spillovers. Irrepresentability and beta-min are imposed
only for support recovery and the post-selection refit.

\begin{proposition}
\label{prop:staged-id}
The baseline studied in this paper uses staged rather than unrestricted one-step identification.
\begin{enumerate}[label=(\roman*)]
  \item In the restricted Stage-0 model, $M_0=\rho_0W^\star$, where $W^\star$ is fixed, deterministic, correctly normalized, and satisfies $\diag(W^\star)=0$. Then $(\widetilde\omega,\gamma,\beta,\rho)$ are identified under stability if the only constants $(a_0,b_1,b_2,b_3)\in\mathbb R^N\times\mathbb R^3$ satisfying
  \[
    a_0+b_1q_{t-1}(\theta_0)+b_2z_{t-1}+b_3W^\star z_{t-1}=0
    \quad\text{a.s.}
  \]
  are $a_0=0$ and $b_1=b_2=b_3=0$.
  \item Conditional on the identified scalar dynamics, $M$ is identified if $\Gamma_z=\E(z_{t-1}z_{t-1}^\top)$ is positive definite on the row space allowed for $M$.
  \item Conditional on the true tangent space and sparse support, the
  decomposition $M=R+S$ is tangent--support identifiable and is the unique
  regularization-selected convex optimum under Theorem~\ref{thm:separation};
  global algebraic uniqueness over all rank--support pairs is not claimed.
\end{enumerate}
\end{proposition}

The Stage-0 template is fixed and correctly specified in
Theorem~\ref{thm:asym-normal}. The theorem does not account for estimation of
$W^\star$ on the same sample. A data-dependent template requires sample
splitting or explicit rate and orthogonality conditions. The high-dimensional
network result instead conditions on preliminary residualized quantities
through Assumption~\ref{ass:profile}.

\begin{proposition}
\label{prop:nonequiv}
In a drift-negligible, leverage-free continuous-martingale benchmark in which the return Brownian motion is independent of the volatility driver, linear projections fitted to daily squared-return proxies target attenuated feedback coefficients rather than the integrated-volatility coefficients.
Let $\F^\sigma$ denote the $\sigma$-field generated by the volatility path.
\begin{enumerate}[label=(\roman*)]
\item \textbf{Level scale.} If
\begin{equation}
  r_{i,t}^2=IV_{i,t}+\nu_{i,t},
  \qquad
  \E(\nu_{i,t}\mid\F^\sigma)=0,\qquad
  \Var(\nu_{i,t}\mid\F^\sigma)=2IV_{i,t}^2,
\end{equation}
then substituting $r_t^2$ for $IV_t$ creates non-vanishing measurement error
at the daily frequency. After removing the intercept, write
$v_t=IV_t-\E(IV_t)$. In the linear projection of the latent level-scale
feedback $K_0v_t$ on the noisy regressor $v_t+\nu_t$, the pseudo-true feedback
matrix is
\begin{equation}
  K^r=K_0\Sigma_v(\Sigma_v+\Sigma_\nu)^{-1},
  \label{eq:daily-attenuation}
\end{equation}
where $\Sigma_v=\Var(v_t)$ and $\Sigma_\nu=\Var(\nu_t)$. In the scalar case,
this reduces to the attenuation factor
\begin{equation}
  \lambda
  =\frac{\Var(IV)}{\Var(IV)+\sigma_\nu^2}<1,
  \qquad \sigma_\nu^2=\E(2IV_t^2).
\end{equation}
\item \textbf{Log scale.} Conditionally on the volatility path $\F^\sigma$, the daily log proxy satisfies
\begin{equation}
  \log r_{i,t}^2=z_{i,t}+\varepsilon_{i,t},
  \qquad
  \varepsilon_{i,t}\mid\F^\sigma\sim\log\chi^2_1,
  \label{eq:log-proxy}
\end{equation}
with $\E(\varepsilon_{i,t})=-(\gamma_{\mathrm E}+\log 2)\approx-1.270$ and
$\Var(\varepsilon_{i,t})=\pi^2/2\approx4.935$, where $\gamma_{\mathrm E}$ is
the Euler--Mascheroni constant; both moments are fixed for every intraday mesh
$m$. Let $x_t$ be the centered filtered design in
\eqref{eq:filtered-design}, so that
$q_{i,t}=c_{q,i}+K_{0,i\cdot}x_t$ with $K_0=\beta I_N+M_0$, and suppose the
$\varepsilon_{i,t}$ are i.i.d.\ across days, independent of the volatility
path, and satisfy $\Var(\varepsilon_t)\succ0$. Fitting the same linear
recursion to the daily proxy replaces $x_t$ by $x^r_t=x_t+e_t$ with
$e_t=\sum_{k\ge0}\gamma^k\{\varepsilon_{t-1-k}-\E(\varepsilon_t)\}$ and
$\Sigma_{\mathrm d}:=\Var(e_t)
=(1-\gamma^2)^{-1}\Var(\varepsilon_t)\succ0$, and the pseudo-true
log-feedback matrix is
\begin{equation}
  K^{r,\log}
  =K_0\,\Gamma_x\,(\Gamma_x+\Sigma_{\mathrm d})^{-1},
  \qquad \Gamma_x=\Var(x_t).
  \label{eq:log-attenuation}
\end{equation}
In the scalar benchmark, or in the isotropic multivariate special case $\Gamma_x=\sigma_x^2I_N$ and $\Sigma_{\mathrm d}=\{(\pi^2/2)/(1-\gamma^2)\}I_N$, the attenuation is the common factor
\begin{equation}
  \lambda_{\log}
  =\frac{\sigma_x^2}{\sigma_x^2+(\pi^2/2)/(1-\gamma^2)}<1.
\end{equation}
Thus the off-diagonal feedback coefficients are scaled toward zero in this special case. In the general multivariate case, \eqref{eq:log-attenuation} is a non-scalar matrix distortion, so no common attenuation factor applies to every connectedness functional. By contrast, the high-frequency proxy has $\Var(\widehat z_{i,t}-z_{i,t})=O(m^{-1})$, so the distortion vanishes as $m\to\infty$.
\end{enumerate}
\end{proposition}

Proposition~\ref{prop:nonequiv} gives a linear-projection comparison. On the
level scale, daily squared returns contain measurement error with conditional
variance $2IV_t^2$. On the log scale, the error is
$\log\chi^2_1$, with variance $\pi^2/2$. Neither error vanishes at the daily
frequency. Winsorization reduces the log-proxy variance but does not eliminate
it. By contrast, the high-frequency log realized-volatility proxy converges to
$z_{i,t}$ at rate $m^{-1/2}$ in the no-noise case.

\begin{theorem}
\label{thm:asym-normal}
Suppose Assumptions \ref{ass:price}, \ref{ass:dep}, \ref{ass:rates}, and
\ref{ass:orth} hold. Let the Stage-0 template $W^\star$ be fixed and correctly
specified as in Proposition~\ref{prop:staged-id}. Let $\Theta_f$ be a compact
parameter space on which $\sup_{\theta\in\Theta_f}|\gamma|<1$ and the recursion
is stable, and let $\theta_f$ be fixed-dimensional with
$\theta_{f,0}\in\operatorname{int}(\Theta_f)$. If the
population Hessian is non-degenerate,
\begin{equation}
  \mathcal H=\E\nabla^2\ell_t(\theta_{f,0})\succ0,
\end{equation}
let $\widetilde\theta_f$ minimize the infeasible criterion
\eqref{eq:infeasible-qml}, and let $\widehat\theta_f$ minimize its feasible
counterpart based on $\widehat z_t$. Then
\begin{equation}
  \sqrt T(\widehat\theta_f-\widetilde\theta_f)=o_p(1),
  \label{eq:first-order-equivalence}
\end{equation}
and both estimators have the same first-order limit:
\begin{equation}
  \sqrt T(\widehat\theta_f-\theta_{f,0})
  \xrightarrow{d}
  \mathcal N(0,\Xi),
  \qquad
  \Xi=\mathcal H^{-1}\mathcal V_0\mathcal H^{-1},
  \label{eq:theta-clt}
\end{equation}
where
\begin{equation}
  \mathcal V_0
  =4\E\!\left[
    \nabla q_t(\theta_{f,0})^\top
    \widetilde u_t\widetilde u_t^\top
    \nabla q_t(\theta_{f,0})
  \right],
  \qquad
  \widetilde u_t=z_t-q_t(\theta_{f,0}).
\end{equation}
A consistent sandwich estimator is obtained by substituting sample analogues evaluated at $\widehat\theta_f$.
\end{theorem}

\begin{remark}
If $T\asymp m^2$ without noise, the feasible estimator acquires an asymptotic bias rather than an inflated variance. The bias can be removed by using a bias-corrected realized-volatility estimator, or avoided by taking $m\gg \sqrt T$.
\end{remark}

\section{Network Structure, Recovery, and Connectedness}

\subsection{Network representation and regularized separation}

Under \eqref{eq:stability}, $M$ determines how a perturbation to lagged log
integrated volatility propagates through the conditional log-volatility
system. This yields a model-implied connectedness object defined directly by
the estimated dynamics rather than by a reduced-form forecast-error variance
decomposition.

\begin{proposition}
\label{prop:vma}
If \eqref{eq:stability} holds and $\{\widetilde u_t\}$ is strictly stationary with $\E\|\widetilde u_t\|_2<\infty$, the log-volatility recursion admits a stationary causal representation. The response matrix
\begin{equation}
  \Phi=(I_N-(\gamma+\beta)I_N-M)^{-1}M
  \label{eq:response}
\end{equation}
and its finite-horizon analogues are well defined. Therefore model-implied to/from/net connectedness measures based on off-diagonal response mass exist.
\end{proposition}

To fix the normalization, define total and off-diagonal absolute response mass by
\[
  D_{\mathrm{all}}(\Phi)=\sum_{k=1}^N\sum_{\ell=1}^N|\Phi_{k\ell}|,
  \qquad
  D_{\mathrm{off}}(\Phi)=\sum_{k\ne\ell}|\Phi_{k\ell}|,
\]
and assume $D_{\mathrm{off}}(\Phi)>0$ (hence $D_{\mathrm{all}}(\Phi)>0$). Since $M_{ij}$ maps a perturbation in asset $j$ into asset $i$, row $i$ measures volatility received by $i$ and column $i$ measures volatility transmitted by $i$. We therefore define
\begin{align}
  C_i^{\mathrm{from}}
  &=\frac{\sum_{j\ne i}|\Phi_{ij}|}{D_{\mathrm{off}}(\Phi)}, &
  C_i^{\mathrm{to}}
  &=\frac{\sum_{j\ne i}|\Phi_{ji}|}{D_{\mathrm{off}}(\Phi)},\nonumber\\
  C_i^{\mathrm{net}}
  &=C_i^{\mathrm{to}}-C_i^{\mathrm{from}}, &
  C^{\mathrm{total}}
  &=\frac{D_{\mathrm{off}}(\Phi)}{D_{\mathrm{all}}(\Phi)}.
  \label{eq:connectedness-functionals}
\end{align}
Hence $\sum_iC_i^{\mathrm{from}}=\sum_iC_i^{\mathrm{to}}=1$,
$\sum_iC_i^{\mathrm{net}}=0$, and $C^{\mathrm{total}}$ is the off-diagonal
share of total absolute response mass. Finite-horizon measures replace $\Phi$
by the corresponding cumulative response matrix.

The impulse experiment is designed to isolate the cross-network channel: a
one-unit perturbation in asset $j$ enters through $Me_j$, with the
own-feedback term $\beta I_N$ held fixed, and then propagates through the
stable recursion. A perturbation to the full observed
state would instead enter through $(\beta I_N+M)e_j$. Thus
\eqref{eq:response} measures cross-network propagation on the log-volatility
scale; it is not a return-shock variance decomposition. A one-percent increase
in integrated volatility scales the level response by
$\log(1.01)\approx0.01$, but leaves the connectedness shares unchanged.
Level responses can be obtained by applying the delta method at a given
forecast origin. We report log-scale measures throughout.

Let $L_0=U_0\Sigma_0V_0^\top$ be the unique minimum-nuclear-norm completion of $R_0$, so that $\Omega_*(R_0)=\|L_0\|_*$. Define the effective off-diagonal tangent space
\begin{equation}
  \T=\{P_{\mathrm{off}}(U_0X^\top+YV_0^\top):X,Y\in\R^{N\times r}\},
\end{equation}
and let $\OmegaS=\{Z:\supp(Z)\subseteq\A\}$, where $\A=\supp(S_0)$. Define
\begin{equation}
  \xi(\T)=\max_{Z\in\T,\|Z\|_2\le1}\|Z\|_\infty,
  \qquad
  \mu(\OmegaS)=\max_{Z\in\OmegaS,\|Z\|_\infty\le1}\|Z\|_2.
\end{equation}

\begin{lemma}
\label{lem:projected-norm}
The projected nuclear norm $\Omega_*$ is a norm on the off-diagonal linear space. Its dual norm is
\begin{equation}
  \Omega_*^\circ(Z)
  =
  \|P_{\mathrm{off}}(Z)\|_2 .
  \label{eq:projected-dual}
\end{equation}
\end{lemma}

\begin{lemma}
\label{lem:trans}
If $2\xi(\T)\mu(\OmegaS)<1$, then $\T\cap\OmegaS=\{0\}$. Moreover $I-P_\Omega P_\T$ is invertible on $\OmegaS$ with a Neumann-series inverse.
\end{lemma}

\begin{theorem}
\label{thm:separation}
Under Assumption~\ref{ass:network} with $r\ge1$, write $\xi=\xi(\T)$ and
$\mu=\mu(\OmegaS)$. If $\xi\mu<1/10$, then the pair $(R_0,S_0)$ is
\emph{tangent--support identifiable}: no other decomposition $M_0=R_1+S_1$ exists
with $R_1-R_0\in\T$ and $S_1-S_0$ supported on $\A$
(Lemma~\ref{lem:trans}), and, for
\begin{equation}
  \frac{\lambda_S}{\lambda_R}\in
  \left(
    \frac{\xi}{1-4\xi\mu},
    \frac{1-5\xi\mu}{3\mu}
  \right),
  \label{eq:tuning-interval}
\end{equation}
the pair $(R_0,S_0)$ is the unique optimum of
\begin{equation}
  \min_{R,S}\lambda_R\Omega_*(R)+\lambda_S\|S\|_1
  \quad\text{subject to}\quad R+S=M_0,\ \diag(R)=\diag(S)=0.
\end{equation}
\end{theorem}

The constants in Theorem~\ref{thm:separation} are not optimized.
Lemma~\ref{lem:trans} gives $\T\cap\OmegaS=\{0\}$ under the weaker condition
$2\xi\mu<1$. The algebraic identification statement is local to the true
tangent space and support, whereas the regularization-selected optimum is
unique over the feasible set of the stated convex program. Global algebraic
uniqueness over all rank--support pairs would require uniform transversality,
as in \citet{chandrasekaran2011}.

The interval \eqref{eq:tuning-interval} applies to the noiseless decomposition
of a known $M_0$. The stochastic error bound below instead uses the joint-cone
curvature and score conditions. To impose both results, the penalty ratio must
satisfy \eqref{eq:tuning-interval} and both penalties must dominate their
respective scores.

\subsection{Network recovery and connectedness inference}

For network perturbations decomposed as $\Delta_R+\Delta_S$, define the
penalty-weighted joint cone of perturbation pairs
\begin{equation}
  \begin{aligned}
  \mathcal C_\oplus
  =
  \big\{(\Delta_R,\Delta_S):\;&
  \lambda_R\,\Omega_*(P_{\T^\perp}\Delta_R)
  +\lambda_S\|P_{\Omega^c}\Delta_S\|_{1,a}
  \le
  3\big(\lambda_R\,\Omega_*(P_\T\Delta_R)+\lambda_S\|P_\Omega\Delta_S\|_{1,a}\big),\\
  &\sum_{i=1}^N\|\Delta_{S,i\cdot}\|_1^2\le 4d_N\|\Delta_S\|_F^2
  \big\}.
  \end{aligned}
  \label{eq:joint-cone}
\end{equation}
The first condition is the weighted cone implied by the basic inequality. The
second is stated row-wise because a global bound
$\|\Delta_S\|_1^2\le Cd_N\|\Delta_S\|_F^2$ fails when many rows are active.
Under \eqref{eq:row-degree}, each row of
$\Delta_S=\widehat S-S_0$ has at most $3d_N$ nonzero entries. Row-wise
Cauchy--Schwarz then gives the second condition.

\begin{lemma}
\label{lem:cone-equivalence}
Suppose the projected transversality margin satisfies
\begin{equation}
  2|\langle \Delta_R,\Delta_S\rangle|
  \le
  \kappa_{\mathrm{cone}}\{\|\Delta_R\|_F^2+\|\Delta_S\|_F^2\},
  \qquad (\Delta_R,\Delta_S)\in\mathcal C_\oplus,
  \label{eq:cone-transversality}
\end{equation}
for some $\kappa_{\mathrm{cone}}<1$. Then every
$(\Delta_R,\Delta_S)\in\mathcal C_\oplus$ obeys
\begin{equation}
  \|\Delta_R\|_F^2+\|\Delta_S\|_F^2
  \le
  \frac{1}{1-\kappa_{\mathrm{cone}}}
  \|\Delta_R+\Delta_S\|_F^2.
  \label{eq:cone-equivalence}
\end{equation}
\end{lemma}

\begin{assumption}
\label{ass:profile}
Let $\mathcal L_T(R,S;\widehat\theta_f)$ denote the residualized quadratic loss in \eqref{eq:penalized}, where $\widehat\theta_f$ collects the preliminary quantities used to construct $\widehat w_t^{\,\mathrm{res}}$ and $\widehat x_t$. Their replacement of the true dynamic parameters contributes asymptotically negligible score terms:
\[
  \|\nabla_M\mathcal L_T(R_0,S_0;\widehat\theta_f)
    -\nabla_M\mathcal L_T(R_0,S_0;\theta_{f,0})\|_{\max}
  =o_p(\lambda_S),
\]
and
\[
  \Omega_*^\circ\!\left\{\nabla_M\mathcal L_T(R_0,S_0;\widehat\theta_f)
    -\nabla_M\mathcal L_T(R_0,S_0;\theta_{f,0})\right\}
  =o_p(\lambda_R).
\]
The same first-stage replacement changes the restricted curvature by $o_p(1)$ on $\mathcal C_\oplus$. This condition is satisfied by a Neyman-orthogonal residualized score or by a first-stage estimator whose error is of smaller order than the network penalties.

For fixed support and rank, the condition follows from an orthogonal-score
expansion and a fixed-dimensional first-stage rate. With growing support, it is
maintained as a local condition on the residualized network score.
\end{assumption}

\begin{theorem}
\label{thm:error}
Under Assumptions~\ref{ass:network}, \ref{ass:rates},
\ref{ass:hd-concentration}, and \ref{ass:profile}, bounded adaptive weights,
and penalty levels
\begin{equation}
  \lambda_S\asymp a_{S,T},
  \qquad
  \lambda_R\asymp a_{R,T},
\end{equation}
with multiplicative constants sufficiently large relative to the score constants in Assumption~\ref{ass:hd-concentration}, any feasible estimator in \eqref{eq:penalized} that satisfies the basic objective inequality relative to $(R_0,S_0)$ obeys
\begin{equation}
  \|\widehat R-R_0\|_F^2+\|\widehat S-S_0\|_F^2
  \le
  \frac{C}{\underline\phi^2(1-\kappa_{\mathrm{cone}})^2}
  \{\lambda_R^2r+\lambda_S^2s_N\}
\end{equation}
with probability tending to one, where $s_N=|\A|$ and $C$ depends only on
the adaptive-weight bounds, $C_{\mathrm{tan}}$, and fixed numerical constants.
When $r=0$, the projected component and the terms involving $\lambda_R$ are
omitted.
\end{theorem}

Because the first-step estimator targets rate optimality rather than
selection, support recovery is established for a second-step defactored
weighted LASSO. Let $\widehat R$ be the first-step projected-factor estimate. Let $r_f$ denote
the number of factors in the predictor design, $P_\Lambda$ the population
projector onto their loading space, and $\widehat P_{r_f}$ its POET estimate.
The factor dimension $r_f$ is distinct from the completion rank $r$. Define
\begin{equation}
  \widetilde x_t=(I-\widehat P_{r_f})\,\widehat x_t,
  \qquad
  \widetilde w_t=\widehat w_t^{\,\mathrm{res}}-\widehat R\,\widehat x_t,
  \label{eq:defactored-data}
\end{equation}
where $\widehat w_t^{\,\mathrm{res}}$ is defined in \eqref{eq:residualized-response}. For response row $i$, write $\mathcal J_i=\{1,\ldots,N\}\setminus\{i\}$, $\A_i=\{j\in\mathcal J_i:S_{0,ij}\ne0\}$, and $\A_i^c=\mathcal J_i\setminus\A_i$; the global edge support is $\A=\{(i,j):j\in\A_i\}$. The second-step estimator is the unconstrained row-wise weighted LASSO
\begin{equation}
  \widehat S^{(2)}
  =\arg\min_{S:\,\diag(S)=0}
  \underbrace{\frac1{2T}\sum_{t=1}^T
  \|\widetilde w_t-S\widetilde x_t\|_2^2}_{\mathcal L_T^{(2)}(S)}
  +\lambda_S\sum_{i\ne j}a_{ij}|S_{ij}| .
  \label{eq:two-step}
\end{equation}
Its score at $S_0$ decomposes into an estimated-defactored score plus the projected-factor estimation contamination
\begin{equation}
  \Delta_T
  =\frac1T\sum_{t=1}^T\bigl\{(R_0-\widehat R)\,\widehat x_t+S_0\widehat P_{r_f}\,\widehat x_t\bigr\}\widetilde x_t^\top .
  \label{eq:contamination}
\end{equation}

\begin{assumption}
\label{ass:defactor}
Let
\[
  \widehat{\widetilde\Gamma}=T^{-1}\sum_{t=1}^T\widetilde x_t\widetilde x_t^\top,
  \qquad
  \widetilde\Gamma=(I-P_\Lambda)\Gamma_x(I-P_\Lambda)
  =(I-P_\Lambda)\Sigma_{e,x}(I-P_\Lambda),
\]
and define the estimated-defactored stochastic score by
\[
  G_T^{\mathrm{df}}
  =\nabla_S\mathcal L_T^{(2)}(S_0)+\Delta_T,
\]
where $\mathcal L_T^{(2)}$ is the second-step loss. For
\[
  a_{\mathrm{df},T}
  =\sqrt{\frac{(\log T)^2\log N}{T}}+N^{-1/2}+m^{-1/2},
\]
assume
\begin{align*}
  \|\widehat{\widetilde\Gamma}-\widetilde\Gamma\|_{\max}
  &=O_p(a_{\mathrm{df},T}),\\
  \|G_T^{\mathrm{df}}\|_{\max}
  &=O_p\!\left(\sqrt{\frac{(\log T)^2\log N}{T}}
    +\sqrt{\frac{\log N}{m}}\right),\\
  \|\Delta_T\|_{\max}&=o_p(\lambda_S).
\end{align*}
These conditions concern the estimated projector constructed from
$\{\widehat x_t\}$. Since $\widehat P_{r_f}$ uses the full sample, we do not
deduce mixing of the transformed design from mixing of $x_t$. Primitive POET
conditions also require pervasiveness, an eigengap, restrictions on
$\Sigma_{e,x}$, and suitable tails and dependence \citep{fan2013}. Sample
splitting or cross-fitting provides another route. The projector $P_\Lambda$
acts on the predictor factor space and is unrelated to the singular spaces of
$L_0$. The leakage bound is imposed separately because the first-step error in
$\widehat R$ need not be smaller than $\lambda_S$.
\end{assumption}

\begin{theorem}
\label{thm:oracle}
Suppose the conditions of Theorem~\ref{thm:error} and
Assumption~\ref{ass:defactor} hold, and the multiplicative constant in
$\lambda_S$ is sufficiently large relative to the sparse-score constant and
the irrepresentability margin $\varepsilon$. Let
$\widetilde\Gamma=(I-P_\Lambda)\Gamma_x(I-P_\Lambda)$ and
$\widehat{\widetilde\Gamma}=T^{-1}\sum_t\widetilde x_t\widetilde x_t^\top$.
For $B\subseteq\mathcal J_i$, define
$D_{i,B}=\diag(a_{ij}:j\in B)$. Expressions involving an empty $\A_i$ are
vacuous. Suppose the nonempty active blocks are nonsingular and, with
probability tending to one,
\[
  \max_i\|\widetilde\Gamma_{\A_i\A_i}^{-1}\|_{\op,\infty}
  \le\phi_{\min}^{-1},
  \qquad
  \max_i\|\widehat{\widetilde\Gamma}_{\A_i\A_i}^{-1}\|_{\op,\infty}
  \le2\phi_{\min}^{-1}
\]
for some $\phi_{\min}>0$. Also assume, uniformly over response rows $i$, the weighted population and empirical irrepresentability conditions
\begin{equation}
  \max_i\left\|D_{i,\A_i^c}^{-1}
  \widetilde\Gamma_{\A_i^c\A_i}
  \widetilde\Gamma_{\A_i\A_i}^{-1}D_{i,\A_i}\right\|_{\op,\infty}
  \le 1-\varepsilon,
  \qquad
  \max_i\left\|D_{i,\A_i^c}^{-1}
  \widehat{\widetilde\Gamma}_{\A_i^c\A_i}
  \widehat{\widetilde\Gamma}_{\A_i\A_i}^{-1}D_{i,\A_i}\right\|_{\op,\infty}
  \le 1-\varepsilon/2
  \label{eq:irrep-main}
\end{equation}
with probability tending to one. Lemma~\ref{lem:irrep} gives primitive
diagonal-dominance and cross-leakage conditions sufficient for the population
inequality and shows how the estimated-defactored Gram rate transfers its
margin to the empirical inequality. Finally, assume the beta-min condition
\begin{equation}
  \min_{i:\A_i\ne\varnothing}\min_{j\in\A_i}|S_{0,ij}|
  \gg \lambda_S/\phi_{\min}
\end{equation}
holds. Then the second-step estimator recovers the support and signs:
\begin{equation}
  \p\{\supp(\widehat S^{(2)})=\A,\ \operatorname{sign}(\widehat S^{(2)}_\A)=\operatorname{sign}(S_{0,\A})\}\to1.
\end{equation}
\end{theorem}

Theorem~\ref{thm:oracle} concerns support selection in the second-step
estimator, not the final inference refit, to which we now turn. We consider a
differentiable connectedness functional $g$ constructed from the response
matrix in \eqref{eq:response}; once the refit coordinates $\eta$ are
introduced below, $g$ is regarded as a function $g(\eta)$ of those
coordinates. The refit is parameterized by the identified component $R$. For a fixed
completion rank $r$, let
\[
  \mathcal R_r=\{P_{\mathrm{off}}(L):\rank(L)=r\}.
\]
Near $R_0$, use a smooth one-to-one chart $R(\vartheta)$ for the locally
regular part of $\mathcal R_r$. One explicit construction starts from a
nonsingular $r\times r$ submatrix of $L_0$ and, after row and column
permutations, writes
\begin{equation}
  L(\vartheta)=
  \begin{pmatrix}
    D & F\\
    G & GD^{-1}F
  \end{pmatrix},
  \qquad
  \vartheta=\{\operatorname{vec}(D)^\top,\operatorname{vec}(F)^\top,
  \operatorname{vec}(G)^\top\}^\top .
  \label{eq:rank-chart}
\end{equation}
This is a unique local chart for rank-$r$ matrices. Assume that
$P_{\mathrm{off}}$ has constant rank on the chart. After removing any locally
redundant coordinates, $R(\vartheta)=P_{\mathrm{off}}\{L(\vartheta)\}$ is
one-to-one and has a full-column-rank Jacobian. When $r=0$, set
$R(\vartheta)\equiv0$.

For a candidate sparse support $\mathcal B$, write
\[
  S_{\mathcal B}(\alpha)
  =\sum_{(i,j)\in\mathcal B}\alpha_{ij}e_ie_j^\top,
  \qquad
  \eta_{\mathcal B}
  =(\widetilde\omega^\top,\gamma,\beta,\vartheta^\top,\alpha^\top)^\top ,
\]
and let $\widehat q_t(\eta_{\mathcal B})$ be the recursion
\eqref{eq:q-recursion} driven by $\widehat z_{t-1}$ with
$M(\eta_{\mathcal B})=R(\vartheta)+S_{\mathcal B}(\alpha)$. The unpenalized
local-refit criterion and stability constraint are
\begin{equation}
  \widehat Q_{T,\mathcal B}(\eta_{\mathcal B})
  =
  \frac1T\sum_{t=1}^T
  \|\widehat z_t-\widehat q_t(\eta_{\mathcal B})\|_2^2,
  \qquad
  \kappa_{\mathcal B}(\eta_{\mathcal B})
  =
  \varrho\{(\gamma+\beta)I_N+R(\vartheta)+S_{\mathcal B}(\alpha)\}
  -\bar\varrho
  \le0,
  \quad \bar\varrho<1 .
  \label{eq:oracle-refit}
\end{equation}
The cap $\bar\varrho$ is fixed and chosen above the true spectral radius. Conditional on a fixed $\gamma$ and the corresponding filtered initialization, the same matrix target can be written using the filtered design:
\begin{equation}
  \widehat Q^{\,\mathrm{fil}}_{T,\mathcal B}(c_q,\beta,R,\alpha;\gamma)
  =
  \frac1T\sum_{t=1}^T
  \|\widehat z_t-c_q-\{\beta I_N+R+S_{\mathcal B}(\alpha)\}
  \widehat x_t(\gamma)\|_2^2 ,
  \qquad R\in\mathcal R_r,
  \label{eq:oracle-refit-filtered}
\end{equation}
with the same stability cap. Equations
\eqref{eq:oracle-refit}--\eqref{eq:oracle-refit-filtered} define the numerical
refit. A factorized algorithm may search for a constrained stationary point.
An ordinary SVD of the unrestricted coefficient matrix does not solve this
problem unless the design makes that projection exact. At an interior local
solution, the KKT conditions are
\begin{equation}
  \nabla_{(\widetilde\omega,\gamma,\beta,\alpha)}
  \widehat Q_{T,\mathcal B}=0,
  \qquad
  \left\{\frac{\partial\operatorname{vec}R(\vartheta)}
  {\partial\vartheta^\top}\right\}^{\!\top}
  \operatorname{vec}(\nabla_R\widehat Q_{T,\mathcal B})=0,
  \qquad
  \kappa_{\mathcal B}(\eta_{\mathcal B})<0.
  \label{eq:oracle-kkt}
\end{equation}
If the stability cap binds and the stable-transition matrix has a unique,
real, simple eigenvalue of maximal modulus in a neighborhood of the solution,
the corresponding KKT conditions are
\begin{equation}
\begin{split}
  &\nabla_{(\widetilde\omega,\gamma,\beta,\alpha)}
  \widehat Q_{T,\mathcal B}
  +\tau_{\mathcal B}
  \nabla_{(\widetilde\omega,\gamma,\beta,\alpha)}
  \kappa_{\mathcal B}=0,\\
  &\left\{\frac{\partial\operatorname{vec}R(\vartheta)}
  {\partial\vartheta^\top}\right\}^{\!\top}
  \operatorname{vec}\!\left(
  \nabla_R\widehat Q_{T,\mathcal B}
  +\tau_{\mathcal B}\nabla_R\kappa_{\mathcal B}\right)=0,\\
  &\tau_{\mathcal B}\ge0,\qquad
  \kappa_{\mathcal B}(\eta_{\mathcal B})=0 .
\end{split}
\label{eq:oracle-kkt-active}
\end{equation}
Otherwise, \eqref{eq:oracle-kkt-active} is interpreted as the corresponding
normal-cone KKT inclusion for the stability-feasible set.

For each support $\mathcal B$, let $\bar\eta_{\mathcal B}$ be the measurable
initialization from the regularized fit. A fixed measurable rule
$\mathcal S_T(\mathcal B,\bar\eta_{\mathcal B})$ selects a constrained-KKT
point in its local basin. This point satisfies \eqref{eq:oracle-kkt} when the
cap is inactive and \eqref{eq:oracle-kkt-active} when it is active. The chart
and active-coordinate parameterization impose the rank and support
restrictions. No global minimum of the nonconvex criterion is required.

\begin{assumption}
\label{ass:oracle-functional}
For the true support $\A$ and fixed completion rank $r$, the following local
conditions hold. The projected-rank set admits the identifiable twice
continuously differentiable chart $R(\vartheta)$ above; the true coordinate
$\eta_{\A,0}$ is an interior point of that chart; and, for some fixed
$\delta_{\mathrm{stab}}>0$,
\[
  \varrho\{(\gamma_0+\beta_0)I_N+R_0+S_0\}
  \le 1-\delta_{\mathrm{stab}}<\bar\varrho .
\]
The measurable initialization satisfies
$\bar\eta_{\A}\to_p\eta_{\A,0}$, and the local-solution rule
$\mathcal S_T(\A,\bar\eta_{\A})$ selects, with probability tending to one, a
constrained-KKT point $\widetilde\eta_{\A}$ in the local chart with
$\|\widetilde\eta_{\A}-\eta_{\A,0}\|_2=o_p(1)$. No global optimality
condition is imposed.
Let $J_{\A,t}=\partial q_t(\eta_{\A,0})/\partial\eta_{\A}^\top$. The oracle population criterion is locally identified and has nonsingular Hessian
\begin{equation}
  H_{\A}=2\E(J_{\A,t}^\top J_{\A,t})\succ0.
  \label{eq:oracle-hessian}
\end{equation}
At this local KKT solution, the feasible oracle score has the same first-order expansion as its latent-log-IV counterpart, and the relevant score array satisfies the central limit theorem and stochastic equicontinuity needed for the local quadratic expansion of \eqref{eq:oracle-refit}.

For each reported connectedness functional $g$, the local oracle refit admits
\begin{equation}
  \sqrt T\{g(\widetilde\eta_{\A})-g(\eta_{\A,0})\}
  =
  T^{-1/2}\sum_{t=1}^T\psi_{g,t}+o_p(1),
  \label{eq:oracle-linear}
\end{equation}
where $\{\psi_{g,t}\}$ is a martingale-difference or sufficiently mixing influence-function array satisfying a central limit theorem with long-run variance
\begin{equation}
  V_g
  =\lim_{T\to\infty}\Var\!\left(T^{-1/2}\sum_{t=1}^T\psi_{g,t}\right)
  =\sum_{\ell=-\infty}^{\infty}\Cov(\psi_{g,t},\psi_{g,t-\ell})
  \in(0,\infty),
  \label{eq:long-run-vg}
\end{equation}
where the covariance sum is absolutely summable. If the influence function is
a martingale difference, all nonzero-lag covariances vanish. For fixed $N$,
support, and rank, \eqref{eq:oracle-linear} follows from the local quadratic
expansion and the delta method. With a growing nuisance dimension, it is
maintained only for the low-dimensional projection defined by $g$.
\end{assumption}

\begin{theorem}
\label{thm:conn-inference}
Let $\widehat\A=\supp(\widehat S^{(2)})$, where $\widehat S^{(2)}$ is the second-step estimator \eqref{eq:two-step}. Define $\widetilde\eta_{\widehat\A}$ by applying the fixed measurable local-solution rule $\mathcal S_T$ to the criterion and constraints in \eqref{eq:oracle-refit}, initialized at $\bar\eta_{\widehat\A}$. If Theorem~\ref{thm:oracle} and Assumption~\ref{ass:oracle-functional} hold and $g(\eta)$ is locally continuously differentiable at $\eta_{\A,0}$ with denominator bounded away from zero, then
\begin{equation}
  \sqrt T\{g(\widetilde\eta_{\widehat\A})-g(\eta_{\A,0})\}
  \xrightarrow{d}
  \mathcal N(0,V_g).
\end{equation}
A consistent influence-function variance estimator gives asymptotically valid confidence intervals: the ordinary plug-in sandwich is sufficient for a martingale-difference score, while a HAC estimator is used under mixing dependence.
For connectedness measures based on absolute response masses, local
differentiability requires every response entry inside an absolute value to be
either bounded away from zero or locally fixed at zero by a maintained
restriction and excluded from the active derivative.
\end{theorem}

\begin{remark}
\label{rem:oracle-scope}
The result is pointwise and relies on
$\p(\widehat\A=\A)\to1$. It is not uniform over networks with spillover
coefficients approaching zero \citep{leeb2005}. The confidence intervals
therefore require detectable true edges. In finite samples, unconditional
selected-estimator coverage, coverage conditional on target-edge selection,
and fixed-support oracle coverage are distinct quantities.
\end{remark}

\begin{corollary}
\label{cor:net-tests}
For each asset $i$, the null hypothesis $H_0:C_i^{\mathrm{net}}=0$ can be tested by a Wald statistic based on Theorem~\ref{thm:conn-inference}. Simultaneous bands for the vector of node-level connectedness measures can be constructed using the corresponding multivariate delta-method covariance, while the scalar $C^{\mathrm{total}}$ has an ordinary pointwise confidence interval.
\end{corollary}

\section{Simulations}

\input{simulation_section}

\section{Empirical Analysis}
\label{sec:empirical}

\subsection{Data and volatility measurement}

The application uses split-adjusted intraday bars for the nine original Select
Sector SPDR ETFs: XLB, XLE, XLF, XLI, XLK, XLP, XLU, XLV, and XLY. The
balanced one-minute panel runs from July 26, 2007 through December 2025.
We exclude XLRE because its 2015 inception would remove the global financial
crisis from a balanced panel. SPY is used only in the common-day data-quality
filter. Including a broad index in the sector network could turn mechanical
portfolio overlap into estimated transmission. We exclude DIA, QQQ, and IWM
for the same reason.

The FirstRate Data files contain one-minute open, high, low, close, and volume
bars. We use the regular-session grid from 9:30 a.m. to 3:59 p.m. New York
time, which gives $m=390$ returns on a full trading day. For duplicate
timestamps, we retain the last observation. A day is kept when every sector ETF
and SPY has at least 371 observed minutes.
Previous-tick interpolation fills isolated missing prices; early closes and
days with more missing observations are removed. Of 4,780 source dates, 4,182
remain in the one-minute panel and 4,626 remain in the separately constructed
five-minute panel. Missing one-minute dates are concentrated early in the
sample: the GFC subsample contains 344 one-minute days and 585 five-minute
days. The Online Appendix gives annual retention rates. Overnight returns are
excluded because the target is regular-session continuous variation.

Let $r_{i,t,k}$ be the $k$th one-minute log return and define bipower variation by
\[
  BV_{i,t}=\frac{\pi}{2}\sum_{k=2}^{m}
  |r_{i,t,k}|\,|r_{i,t,k-1}|.
\]
To limit finite-activity jump contamination, we retain
\[
  \widetilde r_{i,t,k}
  =
  r_{i,t,k}\,
  \mathbf 1\!\left\{
    |r_{i,t,k}|
    \le 4\sqrt{BV_{i,t}/m}
  \right\},
\]
following the threshold principle of \citet{mancini2009}. The baseline daily proxy applies the finite-sample-corrected two-scale realized-variance estimator of \citet{zhang2005} to the truncated returns, using 20 subsampling grids. We denote it by $\widehat{IV}^{\,TS}_{i,t}$ and use
\[
  \widehat z_{i,t}
  =
  \log\!\left(\max\{\widehat{IV}^{\,TS}_{i,t},10^{-14}\}\right)
\]
in estimation. This construction separates the empirical roles of jump truncation and noise correction rather than treating raw one-minute realized variance as noise-free.
The daily threshold is not rescaled by a separate intraday periodicity estimate.

We use two alternative high-frequency measures. The first is threshold
realized variance, $\sum_k\widetilde r_{i,t,k}^2$. The second is a Parzen
realized kernel computed from the truncated returns
\citep{barndorff2008}, with fixed bandwidth 20. Before taking logs, a negative
or near-zero kernel estimate is set to two percent of threshold realized
variance; this floor does not bind on any retained sector-day. We also construct
a five-minute panel ($m=78$) and use threshold realized variance. A winsorized
open-to-close squared return provides the daily comparison in
Proposition~\ref{prop:nonequiv}; its thresholds are estimated from the training
sample.

The TSRV and realized-kernel estimates are noise-robust empirical
implementations; their in-fill rates differ from the $m^{-1/2}$ proxy used in
Theorem~\ref{thm:error}.

\subsection{Estimation and evaluation design}

The empirical implementation follows the estimator used in the Monte Carlo
study. In each estimation
window, validation QLIKE selects the projected-nuclear and sparse penalties,
and a row-wise ranked GIC with multiplier 0.25 and degree cap three selects the
forecast support. The active sparse coefficients and dynamic parameters are
then refitted subject to stability, conditional on the penalized projected
component. The GIC is a finite-sample forecasting device; inference instead
uses the weighted-LASSO selector covered by Theorem~\ref{thm:oracle}.

Connectedness inference instead uses the support selected by the second-step
weighted LASSO in \eqref{eq:two-step} and fixes the completion rank at one
before bootstrap estimation. A local constrained-KKT refit jointly estimates
the dynamic parameters, active sparse coefficients, and rank-constrained
completion. The resulting intervals are conditional on the selected support
and fixed rank, as in Theorem~\ref{thm:conn-inference}.

All fitted recursions satisfy \eqref{eq:stability}. Level forecasts use a
smearing correction estimated from training residuals. No full-sample
transformation or tuning parameter enters an out-of-sample forecast.

Forecast evaluation uses annual expanding windows. The first estimation window
ends in December 2016 and produces one-day-ahead forecasts for 2017. The
window then expands one year at a time through 2025. Predictor-factor rank,
penalty, and support selection, as well as refitting and the smearing
correction, use only observations available before each forecast year. The
primary loss is QLIKE,
\[
  \ell_{i,t}(\widehat h_{i,t})
  =
  \frac{\widehat{IV}_{i,t}}{\widehat h_{i,t}}
  -\log\!\left(
    \frac{\widehat{IV}_{i,t}}{\widehat h_{i,t}}
  \right)-1,
\]
whose ranking is robust to an imperfect volatility proxy under the conditions
in \citet{patton2011}. We also report log-scale mean squared error. Forecast
comparisons use average loss differences and HAC Diebold--Mariano statistics
\citep{diebold1995}. We use Bartlett--Newey--West standard errors with bandwidth
$\lfloor T_{\mathrm{eval}}^{1/3}\rfloor$, which equals 13 in the baseline
evaluation. A positive QLIKE gain favors the proposed model.

Deletion diagnostics set $R$, $S$, or all cross-volatility coefficients to
zero while holding the remaining full-model estimates fixed. The benchmarks
are a diagonal log-realized-volatility autoregression, HAR
\citep{corsi2009}, a stable unrestricted VAR in log realized volatility, and
the network estimator fitted to the daily squared-return proxy. The deletion
diagnostics locate predictive content within the fitted model; they are not
separately tuned estimators.

We interpret $M_{ij}$ as transmission from sector $j$ to sector $i$ and
compute the measures in \eqref{eq:connectedness-functionals}. Annual point
estimates use the expanding-window fits. Full-sample inference uses a 10-day
circular block bootstrap \citep{politis1992} of the joint fixed-support,
fixed-rank refit. Each draw re-estimates $(\gamma,\beta,R,S)$ under the same
stability and stationarity conditions as the full-sample estimator. Results
are reported only if at least 95 percent of the requested draws pass this
common screen. The bootstrap is used as a finite-sample approximation to the
conditional long-run law in Theorem~\ref{thm:conn-inference}; standard block
bootstrap validity uses a block length that diverges but is $o(T)$, and the
10-day finite-sample choice retains two trading weeks of short-run dependence.
Edge stability is evaluated separately with
one-minute TSRV, one-minute realized kernel, and five-minute threshold RV.

The regime analysis covers the global financial crisis (July 2007--2009),
the pre-pandemic period (2013--2019), COVID (2020--2021), monetary tightening
(2022--2023), and 2024--2025. Because early one-minute coverage is limited,
we report each regime beside its five-minute counterpart and treat the
decomposition as descriptive.

\subsection{Empirical results}
\label{sec:empirical-results}

The one-minute panel contains 4,182 trading days. Across sectors, the median
ratio of truncated two-scale RV to threshold RV ranges from 0.66 to 0.82, and
the correlation between their logs ranges from 0.91 to 0.96. The recursive
evaluation contains 2,239 one-day-ahead forecasts from 2017 through 2025.

Table~\ref{tab:emp-forecast} reports the forecast results. The proposed model
has mean QLIKE 0.1770, which is lower than that of the stable VAR, the
diagonal autoregression, and the daily-return network by 0.0165, 0.0244, and
0.1296, respectively, using unrounded losses; all three HAC tests reject equal
predictive accuracy at conventional levels. Relative to HAR, the reduction is
0.0026, or 1.5 percent of HAR loss, with a HAC statistic of 1.15
($p=0.2515$), and the proposed model has lower annual QLIKE in four of the nine
forecast years. The model has lower QLIKE than the multivariate and daily-proxy
benchmarks, while its advantage over HAR is positive but not statistically
significant.

\begin{table}[htbp]
\centering
\caption{Recursive one-day-ahead forecast comparison}
\label{tab:emp-forecast}
\small
\input{generated_empirical/table_forecast_main.tex}
\end{table}

The deletion diagnostics attribute most of the forecast improvement to the
projected factor component. Removing the selected sparse component changes
average QLIKE little in this panel.

Table~\ref{tab:emp-robustness} repeats the recursive exercise with a
one-minute Parzen realized kernel and a separately constructed five-minute
threshold-RV panel. Across all three proxies the gain over HAR is positive but
statistically insignificant, whereas the gain over the VAR is positive and
significant in each case. QLIKE levels are not compared across rows because
the evaluation proxy changes. We therefore retain one-minute truncated
two-scale RV as the baseline and use the five-minute results as a
sampling-frequency check.

\begin{table}[htbp]
\centering
\caption{Sampling-frequency and volatility-proxy robustness}
\label{tab:emp-robustness}
\small
\resizebox{\textwidth}{!}{\input{generated_empirical/table_sampling_proxy_robustness.tex}}
\end{table}

The network estimates are less stable across volatility proxies than the
forecast rankings. No directed sparse edge satisfies the frequency and sign
criteria under all three proxies, and the top transmitter and receiver change
across specifications. We therefore do not interpret the point rankings as a
stable sector hierarchy.

Table~\ref{tab:emp-connectedness-inference} reports full-sample connectedness
inference. The second-step weighted LASSO selects an empty sparse support, so
the refit is conditional on this support and a prespecified completion rank of
one. All 499 requested circular block bootstrap draws pass the common stability
and constrained-KKT screen. Total connectedness is 0.9126, with a pointwise 95
percent interval of $[0.9058,0.9194]$. Energy has positive net connectedness
after simultaneous max-$t$ adjustment, whereas the simultaneous intervals for
the other sectors include zero. Thus the full-sample evidence concerns
systemwide propagation through the projected factor component, not individual
sparse transmission channels. The Online Appendix reports completion-rank and
response-horizon sensitivity.

\begin{table}[htbp]
\centering
\caption{Connectedness inference}
\label{tab:emp-connectedness-inference}
\footnotesize
\input{generated_empirical/table_connectedness_main.tex}
\end{table}
\FloatBarrier

Energy transmits 26.72 percent of normalized off-diagonal response mass and
receives 9.86 percent, giving a net position of 16.86 percentage points. This
pattern is consistent with energy-price volatility being transmitted to other
sectors through input-cost and inflation channels during the sample period. It
remains a model-implied log-volatility response, not an externally identified
causal shock.

Figure~\ref{fig:emp-network-regimes} summarizes the regime estimates. Total
connectedness is similar across the one- and five-minute panels, whereas the
net transmitter--receiver positions change across regimes. These differences
are descriptive because early one-minute coverage is limited and the node
estimates vary with the volatility proxy. Total connectedness is the
off-diagonal share of absolute cumulative responses, not a forecast-error
variance share.

\begin{figure}[htbp]
\centering
\begin{subfigure}{0.96\textwidth}
\centering
\includegraphics[width=\textwidth]{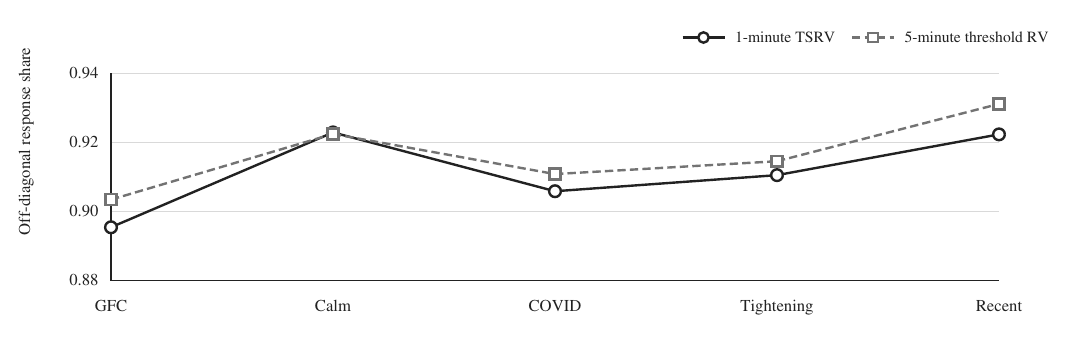}
\caption{Total connectedness.}
\label{fig:emp-network-total}
\end{subfigure}

\vspace{0.4em}

\begin{subfigure}{0.96\textwidth}
\centering
\includegraphics[width=\textwidth]{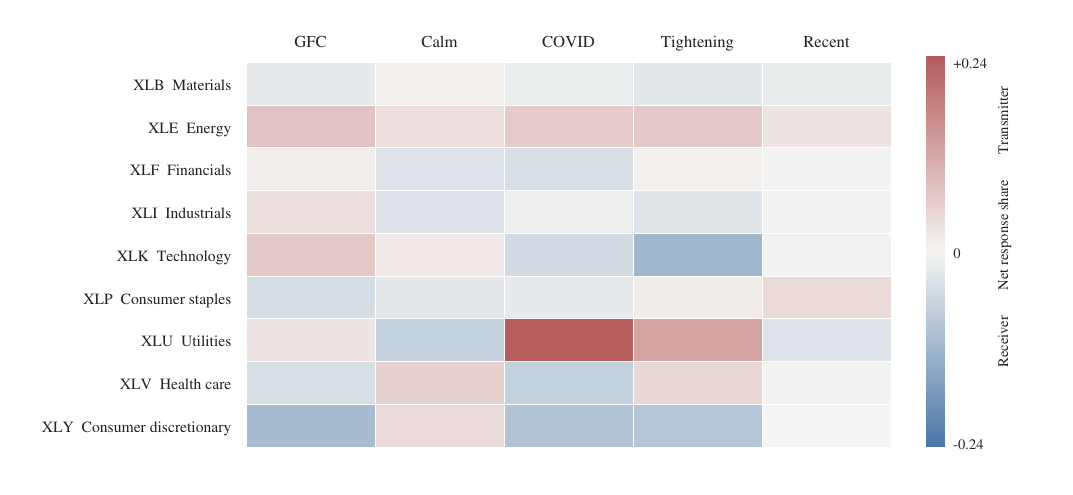}
\caption{Net connectedness.}
\label{fig:emp-network-net}
\end{subfigure}
\caption{Connectedness across market regimes}
\label{fig:emp-network-regimes}
\end{figure}
\FloatBarrier

The Online Appendix reports the factor--sparse decomposition and the selected
sparse channels. The composition is less stable than total connectedness:
component norms and selected edges change with the sampling rule.

\section{Conclusion}

We propose a continuous-time model for volatility transmission identified from
high-frequency observations. Its daily recursion describes an unknown directed
network in integrated volatility. For a fixed network template, feasible
log-volatility estimation is first-order equivalent to estimation based on
latent integrated volatility. For an unknown network, regularization selects
the relevant structure and a local refit provides conditional inference for
response-based connectedness.

In the sector-ETF application, the network has lower average out-of-sample
QLIKE than the VAR, diagonal, and daily-proxy benchmarks. Its average QLIKE is
also lower than HAR, although the HAC test does not reject equal predictive
accuracy. The full-sample weighted LASSO selects an empty sparse support;
conditional inference for the rank-one projected factor component identifies
Energy as a net transmitter. Forecast comparisons are stable across the main
high-frequency proxies, but sparse edges and annual node rankings are not. The
model therefore captures predictive cross-sectional volatility dependence,
while detailed network interpretation remains sensitive to volatility
measurement.

Several extensions are left for future work. Extending the recovery theory to
noisy-price volatility estimators requires incorporating their
estimator-specific in-fill rates. Relaxing the $N^2/m=o(1)$ requirement calls
for sharper mixed-norm bounds on the cross terms in the restricted-curvature
argument. Finally, a growing-$N$ central limit theory for the defactored score
would replace the conditional oracle transfer with unconditional inference on
the selected network.

\clearpage
\appendix

\section{Proof of the Exponential Embedding}
\label{app:embedding}

\begin{proof}
We give an explicit construction of a positive within-day variance process whose conditional integrated volatility satisfies the proposed daily network recursion. The construction proceeds in four steps: (1) the filtration, (2) the conditional-mean state, (3) the within-day variance factor, and (4) verification of the conditional-moment identity, positivity, and measurability. A final step records the properties of the resulting innovation.

\emph{Step 1: Filtration.} Let $(W^\sigma_{i,t}(u))_{u\in[0,1]}$,
$i=1,\ldots,N$, $t=1,2,\ldots$, be standard Brownian motions independent
across $(i,t)$ and independent of $B$, $J$, and $\F_0$. Enlarge the
probability space if needed. Let $\{\F_s\}_{s\ge0}$ be the augmented,
right-continuous filtration generated by $\F_0$, $B$, $J$, and
$\{W^\sigma_{i,\tau}(u):\tau-1+u\le s\}$. Then
$W^\sigma_{i,t}(\cdot)$ is independent of $\F_{t-1}$ and
$W^\sigma_{i,t}(u)$ is $\F_{t-1+u}$-measurable.

\emph{Step 2: Conditional-mean state.} Given the $\F_0$-measurable initial conditions $h_{i,0}>0$ and $IV_{i,0}>0$, define $y_t$ recursively by \eqref{eq:vector-recursion} and set $h_{i,t}=\exp(y_{i,t})$. By induction: $y_{t-1}$ and $z_{t-1}=\log IV_{t-1}$ are $\F_{t-1}$-measurable, so $y_t$ is $\F_{t-1}$-measurable, and since $\omega_i$ is finite and $M$ has finite row sums, $y_{i,t}$ is finite whenever $y_{t-1}$ and $z_{t-1}$ are. Hence $h_{i,t}=\exp(y_{i,t})$ is finite, strictly positive, and $\F_{t-1}$-measurable. Signed entries of $M$ create no positivity problem because the spillover terms enter before exponentiation; this is the reason for working on the log scale.

  \emph{Step 3: Within-day variance factor.} For $s\in[t-1,t)$, write
$u=s-(t-1)\in[0,1)$ and fix constants $\upsilon_i\ge0$. Define
\begin{equation}
  \mathcal E_{i,t}(u)
  =
  \exp\!\left\{\upsilon_i W^\sigma_{i,t}(u)
  -\frac12\upsilon_i^2u\right\},
  \qquad
  \sigma_{i,t-1+u}^2=h_{i,t}\,\mathcal E_{i,t}(u).
\end{equation}
Then $\mathcal E_{i,t}(u)>0$ and, as a Dol\'eans-Dade exponential of
$\upsilon_iW^\sigma_{i,t}$, the process
$u\mapsto\mathcal E_{i,t}(u)$ is a continuous positive martingale on $[0,1]$
with respect to $\{\F_{t-1+u}\}_u$, with
\begin{equation}
  \E\{\mathcal E_{i,t}(u)\mid\F_{t-1}\}=1,
  \qquad u\in[0,1].
  \label{eq:mart-norm}
\end{equation}
Within each open day, $\sigma_{i,\cdot}^2$ is therefore a continuous positive semimartingale (an $\F_{t-1}$-measurable constant times a geometric Brownian motion). At an integer $t$, define the process by its right-hand value, $\sigma_{i,t}^2=h_{i,t+1}\mathcal E_{i,t+1}(0)=h_{i,t+1}$; this value is $\F_t$-measurable, while the left limit is $h_{i,t}\mathcal E_{i,t}(1)$. Thus the half-open-day concatenation is right-continuous with left limits and adapted. On each finite horizon it has only finitely many integer-time jumps, so the concatenated process is a positive c\`adl\`ag semimartingale. Endpoint values do not affect the daily integrals below.

\emph{Step 4: Conditional-moment identity.} The daily integrated volatility generated by this construction is
\begin{equation}
  IV_{i,t}
  =
  \int_{t-1}^{t}\sigma_{i,s}^2\,ds
  =
  h_{i,t}\int_0^1 \mathcal E_{i,t}(u)\,du,
  \label{eq:iv-construct}
\end{equation}
where the integral is finite and strictly positive almost surely because $u\mapsto\mathcal E_{i,t}(u)$ is continuous and strictly positive on the compact interval $[0,1]$. Since $h_{i,t}$ is $\F_{t-1}$-measurable and $\mathcal E_{i,t}\ge0$, Tonelli's theorem and \eqref{eq:mart-norm} give
\begin{equation}
  \E(IV_{i,t}\mid\F_{t-1})
  =
  h_{i,t}\int_0^1\E\{\mathcal E_{i,t}(u)\mid\F_{t-1}\}\,du
  =
  h_{i,t}.
\end{equation}
Hence $h_{i,t}=\E(IV_{i,t}\mid\F_{t-1})=\exp(y_{i,t})$ and the recursive definition of $y_t$ is exactly \eqref{eq:log-recursion}. The realized log volatility $z_{i,t}=\log IV_{i,t}$ is well defined and $\F_t$-measurable by \eqref{eq:iv-construct}, so it can drive the next day's state, closing the induction. Combining this variance process with the price equation \eqref{eq:price} completes the continuous-time embedding.

\emph{Step 5: Properties of the innovation.} The construction yields
\begin{equation}
  u_{i,t}:=z_{i,t}-y_{i,t}=\log\!\int_0^1\mathcal E_{i,t}(u)\,du .
\end{equation}
Because $\log$ is concave, Jensen's inequality applied twice gives the two-sided bound
\[
  -\frac{\upsilon_i^2}{4}
  =\int_0^1\E\{\log \mathcal E_{i,t}(u)\}\,du
  \;\le\;
  \E(u_{i,t}\mid\F_{t-1})
  \;\le\;
  \log\E\!\left\{\int_0^1\mathcal E_{i,t}(u)\,du\;\Big|\;\F_{t-1}\right\}=0,
\]
where the lower bound uses
$\log\int_0^1\mathcal E\,du\ge\int_0^1\log\mathcal E\,du$ and
$\E\log\mathcal E_{i,t}(u)=-\upsilon_i^2u/2$. Thus
$c_{0,i}=\E(u_{i,t}\mid\F_{t-1})\in[-\upsilon_i^2/4,0]$ is finite. It is
nonrandom because $W^\sigma_{i,t}$ is independent of $\F_{t-1}$ and has the
same law on every day. It equals zero only when $\upsilon_i=0$, and
$\widetilde u_{i,t}=u_{i,t}-c_{0,i}$ is a martingale difference. In this
construction, $\{u_{i,t}\}_t$ is i.i.d.\ across days. Assumption~\ref{ass:dep}
imposes beta-mixing separately because spectral stability alone does not imply
the required Markov-chain regularity. More general within-day dynamics can
replace the exponential-martingale factor if they preserve
\eqref{eq:mart-norm} and the constant Jensen-gap condition
\[
  \E\!\left[\log\!\left\{\int_0^1\mathcal E_{i,t}(u)\,du\right\}\middle|\F_{t-1}\right]
  =c_{0,i},
\]
with nonrandom $c_{0,i}$; the first condition alone embeds $h_t$, but the second is also required for the shifted log recursion used in estimation.
\end{proof}

\section{Proofs for Estimation Results}
\label{app:estimation-proofs}

\subsection{Proof of Proposition~\ref{prop:staged-id}}

\begin{proof}
We prove the two identification steps separately. Throughout Stage-0,
$W^\star$ is fixed, deterministic, correctly normalized, and zero diagonal,
with $M_0=\rho_0W^\star$. Let
\[
  \theta=(\widetilde\omega^\top,\gamma,\beta,\rho)^\top
\]
belong to a parameter space $\Theta_0$ on which $|\gamma|<1$ and the induced recursion is stable. For each $\theta\in\Theta_0$, define the filtered conditional mean by
\begin{equation}
  q_t(\theta)
  =
  \widetilde\omega+\gamma q_{t-1}(\theta)+(\beta I_N+\rho W^\star)z_{t-1}.
  \label{eq:stage0-filter}
\end{equation}
Under $|\gamma|<1$, this filter has the causal representation
\begin{equation}
  q_t(\theta)
  =
  \sum_{k=0}^{\infty}\gamma^k\widetilde\omega
  +
  \sum_{k=0}^{\infty}\gamma^k(\beta I_N+\rho W^\star)z_{t-1-k},
  \label{eq:stage0-causal}
\end{equation}
so $q_t(\theta)$ is measurable with respect to the observed past $\sigma(z_{t-1},z_{t-2},\ldots)$. The series converges absolutely, almost surely and in $L^2$, because $\E\|z_t\|_2^2<\infty$ by Assumption~\ref{ass:dep}, the weights are geometric, and $W^\star$ is a fixed normalized matrix; in implementation the filter is started at a fixed initial value $q_0$, and the initialization error decays as $C\gamma^t$ uniformly over the compact stable parameter set, so it affects neither the population criterion below nor any of the asymptotic statements.

Let $\theta_0$ denote the true Stage-0 parameter and, consistently with the recentering in \eqref{eq:q-recursion}, write
\[
  z_t=q_t(\theta_0)+\widetilde u_t,
  \qquad
  \E(\widetilde u_t\mid\F_{t-1})=0.
\]
Define the population criterion
\[
  Q_0(\theta)=\E\|z_t-q_t(\theta)\|_2^2,
\]
which is finite for every $\theta\in\Theta_0$ by stationarity, $\E\|z_t\|_2^2<\infty$, and the square-summability of the filter weights. Because $q_t(\theta_0)-q_t(\theta)$ is $\F_{t-1}$-measurable, the martingale-difference property gives
\[
  \E\!\left[\widetilde u_t^\top\{q_t(\theta_0)-q_t(\theta)\}\right]=0.
\]
Therefore
\begin{equation}
  Q_0(\theta)-Q_0(\theta_0)
  =
  \E\|q_t(\theta)-q_t(\theta_0)\|_2^2
  \ge0.
\end{equation}
Thus any population minimizer must satisfy $q_t(\theta)=q_t(\theta_0)$ almost surely.

It remains to show that equality of conditional means implies equality of parameters. Let
\[
  \Delta_\omega=\widetilde\omega-\widetilde\omega_0,\quad
  \Delta_\gamma=\gamma-\gamma_0,\quad
  \Delta_\beta=\beta-\beta_0,\quad
  \Delta_\rho=\rho-\rho_0.
\]
Subtracting \eqref{eq:stage0-filter} at $\theta$ and $\theta_0$, and using $q_t(\theta)=q_t(\theta_0)$ and $q_{t-1}(\theta)=q_{t-1}(\theta_0)$ almost surely, gives
\begin{equation}
  \Delta_\omega
  +\Delta_\gamma q_{t-1}(\theta_0)
  +\Delta_\beta z_{t-1}
  +\Delta_\rho W^\star z_{t-1}
  =0
  \qquad\text{a.s.}
  \label{eq:stage0-rank}
\end{equation}
Equivalently, define the random block regressor
\[
  \mathcal Z^{(0)}_{t-1}
  =
  \begin{bmatrix}
    I_N & q_{t-1}(\theta_0) & z_{t-1} & W^\star z_{t-1}
  \end{bmatrix},
\]
where the last three blocks are $N$-vectors multiplying the scalar
coefficients $(\Delta_\gamma,\Delta_\beta,\Delta_\rho)$. The Stage-0 rank
condition is that
$\E\{(\mathcal Z^{(0)}_{t-1})^\top\mathcal Z^{(0)}_{t-1}\}$ is nonsingular on
$\mathbb R^N\times\mathbb R^3$, or, equivalently, if constants
$(a_0,b_1,b_2,b_3)\in\mathbb R^N\times\mathbb R^3$ satisfy
\[
  a_0+b_1q_{t-1}(\theta_0)+b_2z_{t-1}+b_3W^\star z_{t-1}=0
  \qquad\text{a.s.},
\]
then $a_0=0$ and $b_1=b_2=b_3=0$. Applying this condition to
\eqref{eq:stage0-rank} yields
\[
  \Delta_\omega=0,\qquad
  \Delta_\gamma=\Delta_\beta=\Delta_\rho=0.
\]
Hence $\theta_0$ is the unique minimizer of $Q_0$, and Stage-0 is identified.

For Stage 1, take $(\widetilde\omega,\gamma,\beta)$ as identified. The conditional mean equation can be written
\[
  q_t-\widetilde\omega-\gamma q_{t-1}-\beta z_{t-1}=Mz_{t-1}.
\]
Suppose two matrices $M_1$ and $M_2$ generate the same conditional mean. Then
\[
  (M_1-M_2)z_{t-1}=0
  \qquad\text{a.s.}
\]
Let $\Gamma_z=\E(z_{t-1}z_{t-1}^\top)$, which exists by Assumption~\ref{ass:dep}, and let $\mathcal M\subseteq\R^N$ denote the linear space of admissible rows of $M$ (for the unrestricted zero-diagonal parameterization, $\mathcal M_i=\{a\in\R^N:a_i=0\}$ for row $i$). Suppose $\Gamma_z$ is positive definite on $\mathcal M_i$ for every $i$, so $a^\top\Gamma_z a>0$ for all nonzero $a\in\mathcal M_i$. Then for each row $a_i^\top$ of $A=M_1-M_2$, which lies in $\mathcal M_i$,
\[
  0=\E(a_i^\top z_{t-1})^2=a_i^\top\Gamma_z a_i
\]
implies $a_i=0$. Thus $A=0$ and $M_1=M_2$, so the combined network matrix is identified. The zero-diagonal restriction assigns own feedback to $\beta I_N$ and cross feedback to $M$. Finally, once $M_0$ is identified, Theorem~\ref{thm:separation} gives local tangent--support identifiability and the unique regularization-selected convex optimum for $M_0=R_0+S_0$; it does not assert global algebraic uniqueness over all decompositions. This proves the staged identification claim.
\end{proof}

\subsection{Proof of Proposition~\ref{prop:nonequiv}}

\begin{proof}
It suffices to show that replacing $IV_t$ by the daily squared return creates a non-vanishing measurement-error component in the regressor. To distinguish this level-scale benchmark from the filtered log design in \eqref{eq:filtered-design}, let $v_t$ denote the latent integrated-volatility regressor and let the daily-return proxy be
\[
  v_t^{r}=v_t+\nu_t,
  \qquad
  \E(\nu_t\mid v_t)=0,\qquad
  \E(v_t\nu_t^\top)=0.
\]
Let the latent conditional mean contain the feedback term $K_0v_t$, where
$K_0=\beta I_N+M$. The pseudo-true daily coefficient is the population linear
projection coefficient
\[
  K^r
  =
  \E(K_0v_tv_t^{r\top})\{\E(v_t^rv_t^{r\top})\}^{-1}.
\]
Using $\E(v_t\nu_t^\top)=0$ gives
\[
  \E(K_0v_tv_t^{r\top})=K_0\Sigma_v,
  \qquad
  \E(v_t^rv_t^{r\top})=\Sigma_v+\Sigma_\nu,
\]
and therefore \eqref{eq:daily-attenuation}. Unless
$\Sigma_\nu=0$ or the attenuation matrix happens to act as the identity on the
row space of $K_0$, the pseudo-true feedback matrix differs from $K_0$.

The scalar attenuation formula follows by taking one representative feedback
coefficient. If the true conditional mean contains the term $b_0v_t$ and
$v_t^r=v_t+\nu_t$ with $\Var(\nu_t)=\sigma_\nu^2>0$, then
\[
  b^{r}
  =
  \frac{\Cov(v_t^r,b_0v_t)}{\Var(v_t^r)}
  =
  b_0\,\frac{\Var(v_t)}{\Var(v_t)+\Var(\nu_t)}
  =
  b_0\,\lambda,
\]
where
\[
  \lambda=
  \frac{\Var(IV)}{\Var(IV)+\sigma_\nu^2}<1.
\]
If $\Sigma_\nu=\kappa\Sigma_v$ for some $\kappa>0$, the multivariate
attenuation matrix reduces to the common factor $(1+\kappa)^{-1}$; otherwise
it is a non-scalar matrix distortion.

The error $\nu_t=r_t^2-IV_t$ remains a daily-frequency quantity. Under the
leverage-free benchmark, $r_t$ is conditionally Gaussian with variance
$IV_t$. Hence $\Var(\nu_t\mid\F^\sigma)=2IV_t^2$ and
$\sigma_\nu^2=\E(2IV_t^2)>0$. This proves part (i).

For part (ii), work conditionally on the volatility path $\F^\sigma$. The independence and no-leverage conditions give $r_{i,t}\mid\F^\sigma\sim\mathcal N(0,IV_{i,t})$, so $r_{i,t}^2/IV_{i,t}\mid\F^\sigma\sim\chi^2_1$ and \eqref{eq:log-proxy} holds with $\varepsilon_{i,t}=\log(r_{i,t}^2/IV_{i,t})$. The proxy error is independent of the volatility path and i.i.d.\ across days in this benchmark; contemporaneous cross-asset dependence may be retained in $\Var(\varepsilon_t)$. The moments follow from the Gamma representation $\chi^2_1=\mathrm{Gamma}(1/2,2)$: for $Q\sim\mathrm{Gamma}(k,s)$, $\E\log Q=\psi(k)+\log s$ and $\Var\log Q=\psi'(k)$, where $\psi$ is the digamma function. With $k=1/2$ and $s=2$, $\psi(1/2)=-\gamma_{\mathrm E}-2\log2$ gives $\E(\varepsilon_{i,t})=-\gamma_{\mathrm E}-\log2$, and $\psi'(1/2)=\pi^2/2$ gives the variance. The exponential left tail follows from $\p(\varepsilon_{i,t}<-a)=\p(\chi^2_1<e^{-a})\asymp e^{-a/2}$ as $a\to\infty$. The mean is absorbed by the intercept; the variance is fixed at the daily frequency.

For the projection algebra, the contaminated design
\[
  e_t=\sum_{k\ge0}\gamma^k\{\varepsilon_{t-1-k}-\E(\varepsilon_t)\}
\]
involves only lagged daily proxy errors. It is therefore independent of the date-$t$ objects $z_{i,t}$, $\widetilde u_{i,t}$, and $\varepsilon_{i,t}$, and of the design signal $x_t$. Hence, after demeaning,
\[
  \E(z^r_{i,t}x^{r\top}_t)
  =\E(z_{i,t}x_t^\top)
  =K_{0,i\cdot}\,\Gamma_x,
  \qquad
  \E(x^r_tx^{r\top}_t)=\Gamma_x+\Sigma_{\mathrm d},
\]
where the first equality uses $\E(\varepsilon_{i,t}x_t^\top)=0$, $\E(z_{i,t}e_t^\top)=0$, and the martingale-difference property $\E(\widetilde u_{i,t}x_t^\top)=0$, and the second uses $\E(x_te_t^\top)=0$. The geometric sum $\Var(e_t)=\sum_{k\ge0}\gamma^{2k}\Var(\varepsilon_t)=(1-\gamma^2)^{-1}\Var(\varepsilon_t)$ gives $\Sigma_{\mathrm d}$, and \eqref{eq:log-attenuation} follows from the population linear-projection formula. Because $\Sigma_{\mathrm d}$ is a fixed positive definite matrix that does not depend on $m$, $K^{r,\log}\neq K_0$ whenever $K_0\neq0$; in the isotropic special case stated in the proposition, the attenuation is the common factor $\lambda_{\log}<1$.

By contrast, a high-frequency estimator satisfies $\widehat z_{i,t}=z_{i,t}+O_p(m^{-1/2})$ on the localized event of Lemma~\ref{lem:logiv-filter}, so the analogous contamination covariance is $O(m^{-1})$ and the feasible realized-volatility criterion has the same probability limit as the latent-$IV$ criterion as $m\to\infty$. Hence the daily-proxy linear projection and the high-frequency integrated-volatility projection are not asymptotically equivalent in this benchmark, on either scale.
\end{proof}

\begin{lemma}
\label{lem:logiv-filter}
Suppose Assumptions~\ref{ass:price}, \ref{ass:dep}, \ref{ass:rates}, and \ref{ass:orth} hold for the fixed-dimensional staged recursion. Then the localization event
\[
  \mathcal L_{N,T}
  =\left\{\min_{i,t}IV_{i,t}>2c_{N,T},
  \max_{i,t}|\widehat{IV}_{i,t}/IV_{i,t}-1|\le1/2\right\}
\]
has probability tending to one. On $\mathcal L_{N,T}$ trimming is inactive and
\[
  \widehat z_t=z_t+\zeta_t,\qquad
  \zeta_{i,t}
  =
  b^z_{i,t}+\mathring\zeta_{i,t}+r^z_{i,t},
\]
where $\mathring\zeta_{i,t}$ is the conditionally centered leading discretization term,
$\max_{i,t}|b^z_{i,t}|=O_p(m^{-1})$ in the no-noise case, with the corresponding pre-averaged bias rate under microstructure noise, and
$T^{-1}\sum_{t=1}^T\|r^z_t\|=O_p(m^{-1})$ in the no-noise case. Moreover, uniformly over the compact stable parameter set,
\[
  \widehat q_t(\theta)-q_t(\theta)
  =
  \sum_{k=1}^{t-1}G_{t,k}(\theta)\zeta_{t-k}+r_{q,t}(\theta),
  \qquad
  \|G_{t,k}(\theta)\|\le C\bar q^k,\quad \bar q<1,
\]
and
\[
  \sup_{\theta}T^{-1}\sum_{t=1}^T\|r_{q,t}(\theta)\|
  \le
  C\,T^{-1}\sum_{t=1}^T\|\zeta_t\|^2+o_p(T^{-1/2}).
\]
Consequently the feasible score differs from the infeasible score by a bias
term of order $O_p(\sqrt T/m)$ without noise,
$O_p(\sqrt T/\sqrt m)$ under noise, plus a centered stochastic term controlled
by Assumption~\ref{ass:orth} and a quadratic remainder of the same order.
\end{lemma}

\begin{proof}
\emph{Step 1: Localization and log expansion.} By the negative-moment condition and a union bound,
\[
  \p\!\left(\min_{i,t}IV_{i,t}\le2c_{N,T}\right)
  \le (2c_{N,T})^p\sum_{i,t}\E(IV_{i,t}^{-p})=o(1).
\]
The conditional relative-error tail in Assumption~\ref{ass:price}, together with the bias order, gives
\[
  \p\!\left(\max_{i,t}|\delta_{i,t}|>1/2\right)
  \le 2NT\exp(-cm)+o(1)=o(1),
\]
where the last equality uses $\log(NT)/m=o(1)$. This is the no-noise bound; the fixed-dimensional noisy case is analogous under its stated pre-averaging rate and in-fill condition. Hence $\p(\mathcal L_{N,T})\to1$. On this event $\widehat{IV}_{i,t}>c_{N,T}$, so trimming is inactive. Since $|\log(1+x)-x|\le x^2$ for $|x|\le1/2$,
\[
  \zeta_{i,t}
  =\log\widehat{IV}_{i,t}-\log IV_{i,t}
  =\delta_{i,t}+\rho_{i,t},
  \qquad
  |\rho_{i,t}|\le\delta_{i,t}^2 .
\]
Assumption~\ref{ass:price} imposes an estimator-specific relative-error
expansion. Its rates are motivated by stable-limit theory for jump-robust
estimators \citep{jacod2012,mancini2009} and noise-robust estimators
\citep{jacod2009,barndorff2008}; joint jump-noise settings require the
corresponding estimator-specific modification. Under the maintained
expansion, the relative error decomposes as
\[
  \delta_{i,t}
  =b_{i,t}/m+\eta^{IV}_{i,t},
\]
where the conditional bias is uniformly $O_p(m^{-1})$ and
$\eta^{IV}_{i,t}$ is conditionally centered with $O_p(m^{-1/2})$ scale
(respectively $O_p(m^{-1/4})$ under noise, in which case all subsequent
$m$-rates are adjusted accordingly). Setting
$b^z_{i,t}=b_{i,t}/m$, $\mathring\zeta_{i,t}=\eta^{IV}_{i,t}$, and
$r^z_{i,t}=\rho_{i,t}$ gives the stated decomposition; the remainder obeys
$T^{-1}\sum_t\|r^z_t\|=O_p(m^{-1})$ in the fixed-dimensional no-noise case
because $T^{-1}\sum_t\delta_{i,t}^2=O_p(m^{-1})$ coordinatewise. The
high-dimensional versions used below follow from the conditional vector-tail
condition in Assumption~\ref{ass:orth}.

\emph{Step 2: Exact linear filter perturbation.} The conditional-mean filter is \emph{linear} in the lagged log-volatility inputs: for the staged parameterization,
\[
  q_t(\theta)=\widetilde\omega+\gamma q_{t-1}(\theta)+K(\theta)z_{t-1},
  \qquad K(\theta)=\beta I_N+\rho W^\star,
\]
and the feasible filter $\widehat q_t(\theta)$ satisfies the same recursion with $z$ replaced by $\widehat z=z+\zeta$. Subtracting the two recursions eliminates $\widetilde\omega$ and yields the exact identity
\[
  \widehat q_t(\theta)-q_t(\theta)
  =\gamma\{\widehat q_{t-1}(\theta)-q_{t-1}(\theta)\}
  +K(\theta)\zeta_{t-1},
\]
which iterates to
\[
  \widehat q_t(\theta)-q_t(\theta)
  =\sum_{k=1}^{t-1}\gamma^{k-1}K(\theta)\,\zeta_{t-k}
  +\gamma^{t-1}\{\widehat q_1(\theta)-q_1(\theta)\}.
\]
Hence the filter weights are exactly $G_{t,k}(\theta)=\gamma^{k-1}K(\theta)$, and $\|G_{t,k}(\theta)\|_2\le C\bar q^k$ with $\bar q=\sup_{\Theta_f}|\gamma|<1$ holds uniformly over the compact stable parameter set because $K(\theta)$ is bounded on it; no expansion in the generated regressors is needed since the recursion is linear in them. The remainder $r_{q,t}(\theta)=\gamma^{t-1}\{\widehat q_1(\theta)-q_1(\theta)\}$ collects only the initialization difference; it is bounded by $C\bar q^{t-1}\|\zeta_0\|$ (with $\zeta_0$ the plug-in error entering the common initialization), so
\[
  \sup_{\theta}T^{-1}\sum_{t=1}^T\|r_{q,t}(\theta)\|
  \le \frac{C}{T(1-\bar q)}\|\zeta_0\|
  =O_p(T^{-1}),
\]
which is smaller than the stated bound. The stated bound also covers the general network step, in which $K(\theta)=\beta I_N+M$: the identity above holds verbatim with this $K(\theta)$, and boundedness on the compact stable set again gives the geometric weights. The quadratic term $C\,T^{-1}\sum_t\|\zeta_t\|^2$ quoted in the lemma arises not from the filter but from the log expansion in Step 1 (the $r^z_t$ component of $\zeta_t$) and from products of two $\zeta$'s in the score expansion of Step 3; stating the bound with this term keeps the lemma applicable to both uses.

\emph{Step 3: Score decomposition.} Substituting $\widehat z_t=z_t+\zeta_t$ and the exact filter identity into the feasible score
$\nabla\widehat Q_T(\theta)=-\tfrac2T\sum_t\nabla\widehat q_t(\theta)^\top\{\widehat z_t-\widehat q_t(\theta)\}$
produces three groups of terms: (a) the infeasible score; (b) terms linear in
$\{\zeta_t\}$ with adapted, geometrically weighted coefficients; and (c)
terms quadratic in $\{\zeta_t\}$. Assumption~\ref{ass:orth} controls the
centered linear terms. The bias terms are
$O_p(\sqrt T/m)$ without noise and $O_p(\sqrt T/\sqrt m)$ under noise. The
quadratic terms have the same orders because they are bounded by
$C\sqrt T\,T^{-1}\sum_t\|\zeta_t\|^2$. The difference between
$\nabla\widehat q_t$ and $\nabla q_t$ is linear in lagged $\zeta$ and is
included in these two groups.
\end{proof}

\subsection{Proof of Theorem~\ref{thm:asym-normal}}

\begin{proof}
Let $\widehat Q_T$ be the feasible criterion and $Q_T$ the infeasible
criterion using $z_t$, both restricted to the fixed-dimensional staged
parameter $\theta_f$. Write $\Theta_f$ for the compact stable parameter set.
Proposition~\ref{prop:staged-id} gives a unique population minimizer
$\theta_{f,0}$. The map $\theta_f\mapsto\ell_t(\theta_f)$ is continuous
(indeed smooth, by the causal representation \eqref{eq:stage0-causal}, whose
geometric weights are smooth in $\theta_f$), and it is dominated by
\[
  C\left(1+\sum_{k\ge0}\bar q^k\|z_{t-1-k}\|_2\right)^2.
\]
This envelope has finite expectation under Assumption~\ref{ass:dep} by
Minkowski's inequality and geometric summability. Strict stationarity and
ergodicity give the pointwise law of large numbers, and compactness of
$\Theta_f$ plus dominated continuity upgrade it to the uniform law
\[
  \sup_{\theta_f\in\Theta_f}|Q_T(\theta_f)-Q(\theta_f)|=o_p(1),
  \qquad
  Q(\theta_f)=\E\ell_t(\theta_f),
\]
by the standard argument for extremum estimators \citep[Lemma~2.4 and Theorem~2.1]{newey1994}.
On the localized event of Lemma~\ref{lem:logiv-filter}, $\widehat z_t=z_t+\zeta_t$ and the feasible filter is a geometrically weighted perturbation of the infeasible filter. The mean-value theorem and the summability of the filter derivatives give
\[
  \sup_{\theta_f\in\Theta_f}|\widehat Q_T(\theta_f)-Q_T(\theta_f)|
  \le
  C\,T^{-1}\sum_{t=1}^T(1+\|z_t\|+\|\widehat z_t\|)\|\zeta_t\|
  +
  C\,T^{-1}\sum_{t=1}^T\|\zeta_t\|^2
  =o_p(1).
\]
Hence $\widehat\theta_f\to_p\theta_{f,0}$.

For asymptotic normality, the feasible first-order condition and a mean-value expansion around $\theta_{f,0}$ give
\begin{equation}
  \sqrt T(\widehat\theta_f-\theta_{f,0})
  =-\widehat{\mathcal H}_T(\bar\theta)^{-1}
  \sqrt T\nabla\widehat Q_T(\theta_{f,0}),
\end{equation}
where $\bar\theta$ lies between $\widehat\theta_f$ and $\theta_{f,0}$. The same ergodic LLN, applied to the Hessian class in a neighborhood of $\theta_{f,0}$, gives
\[
  \widehat{\mathcal H}_T(\bar\theta)\xrightarrow{p}\mathcal H.
\]
For the score, Lemma~\ref{lem:logiv-filter} gives the feasible residual expansion
\[
  \widehat{\widetilde u}_t(\theta_{f,0})
  =
  \widehat z_t-\widehat q_t(\theta_{f,0})
  =
  \widetilde u_t+\zeta_t-\sum_{k\ge1}G_{t,k}(\theta_{f,0})\zeta_{t-k}+r^u_t,
\]
where $\|G_{t,k}(\theta_{f,0})\|\le C\bar q^k$ and $T^{-1}\sum_t\|r^u_t\|=O_p(T^{-1}\sum_t\|\zeta_t\|^2)+o_p(T^{-1/2})$. Therefore the feasible score admits the expansion
\[
  \sqrt T\,\nabla\widehat Q_T(\theta_{f,0})
  =
  \sqrt T\,\nabla Q_T(\theta_{f,0})
  -
  \frac{2}{\sqrt T}\sum_{t=1}^T
  J_t^\top\left\{\zeta_t-\sum_{k\ge1}G_{t,k}(\theta_{f,0})\zeta_{t-k}\right\}
  +
  R_T,
\]
where $J_t=\nabla q_t(\theta_{f,0})\in\R^{N\times p_f}$ and
\[
  \|R_T\|
  \le
  C\sqrt T\,T^{-1}\sum_{t=1}^T
  \{\|\zeta_t\|^2+\|\zeta_{t-1}\|^2\}
  =
  \begin{cases}
    O_p(\sqrt T/m), & \text{without noise},\\
    O_p(\sqrt T/\sqrt m), & \text{under noise}.
  \end{cases}
\]
Decompose $\zeta_t=b^z_t+\mathring\zeta_t+r^z_t$ into its conditional bias, centered discretization error, and Taylor remainder. For the centered part, collect the coefficient of $\mathring\zeta_s$ across the contemporaneous and lagged terms in the $N\times p_f$ matrix
\[
  H_s
  :=
  J_s-\sum_{\substack{k\ge1\\s+k\le T}}
  G_{s+k,k}(\theta_{f,0})^\top J_{s+k}.
\]
This expression distinguishes the filter weights
$G_{t,k}\in\R^{N\times N}$ from the Jacobians
$J_t\in\R^{N\times p_f}$. Each column $H_{s,\ell}$ is a
$\mathcal G$-measurable vector weight, and, because $p_f$ is fixed,
$T^{-1}\sum_s\|H_s\|_F^2=O_p(1)$ by stationarity,
$\E\|J_t\|_F^2<\infty$, and the geometric decay
$\|G_{t,k}\|_2\le C\bar q^k$. Applying
Assumption~\ref{ass:orth}(i)--(ii) and its conditional-Chebyshev calculation
column by column therefore gives
\[
  T^{-1/2}\sum_{t=1}^T
  J_t^\top\left\{\mathring\zeta_t-\sum_{k\ge1}G_{t,k}(\theta_{f,0})\mathring\zeta_{t-k}\right\}
  =
  \begin{cases}
    O_p(m^{-1/2}), & \text{without noise},\\
    O_p(m^{-1/4}), & \text{under noise},
  \end{cases}
  =o_p(1),
\]
while the bias and Taylor-remainder components are $O_p(\sqrt T/m)$ without noise and $O_p(\sqrt T/\sqrt m)$ under noise. The rate restrictions $T=o(m^2)$ and $T=o(m)$ make these terms and $R_T$ negligible. Hence
\[
  \sqrt T\,\nabla\widehat Q_T(\theta_{f,0})
  =
  \sqrt T\,\nabla Q_T(\theta_{f,0})+o_p(1).
\]
The infeasible first-order condition has the analogous expansion
\[
  \sqrt T(\widetilde\theta_f-\theta_{f,0})
  =
  -\mathcal H^{-1}\sqrt T\,\nabla Q_T(\theta_{f,0})+o_p(1).
\]
The feasible Hessian converges to the same matrix $\mathcal H$. Subtracting
the two first-order expansions and using the preceding score equivalence gives
\[
  \sqrt T(\widehat\theta_f-\widetilde\theta_f)=o_p(1),
\]
which establishes \eqref{eq:first-order-equivalence}.

At the true parameter,
\[
  \nabla\ell_t(\theta_{f,0})
  =
  -2\nabla q_t(\theta_{f,0})^\top \widetilde u_t,
\]
where $\nabla q_t(\theta_{f,0})$ is $\F_{t-1}$-measurable and $\widetilde u_t$ is a martingale difference. The score moment imposed in Assumption~\ref{ass:dep} gives square integrability and the Lindeberg condition. Hence the summands form a stationary, ergodic martingale-difference sequence with variance matrix $\mathcal V_0$. The martingale CLT \citep[Corollary~3.1]{hall1980}, applied through Cram\'er--Wold, gives
\[
  \sqrt T\,\nabla Q_T(\theta_{f,0})
  \xrightarrow{d}\mathcal N(0,\mathcal V_0).
\]
Combining the score expansion, Hessian convergence, and Slutsky's theorem yields
\[
  \sqrt T(\widehat\theta_f-\theta_{f,0})
  \xrightarrow{d}
  \mathcal N(0,\mathcal H^{-1}\mathcal V_0\mathcal H^{-1}).
\]
The sandwich estimator is consistent by the same uniform LLN applied to the sample Hessian and the sample outer product of scores evaluated at $\widehat\theta_f$; this fixed-dimensional score is a martingale difference, so no HAC correction is needed for Theorem~\ref{thm:asym-normal}.
\end{proof}

\section{Proofs for Network Results}
\label{app:network-proofs}

\subsection{Proof of Proposition~\ref{prop:vma}}

\begin{proof}
Let
\[
  K=\beta I_N+M,
  \qquad
  A=\gamma I_N+K=(\gamma+\beta)I_N+M.
\]
At the true parameter, $\widetilde u_t=z_t-q_t$ satisfies
$\E(\widetilde u_t\mid\F_{t-1})=0$. Substituting
$q_{t-1}=z_{t-1}-\widetilde u_{t-1}$ into
the recentered recursion \eqref{eq:q-recursion} gives the observed log-IV
representation
\begin{equation}
  z_t
  =
  \widetilde\omega+A z_{t-1}+\widetilde u_t-\gamma\widetilde u_{t-1}.
  \label{eq:arma-proof}
\end{equation}
Thus the log-volatility process is a stable vector ARMA(1,1) recursion with
innovation difference $w_t=\widetilde u_t-\gamma\widetilde u_{t-1}$. Under \eqref{eq:stability},
Gelfand's formula $\varrho(A)=\lim_k\|A^k\|_2^{1/k}$ implies that for any
$\bar q\in(\varrho(A),1)$ there is $C<\infty$ with $\|A^k\|_2\le C\bar q^k$ for all
$k\ge0$ \citep[Corollary~5.6.13]{horn2013}. The moment condition in the
proposition gives $\E\|\widetilde u_t\|_2<\infty$, so the series
\begin{equation}
  z_t
  =
  (I_N-A)^{-1}\widetilde\omega
  +
  \sum_{k=0}^{\infty}A^k(\widetilde u_{t-k}-\gamma\widetilde u_{t-k-1})
  \label{eq:vma-proof}
\end{equation}
converges absolutely almost surely and in $L^1$: indeed
\[
  \sum_k\E\|A^kw_{t-k}\|_2
  \le
  C(1+|\gamma|)\E\|\widetilde u_t\|_2\sum_k\bar q^k
  <\infty,
\]
so the
partial sums are almost surely Cauchy by monotone convergence applied to
$\sum_k\|A^kw_{t-k}\|_2$. Direct substitution shows \eqref{eq:vma-proof}
solves \eqref{eq:arma-proof}, and the solution is strictly stationary because
it is a fixed measurable function of the stationary sequence
$\{\widetilde u_{t-k}\}_{k\ge0}$. It is also the unique strictly stationary solution: if
$z_t'$ is another one, then $k$ backward substitutions of
\eqref{eq:arma-proof} give $z_t-z_t'=A^k(z_{t-k}-z_{t-k}')$, and since
$\|z_{t-k}-z_{t-k}'\|_2$ is tight (by stationarity) while $\|A^k\|_2\to0$
geometrically, $z_t=z_t'$ almost surely \citep[cf.][Section~3.1]{brockwell1991}.

For connectedness, consider a one-unit perturbation in the lagged cross-network
input entering through $Mz_{t-1}$ while holding own feedback fixed. The
  horizon-$h$ response is $A^hM$, so the cumulative response over the first
  $H$ horizons, $0,\ldots,H-1$, is
  \[
  \Phi_H=\sum_{h=0}^{H-1}A^hM.
\]
The geometric bound implies $\Phi_H\to (I_N-A)^{-1}M$, because the Neumann
series $\sum_{h\ge0}A^h$ is absolutely convergent. Hence the finite-horizon and
infinite-horizon response matrices are well defined. Connectedness measures
formed from off-diagonal entries or shares of these response matrices therefore
exist whenever their normalizing denominators are nonzero.
\end{proof}

\subsection{Primitive sufficient conditions for projected-norm regularity}
\label{sec:projected-primitives}

The regularity properties collected in Assumption~\ref{ass:network}(ii) are
high-level. The following lemma derives the closedness, injectivity, and
effective-rank parts from the diagonal leakage of the standard tangent space.
It also exhibits the canonical subgradient in the empirically relevant
\emph{hollow} case, in which the minimum completion itself has zero diagonal.
Decomposability, the compatible-subgradient property, and the spectral tangent
bound remain imposed in Assumption~\ref{ass:network}(ii).

\begin{lemma}
\label{lem:projected-primitives}
Let $L_0=U_0\Sigma_0V_0^\top$ be the minimum-nuclear-norm completion of $R_0$, let $\T_{\mathrm{std}}$ be its standard tangent space, and define the diagonal leakage
\[
  \delta_\T=\max\{\|\diag(H)\|_F:H\in\T_{\mathrm{std}},\ \|H\|_F=1\}.
\]
\begin{enumerate}[label=(\roman*)]
\item If $\delta_\T<1$, then $P_{\mathrm{off}}$ is injective on $\T_{\mathrm{std}}$, $\T=P_{\mathrm{off}}(\T_{\mathrm{std}})$ is a closed subspace, and every $H\in\T$ has a unique tangent completion $\widetilde H\in\T_{\mathrm{std}}$ with
\[
  \|\widetilde H\|_F\le(1-\delta_\T^2)^{-1/2}\|H\|_F .
\]
\item Under the same condition, the effective-rank bound holds with an explicit constant:
\[
  \Omega_*(H)\le\sqrt{2r}\,(1-\delta_\T^2)^{-1/2}\,\|H\|_F,
  \qquad H\in\T .
\]
\item (Hollow case.) If in addition $\diag(L_0)=0$ and $\diag(U_0V_0^\top)=0$, then $R_0=L_0$, $\Omega_*(R_0)=\|L_0\|_*$, and $E_0=U_0V_0^\top\in\T$ is a canonical projected subgradient: $\Omega_*^\circ(E_0)=1$, $\langle E_0,R_0\rangle=\Omega_*(R_0)$, and $\|E_0\|_\infty\le\xi(\T)$.
\end{enumerate}
\end{lemma}

\begin{proof}
(i) For $H\in\T_{\mathrm{std}}$, orthogonality of the diagonal and off-diagonal parts gives $\|P_{\mathrm{off}}(H)\|_F^2=\|H\|_F^2-\|\diag(H)\|_F^2\ge(1-\delta_\T^2)\|H\|_F^2$. Hence $P_{\mathrm{off}}(H)=0$ forces $H=0$, injectivity holds, and the displayed bound follows by inverting; $\T$ is closed as the image of a finite-dimensional subspace under a linear map.

(ii) Every element of $\T_{\mathrm{std}}$ has rank at most $2r$, so $\|\widetilde H\|_*\le\sqrt{2r}\|\widetilde H\|_F$. Since $\widetilde H$ is a feasible completion of $H$, $\Omega_*(H)\le\|\widetilde H\|_*\le\sqrt{2r}\|\widetilde H\|_F\le\sqrt{2r}(1-\delta_\T^2)^{-1/2}\|H\|_F$ by part (i). The theorems that invoke the effective-rank bound absorb the factor $(1-\delta_\T^2)^{-1/2}$ into their constants provided $\delta_\T$ is bounded away from one.

(iii) If $\diag(L_0)=0$, then $R_0=P_{\mathrm{off}}(L_0)=L_0$, and $\Omega_*(R_0)\le\|L_0\|_*$ with equality by minimality of the completion. The matrix $E_0=U_0V_0^\top$ lies in $\T_{\mathrm{std}}$ and is zero-diagonal by hypothesis, hence $E_0=P_{\mathrm{off}}(E_0)\in\T$. By Lemma~\ref{lem:projected-norm} and $\diag(E_0)=0$, $\Omega_*^\circ(E_0)=\|P_{\mathrm{off}}(E_0)\|_2=\|U_0V_0^\top\|_2=1$. The pairing satisfies
\[
  \langle E_0,R_0\rangle=\langle U_0V_0^\top,L_0\rangle=\tr(V_0U_0^\top U_0\Sigma_0V_0^\top)=\tr(\Sigma_0)=\|L_0\|_*=\Omega_*(R_0),
\]
so $E_0$ attains the defining equality of the subdifferential characterization in the proof of Lemma~\ref{lem:projected-norm}. Finally $\|E_0\|_\infty\le\xi(\T)$ holds by the definition of $\xi(\T)$ applied to $E_0\in\T$ with $\|E_0\|_2=1$.
\end{proof}

\subsection{Proof of Lemma~\ref{lem:projected-norm}}

\begin{proof}
The feasible set $\{L:P_{\mathrm{off}}(L)=R\}$ is a nonempty closed affine subspace, and the nuclear norm is convex, positively homogeneous, and coercive, so the infimum in \eqref{eq:projected-nuclear} is attained by Weierstrass' theorem; the quotient construction therefore defines a finite convex gauge on the off-diagonal linear space, and subadditivity and homogeneity are inherited from $\|\cdot\|_*$ because the feasible sets add: if $P_{\mathrm{off}}(L_1)=R_1$ and $P_{\mathrm{off}}(L_2)=R_2$, then $P_{\mathrm{off}}(L_1+L_2)=R_1+R_2$. Definiteness holds because $\Omega_*(R)=0$ forces the attaining completion to satisfy $\|L\|_*=0$, hence $L=0$ and $R=P_{\mathrm{off}}(L)=0$. Thus $\Omega_*$ is a norm.

For the dual norm, use the adjoint identity
\[
  \langle Z,P_{\mathrm{off}}(L)\rangle
  =
  \langle P_{\mathrm{off}}(Z),L\rangle .
\]
Then
\[
  \Omega_*^\circ(Z)
  =
  \sup_{\Omega_*(R)\le1}\langle Z,R\rangle
  =
  \sup_{\|L\|_*\le1}\langle P_{\mathrm{off}}(Z),L\rangle
  =
  \|P_{\mathrm{off}}(Z)\|_2,
\]
where the last equality is the usual nuclear/operator norm duality.

The subdifferential of a quotient norm is
also characterized by
\[
  \partial\Omega_*(R)
  =
  \{Z:\Omega_*^\circ(Z)\le1,\ \langle Z,R\rangle=\Omega_*(R)\}.
\]
The existence of the canonical exposed subgradient, decomposability relative
to $(\T,\T^\perp)$, the compatible-subgradient property, and the
effective-rank bound are imposed as projected-norm regularity conditions in
Assumption~\ref{ass:network}.
\end{proof}

\subsection{Proof of Lemma~\ref{lem:trans}}

\begin{proof}
Let $P_\T$ and $P_\Omega$ denote the orthogonal projections onto $\T$ and $\OmegaS$. For any $Z\in\OmegaS$, the definitions of $\xi(\T)$ and $\mu(\OmegaS)$ imply
\[
  \|P_\Omega P_\T Z\|_\infty
  \le \|P_\T Z\|_\infty
  \le \xi(\T)\|P_\T Z\|_2.
\]
By the spectral tangent-projection bound imposed in Assumption~\ref{ass:network}(ii), $\|P_\T Z\|_2\le 2\|Z\|_2$ (this is automatic for the unprojected rank-$r$ tangent space, since $P_{\T_{\mathrm{std}}}(Z)=P_{U_0}Z+(I-P_{U_0})ZP_{V_0}$ is a sum of two spectral contractions). By the definition of $\mu(\OmegaS)$, $\|Z\|_2\le \mu(\OmegaS)\|Z\|_\infty$ for $Z\in\OmegaS$. Hence
\[
  \|P_\Omega P_\T Z\|_\infty
  \le 2\xi(\T)\mu(\OmegaS)\|Z\|_\infty.
\]
If $2\xi(\T)\mu(\OmegaS)<1$, then $P_\Omega P_\T$ is a contraction on $(\OmegaS,\|\cdot\|_\infty)$. Therefore $I-P_\Omega P_\T$ is invertible on $\OmegaS$ with inverse
\[
  (I-P_\Omega P_\T)^{-1}=\sum_{k=0}^{\infty}(P_\Omega P_\T)^k.
\]
Finally, if $Z\in\T\cap\OmegaS$, then $Z=P_\Omega P_\T Z$, and the contraction inequality gives
\[
  \|Z\|_\infty\le 2\xi(\T)\mu(\OmegaS)\|Z\|_\infty.
\]
Since the contraction constant is strictly less than one, $Z=0$. Thus $\T\cap\OmegaS=\{0\}$.
\end{proof}

\subsection{Proof of Theorem~\ref{thm:separation}}

\begin{proof}
Lemma~\ref{lem:trans} gives $\T\cap\OmegaS=\{0\}$. It remains to verify exact recovery by the noiseless convex program in the identifiable variables $(R,S)$. Let $\lambda_R>0$ and write $\gamma_\lambda=\lambda_S/\lambda_R$. By homogeneity set $\lambda_R=1$ in the construction below.

Let $E_0\in\partial\Omega_*(R_0)\cap\T$ be the canonical projected subgradient from Assumption~\ref{ass:network}, normalized so that $\Omega_*^\circ(E_0)=1$ and $\|E_0\|_\infty\le\xi(\T)$. We seek a dual certificate $Q=A+B$ with $A\in\T$ and $B\in\OmegaS$ satisfying
\begin{equation}
  P_\T Q=\lambda_R E_0,\quad
  \Omega_*^\circ(P_{\T^\perp}Q)<\lambda_R,\quad
  P_\Omega Q=\lambda_S\operatorname{sign}(S_0),\quad
  \|P_{\Omega^c}Q\|_\infty<\lambda_S.
  \label{eq:certificate-conditions}
\end{equation}
Since $A\in\T$ and $B\in\OmegaS$, we have $P_\T Q=A+P_\T B$ and $P_\Omega Q=P_\Omega A+B$, so the two equality constraints are equivalent to
\[
  A+P_\T B=E_0,
  \qquad
  P_\Omega A+B=\gamma_\lambda\operatorname{sign}(S_0).
\]
Eliminating $A=E_0-P_\T B$ gives
\[
  (I-P_\Omega P_\T)B
  =
  \gamma_\lambda\operatorname{sign}(S_0)-P_\Omega E_0
  =:D\in\OmegaS .
\]
By Lemma~\ref{lem:trans}, $I-P_\Omega P_\T$ is invertible on $(\OmegaS,\|\cdot\|_\infty)$ with Neumann-series inverse. Since
\[
  \|P_\Omega P_\T Z\|_\infty\le2\xi\mu\|Z\|_\infty,
  \qquad Z\in\OmegaS,
\]
\[
  \|B\|_\infty
  \le
  \sum_{k\ge0}(2\xi\mu)^k\|D\|_\infty
  \le
  \frac{\gamma_\lambda+\xi}{1-2\xi\mu},
\]
where $\xi=\xi(\T)$, $\mu=\mu(\OmegaS)$, and $\|D\|_\infty\le\gamma_\lambda\|\operatorname{sign}(S_0)\|_\infty+\|E_0\|_\infty\le\gamma_\lambda+\xi$.

For the first dual-feasibility inequality, note $P_{\T^\perp}Q=P_{\T^\perp}(A+B)=P_{\T^\perp}B=B-P_\T B$. Both $B\in\OmegaS$ and $P_\T B\in\T$ are zero-diagonal, so $\Omega_*^\circ(P_{\T^\perp}B)=\|P_{\mathrm{off}}(B-P_\T B)\|_2=\|B-P_\T B\|_2$ by Lemma~\ref{lem:projected-norm}. Using the spectral tangent-projection bound of Assumption~\ref{ass:network}(ii) and $\|B\|_2\le\mu\|B\|_\infty$ for $B\in\OmegaS$,
\[
  \Omega_*^\circ(P_{\T^\perp}Q)
  \le\|B\|_2+\|P_\T B\|_2
  \le3\|B\|_2
  \le3\mu\|B\|_\infty
  \le
  \frac{3\mu(\gamma_\lambda+\xi)}{1-2\xi\mu},
\]
which is strictly smaller than one (recall the normalization $\lambda_R=1$) if and only if
\[
  \gamma_\lambda < \frac{1-5\xi\mu}{3\mu}.
\]
For the second, since $P_{\Omega^c}Q=P_{\Omega^c}A$ and $\|P_\T B\|_\infty\le\xi\|P_\T B\|_2\le2\xi\|B\|_2\le2\xi\mu\|B\|_\infty$,
\[
  \|P_{\Omega^c}Q\|_\infty
  \le
  \|E_0\|_\infty+\|P_\T B\|_\infty
  \le
  \xi+2\xi\mu\|B\|_\infty
  \le
  \frac{\xi(1+2\mu\gamma_\lambda)}{1-2\xi\mu},
\]
which is strictly smaller than $\gamma_\lambda$ if and only if $\xi+2\xi\mu\gamma_\lambda<\gamma_\lambda(1-2\xi\mu)$, that is,
\[
  \gamma_\lambda>\frac{\xi}{1-4\xi\mu}.
\]
Writing $x=\xi\mu$, the interval \eqref{eq:tuning-interval} is nonempty if and only if $3x/(1-4x)<1-5x$, that is, $20x^2-12x+1>0$ with $x<1/4$; the roots of the quadratic are $x=1/10$ and $x=1/2$, so the interval is nonempty exactly when $\xi\mu<1/10$, and a certificate exists for every $\gamma_\lambda$ in \eqref{eq:tuning-interval}.

Now consider any feasible perturbation $(R_0+\Delta,S_0-\Delta)$ with $\Delta\ne0$; feasibility forces $\Delta$ to be zero-diagonal. By decomposability of $\Omega_*$ relative to $(\T,\T^\perp)$ from Assumption~\ref{ass:network}(ii), there exists $F\in\T^\perp$ with $\Omega_*^\circ(F)\le1$, $\langle F,\Delta\rangle=\Omega_*(P_{\T^\perp}\Delta)$, and $E_0+F\in\partial\Omega_*(R_0)$; similarly there exists $G$ supported on $\Omega^c$ with $\|G\|_\infty\le1$, $\langle G,-\Delta\rangle=\|P_{\Omega^c}\Delta\|_1$, and $\operatorname{sign}(S_0)+G\in\partial\|S_0\|_1$. The subgradient inequalities then give
\begin{align*}
  \Omega_*(R_0+\Delta)
  &\ge
  \Omega_*(R_0)+\langle E_0,\Delta\rangle
  +\Omega_*(P_{\T^\perp}\Delta),\\
  \|S_0-\Delta\|_1
  &\ge
  \|S_0\|_1-\langle \operatorname{sign}(S_0),P_\Omega\Delta\rangle
  +\|P_{\Omega^c}\Delta\|_1 .
\end{align*}
Multiply the first by $\lambda_R$, the second by $\lambda_S$, and add. By the certificate equalities $\lambda_RE_0=P_\T Q$ and $\lambda_S\operatorname{sign}(S_0)=P_\Omega Q$,
\[
  \lambda_R\langle E_0,\Delta\rangle-\lambda_S\langle\operatorname{sign}(S_0),P_\Omega\Delta\rangle
  =\langle P_\T Q,\Delta\rangle-\langle P_\Omega Q,\Delta\rangle
  =-\langle P_{\T^\perp}Q,P_{\T^\perp}\Delta\rangle
   +\langle P_{\Omega^c}Q,P_{\Omega^c}\Delta\rangle,
\]
using $\langle P_\T Q,\Delta\rangle=\langle Q,\Delta\rangle-\langle P_{\T^\perp}Q,P_{\T^\perp}\Delta\rangle$ and $\langle P_\Omega Q,\Delta\rangle=\langle Q,\Delta\rangle-\langle P_{\Omega^c}Q,P_{\Omega^c}\Delta\rangle$. Bounding the two pairings by the dual-norm inequalities $|\langle P_{\T^\perp}Q,P_{\T^\perp}\Delta\rangle|\le\Omega_*^\circ(P_{\T^\perp}Q)\,\Omega_*(P_{\T^\perp}\Delta)$ and $|\langle P_{\Omega^c}Q,P_{\Omega^c}\Delta\rangle|\le\|P_{\Omega^c}Q\|_\infty\|P_{\Omega^c}\Delta\|_1$, the objective change satisfies
\begin{align*}
  &\lambda_R\{\Omega_*(R_0+\Delta)-\Omega_*(R_0)\}
  +\lambda_S\{\|S_0-\Delta\|_1-\|S_0\|_1\}\\
  &\qquad\ge
  \{\lambda_R-\Omega_*^\circ(P_{\T^\perp}Q)\}\,\Omega_*(P_{\T^\perp}\Delta)
  +\{\lambda_S-\|P_{\Omega^c}Q\|_\infty\}\,\|P_{\Omega^c}\Delta\|_1 .
\end{align*}
Both coefficients are positive by \eqref{eq:certificate-conditions}. The
objective therefore increases unless $P_{\T^\perp}\Delta=0$ and
$P_{\Omega^c}\Delta=0$. In that case
$\Delta\in\T\cap\OmegaS=\{0\}$ by Lemma~\ref{lem:trans}. Hence
$(R_0,S_0)$ is the unique optimum of the stated convex program. Its local
algebraic identifiability on the true tangent space and support follows
separately from Lemma~\ref{lem:trans}.
\end{proof}

\subsection{Proof of Lemma~\ref{lem:cone-equivalence}}

\begin{proof}
For $\Delta=\Delta_R+\Delta_S$,
\[
  \|\Delta\|_F^2
  =
  \|\Delta_R\|_F^2+\|\Delta_S\|_F^2
  +2\langle\Delta_R,\Delta_S\rangle .
\]
The cone transversality condition \eqref{eq:cone-transversality} implies
\[
  \|\Delta\|_F^2
  \ge
  (1-\kappa_{\mathrm{cone}})\{\|\Delta_R\|_F^2+\|\Delta_S\|_F^2\},
\]
which is \eqref{eq:cone-equivalence}.
\end{proof}

\subsection{Proof of Theorem~\ref{thm:error}}

\begin{proof}
Let $\Delta_R=\widehat R-R_0$, $\Delta_S=\widehat S-S_0$, and $\Delta_M=\Delta_R+\Delta_S$. By the basic objective inequality assumed in Theorem~\ref{thm:error} (which follows, in particular, from global optimality) and feasibility of $(R_0,S_0)$,
\[
  \mathcal L(\widehat R,\widehat S)-\mathcal L(R_0,S_0)
  +
  \lambda_R\{\Omega_*(\widehat R)-\Omega_*(R_0)\}
  +
  \lambda_S(\|\widehat S\|_{1,a}-\|S_0\|_{1,a})
  \le0,
\]
where $\mathcal L$ is the quadratic loss. Let
\[
  G=\nabla_M\mathcal L(R_0+S_0)
\]
be the score with respect to the combined matrix $M=R+S$. At the true
first-stage quantities, $G=-2G_T^{\mathrm f}$; replacing them by preliminary
estimates adds the remainder controlled by Assumption~\ref{ass:profile}.
Thus the factor of two is absorbed by the sufficiently large constants in
the stated penalty levels. Work on the event
\[
  \Omega_*^\circ(G)\le \lambda_R/2,
  \qquad
  \|G\|_{\infty,a}^*\le \lambda_S/2,
\]
whose probability tends to one by the score conditions in Assumption~\ref{ass:hd-concentration}, Assumption~\ref{ass:profile}, and the stated penalty levels. Convexity gives
\[
  \mathcal L(\widehat R,\widehat S)-\mathcal L(R_0,S_0)
  \ge
  \langle G,\Delta_M\rangle .
\]
Using $\langle G,\Delta_R\rangle\ge-\Omega_*^\circ(G)\Omega_*(\Delta_R)$ and
$\langle G,\Delta_S\rangle\ge-\|G\|_{\infty,a}^*\|\Delta_S\|_{1,a}$, the basic
inequality implies
\begin{align*}
  &\lambda_R\{\Omega_*(R_0)-\Omega_*(R_0+\Delta_R)\}
  +\lambda_S(\|S_0\|_{1,a}-\|S_0+\Delta_S\|_{1,a}) \\
  &\hspace{1.5cm}\ge
  \langle G,\Delta_M\rangle
  \ge
  -\frac{\lambda_R}{2}\Omega_*(\Delta_R)
  -\frac{\lambda_S}{2}\|\Delta_S\|_{1,a} .
\end{align*}
For the projected nuclear norm, let $\Delta_{R,\T}=P_\T\Delta_R$ and
$\Delta_{R,\T^\perp}=P_{\T^\perp}\Delta_R$. The subspace pair
$(\T,\T^\perp)$ is decomposable for $\Omega_*$, so
\[
  \Omega_*(R_0+\Delta_R)
  \ge
  \Omega_*(R_0)+\Omega_*(\Delta_{R,\T^\perp})-\Omega_*(\Delta_{R,\T}) .
\]
Similarly, because $S_0$ is supported on $\Omega$,
\[
  \|S_0+\Delta_S\|_{1,a}
  \ge
  \|S_0\|_{1,a}+\|P_{\Omega^c}\Delta_S\|_{1,a}-\|P_\Omega\Delta_S\|_{1,a} .
\]
Substituting these two decomposability inequalities into the left side of the basic inequality, and the triangle inequalities $\Omega_*(\Delta_R)\le\Omega_*(P_\T\Delta_R)+\Omega_*(P_{\T^\perp}\Delta_R)$ and $\|\Delta_S\|_{1,a}\le\|P_\Omega\Delta_S\|_{1,a}+\|P_{\Omega^c}\Delta_S\|_{1,a}$ into the right side, and collecting terms gives
\[
  \frac{\lambda_R}{2}\,\Omega_*(P_{\T^\perp}\Delta_R)
  +\frac{\lambda_S}{2}\,\|P_{\Omega^c}\Delta_S\|_{1,a}
  \le
  \frac{3\lambda_R}{2}\,\Omega_*(P_\T\Delta_R)
  +\frac{3\lambda_S}{2}\,\|P_\Omega\Delta_S\|_{1,a},
\]
which is exactly the weighted penalty constraint in \eqref{eq:joint-cone}. Because both $\widehat S$ (by the constraint \eqref{eq:row-degree}) and $S_0$ (by Assumption~\ref{ass:network}(iii)) have row degree at most $\bar d_N=2d_N$ and $d_N$ respectively, each row of $\Delta_S$ has at most $3d_N$ nonzero entries, and row-wise Cauchy--Schwarz gives $\sum_i\|\Delta_{S,i\cdot}\|_1^2\le4d_N\|\Delta_S\|_F^2$.
Thus $(\Delta_R,\Delta_S)\in\mathcal C_\oplus$. On this cone, the
restricted-curvature condition \eqref{eq:rsc} in
Assumption~\ref{ass:hd-concentration} gives
\[
  \mathcal L(\widehat R,\widehat S)-\mathcal L(R_0,S_0)
  -\langle G,\Delta_M\rangle
  =\frac1T\sum_{t=1}^T\|\Delta_M\widehat x_t\|_2^2
  \ge
  \underline\phi\|\Delta_M\|_F^2,
\]
where the equality holds exactly because $\mathcal L$ is quadratic in $M$. Combining this with the basic inequality, the score bounds, the decomposability inequalities, and dropping the (nonnegative) $\T^\perp$ and $\Omega^c$ terms on the left yields
\[
  \underline\phi\|\Delta_M\|_F^2
  \le
  \frac{3\lambda_R}{2}\,\Omega_*(P_\T\Delta_R)+\frac{3\lambda_S}{2}\,\|P_\Omega\Delta_S\|_{1,a}.
\]
By Assumption~\ref{ass:network}(ii) and $|\A|=s_N$,
\[
  \Omega_*(P_\T\Delta_R)
  \le C_{\mathrm{tan}}\sqrt{2r}\,\|\Delta_R\|_F,
  \qquad
  \|P_\Omega\Delta_S\|_{1,a}\le C_a\sqrt{s_N}\,\|\Delta_S\|_F,
\]
where the second bound is Cauchy--Schwarz on the $s_N$ active entries. By the Cauchy--Schwarz inequality on the two-term sum and Lemma~\ref{lem:cone-equivalence},
\[
  \underline\phi\|\Delta_M\|_F^2
  \le
  C\bigl(\lambda_R^2r+\lambda_S^2s_N\bigr)^{1/2}
  \bigl(\|\Delta_R\|_F^2+\|\Delta_S\|_F^2\bigr)^{1/2}
  \le
  \frac{C\bigl(\lambda_R^2r+\lambda_S^2s_N\bigr)^{1/2}}
       {(1-\kappa_{\mathrm{cone}})^{1/2}}\,\|\Delta_M\|_F,
\]
so $\|\Delta_M\|_F\le C\underline\phi^{-1}(1-\kappa_{\mathrm{cone}})^{-1/2}(\lambda_R^2r+\lambda_S^2s_N)^{1/2}$. A final application of Lemma~\ref{lem:cone-equivalence} converts this into
\[
  \|\Delta_R\|_F^2+\|\Delta_S\|_F^2
  \le
  \frac{1}{1-\kappa_{\mathrm{cone}}}\,\|\Delta_M\|_F^2
  \le
  \frac{C}{\underline\phi^2(1-\kappa_{\mathrm{cone}})^2}
  \{\lambda_R^2r+\lambda_S^2s_N\},
\]
where the final constant $C$ depends only on fixed numerical constants,
$C_{\mathrm{tan}}$, and the weight bounds $(c_a,C_a)$. This proves the stated
error bound.
\end{proof}

\subsection{Primitive content of the irrepresentability condition}
\label{sec:irrep}

Theorem~\ref{thm:oracle} takes the population and empirical irrepresentability conditions as hypotheses. The following lemma gives a primitive population condition and a deterministic perturbation argument conditional on the estimated-defactored Gram rate in Assumption~\ref{ass:defactor}. This explains why factor adjustment is useful for support selection without claiming that the POET rate follows from the network assumptions alone.

\begin{lemma}
\label{lem:irrep}
Write the filtered predictor design as $x_t=\Lambda f_t^x+e_t^x$ with $\Gamma_x=\Lambda\Sigma_f^x\Lambda^\top+\Sigma_{e,x}$, as in Assumption~\ref{ass:network}(iv). A pervasive factor can make the raw design violate irrepresentability because its cross-block Gram entries need not vanish. Work instead with the population-defactored design $\widetilde x_t^0=(I-P_\Lambda)x_t$ and $\widetilde\Gamma=(I-P_\Lambda)\Gamma_x(I-P_\Lambda)=(I-P_\Lambda)\Sigma_{e,x}(I-P_\Lambda)$, estimated using $(I-\widehat P_{r_f})\widehat x_t$. Rows with $\A_i=\varnothing$ are interpreted using the vacuous convention in Theorem~\ref{thm:oracle}. For nonempty active rows, suppose uniformly in $i$ that the active predictor block has diagonal-dominance gap $\sigma_-^2-\theta>0$ and bounded cross-leakage $\max_{j\in\A_i^c}\sum_{k\in\A_i}|\widetilde\Gamma_{jk}|\le\ell\rho_e$, with the weighted margin
\[
  \frac{C_a}{c_a}\frac{\ell\rho_e}{\sigma_-^2-\theta}\le1-\varepsilon .
\]
Then the row-wise population weighted irrepresentability condition holds:
\[
  \max_i\left\|D_{i,\A_i^c}^{-1}\widetilde\Gamma_{\A_i^c\A_i}
  \widetilde\Gamma_{\A_i\A_i}^{-1}D_{i,\A_i}\right\|_{\op,\infty}
  \le1-\varepsilon .
\]
Under the estimated-defactored Gram condition in Assumption~\ref{ass:defactor},
its empirical version holds with margin $1-\varepsilon/2$ with probability
tending to one provided $d_N a_{\mathrm{df},T}=o(\varepsilon)$.
\end{lemma}

\begin{proof}
Fix a response row $i$. For the population statement, the cross-leakage assumption gives
\[
  \|\widetilde\Gamma_{\A_i^c\A_i}\|_{\op,\infty}
  =
  \max_{j\in\A_i^c}\sum_{k\in\A_i}|\widetilde\Gamma_{jk}|
  \le \ell\rho_e .
\]
The within-support diagonal dominance condition means that every row of
$\widetilde\Gamma_{\A_i\A_i}$ satisfies
\[
  |\widetilde\Gamma_{jj}|-\sum_{k\in\A_i,k\ne j}|\widetilde\Gamma_{jk}|
  \ge \sigma_-^2-\theta>0 .
\]
By Varah's bound for strictly diagonally dominant matrices \citep{varah1975},
\[
  \|\widetilde\Gamma_{\A_i\A_i}^{-1}\|_{\op,\infty}
  \le
  (\sigma_-^2-\theta)^{-1}.
\]
Therefore, by submultiplicativity,
\[
  \|D_{i,\A_i^c}^{-1}\widetilde\Gamma_{\A_i^c\A_i}
  \widetilde\Gamma_{\A_i\A_i}^{-1}D_{i,\A_i}\|_{\op,\infty}
  \le
  \frac{C_a}{c_a}\frac{\ell\rho_e}{\sigma_-^2-\theta}
  \le 1-\varepsilon .
\]

For the empirical statement, write
\[
  \Delta_\Gamma=\widehat{\widetilde\Gamma}-\widetilde\Gamma,
  \qquad
  \|\Delta_\Gamma\|_{\max}=O_p(a_{\mathrm{df},T}),
\]
by Assumption~\ref{ass:defactor}. Since $|\A_i|\le d_N$, the induced block
errors obey
\[
  \|\Delta_{\A_i^c\A_i}\|_{\op,\infty}
  \le d_N\|\Delta_\Gamma\|_{\max}
  =O_p(d_N a_{\mathrm{df},T}),
  \qquad
  \|\Delta_{\A_i\A_i}\|_{\op,\infty}
  \le d_N\|\Delta_\Gamma\|_{\max}
  =O_p(d_N a_{\mathrm{df},T}).
\]
If $d_N a_{\mathrm{df},T}=o(\varepsilon)$, then
\[
  \|\widetilde\Gamma_{\A_i\A_i}^{-1}\Delta_{\A_i\A_i}\|_{\op,\infty}
  \le
  \|\widetilde\Gamma_{\A_i\A_i}^{-1}\|_{\op,\infty}
  \|\Delta_{\A_i\A_i}\|_{\op,\infty}
  =o_p(1),
\]
so the Neumann expansion is valid with probability tending to one:
\[
  \widehat{\widetilde\Gamma}_{\A_i\A_i}^{-1}
  =
  (\widetilde\Gamma_{\A_i\A_i}+\Delta_{\A_i\A_i})^{-1}
  =
  (I+\widetilde\Gamma_{\A_i\A_i}^{-1}\Delta_{\A_i\A_i})^{-1}
  \widetilde\Gamma_{\A_i\A_i}^{-1}.
\]
Consequently
\[
  \|\widehat{\widetilde\Gamma}_{\A_i\A_i}^{-1}
  -\widetilde\Gamma_{\A_i\A_i}^{-1}\|_{\op,\infty}
  \le
  C\|\widetilde\Gamma_{\A_i\A_i}^{-1}\|_{\op,\infty}^2
  \|\Delta_{\A_i\A_i}\|_{\op,\infty}
  =
  O_p(d_N a_{\mathrm{df},T}).
\]
Now decompose the weighted empirical irrepresentability matrix as
\begin{align*}
  D_{i,\A_i^c}^{-1}\widehat{\widetilde\Gamma}_{\A_i^c\A_i}
  \widehat{\widetilde\Gamma}_{\A_i\A_i}^{-1}D_{i,\A_i}
  -
  D_{i,\A_i^c}^{-1}\widetilde\Gamma_{\A_i^c\A_i}
  \widetilde\Gamma_{\A_i\A_i}^{-1}D_{i,\A_i}
  &=
  D_{i,\A_i^c}^{-1}\Delta_{\A_i^c\A_i}
  \widetilde\Gamma_{\A_i\A_i}^{-1}D_{i,\A_i}\\
  &\quad+
  D_{i,\A_i^c}^{-1}\widehat{\widetilde\Gamma}_{\A_i^c\A_i}
  (\widehat{\widetilde\Gamma}_{\A_i\A_i}^{-1}
  -\widetilde\Gamma_{\A_i\A_i}^{-1})D_{i,\A_i} .
\end{align*}
The first term is $O_p(d_N a_{\mathrm{df},T})$ in
$\|\cdot\|_{\op,\infty}$. The second term has the same order because
\[
  \|D_{i,\A_i^c}^{-1}
  \widehat{\widetilde\Gamma}_{\A_i^c\A_i}\|_{\op,\infty}
  \le c_a^{-1}\{\ell\rho_e+o_p(1)\},
  \qquad
  \|D_{i,\A_i}\|_{\op,\infty}\le C_a.
\]
Thus
\[
  \left\|
  D_{i,\A_i^c}^{-1}\widehat{\widetilde\Gamma}_{\A_i^c\A_i}
  \widehat{\widetilde\Gamma}_{\A_i\A_i}^{-1}D_{i,\A_i}
  -
  D_{i,\A_i^c}^{-1}\widetilde\Gamma_{\A_i^c\A_i}
  \widetilde\Gamma_{\A_i\A_i}^{-1}D_{i,\A_i}
  \right\|_{\op,\infty}
  =
  O_p(d_N a_{\mathrm{df},T})
  =
  o_p(\varepsilon).
\]
All bounds are uniform in $i$ because $|\A_i|\le d_N$ and the diagonal-dominance and leakage constants are uniform. Combining the perturbation bound with the population margin $1-\varepsilon$ gives the empirical margin $1-\varepsilon/2$ simultaneously for all rows with probability tending to one.
\end{proof}

\subsection{Proof of Theorem~\ref{thm:oracle}}

\begin{proof}
The argument is a primal-dual witness construction in the tradition of
\citet{zhao2006} and \citet{wainwright2009}, with adaptive weights as in
\citet{zou2006}, applied row by row to the second-step defactored problem \eqref{eq:two-step}. The Euclidean loss separates across response rows. Let $\mathcal L_i$ denote the quadratic loss for response $i$; its predictor Gram is $\widehat{\widetilde\Gamma}$, and its score at $S_{0,i\cdot}$ is the corresponding row of the innovation score minus $\Delta_T$ in \eqref{eq:contamination}. After support recovery, the final connectedness refit re-estimates all selected sparse entries jointly with the identifiable projected-rank nuisance defined in \eqref{eq:oracle-refit}.

Fix $i$ and construct the witness row $\widetilde S_{i\cdot}$ by solving the restricted problem on $\A_i$ and setting $\widetilde S_{ij}=0$ for $j\in\A_i^c$. Its restricted KKT equation is
\begin{equation}
  0
  =
  \nabla_{\A_i}\mathcal L_i(S_{0,i\cdot})
  +
  \widehat{\widetilde\Gamma}_{\A_i\A_i}
  (\widetilde S_{i,\A_i}-S_{0,i,\A_i})
  +
  \lambda_S D_{i,\A_i}\operatorname{sign}(\widetilde S_{i,\A_i}),
  \label{eq:restricted-kkt}
\end{equation}
On the uniform inverse-Gram event in the theorem and the score event
$\|\nabla_{\A_i}\mathcal L_i(S_{0,i\cdot})\|_\infty\le c\lambda_S$,
$\|\widehat{\widetilde\Gamma}_{\A_i\A_i}^{-1}\|_{\op,\infty}
\le2/\phi_{\min}$.
Solving \eqref{eq:restricted-kkt} gives
\[
  \|\widetilde S_{i,\A_i}-S_{0,i,\A_i}\|_\infty
  \le C\lambda_S/\phi_{\min}.
\]
The beta-min condition therefore gives
\[
  \operatorname{sign}(\widetilde S_{i,\A_i})
  =\operatorname{sign}(S_{0,i,\A_i}).
\]
It remains to verify that the restricted row is also the unrestricted row solution. The KKT condition on $\A_i^c$ requires
\[
  Z_{i,\A_i^c}
  =
  \lambda_S^{-1}D_{i,\A_i^c}^{-1}
  \left[
    \nabla_{\A_i^c}\mathcal L_i(\widetilde S_{i\cdot})
  \right]
\]
to satisfy $\|Z_{i,\A_i^c}\|_\infty<1$. Because the loss is quadratic,
\[
  \nabla_{\A_i^c}\mathcal L_i(\widetilde S_{i\cdot})
  =
  \nabla_{\A_i^c}\mathcal L_i(S_{0,i\cdot})
  -
  \widehat{\widetilde\Gamma}_{\A_i^c\A_i}
  \widehat{\widetilde\Gamma}_{\A_i\A_i}^{-1}
  \{\nabla_{\A_i}\mathcal L_i(S_{0,i\cdot})
  +\lambda_S D_{i,\A_i}\operatorname{sign}(S_{0,i,\A_i})\}.
\]
For the quadratic loss this expansion is exact. The contamination
$\Delta_T$ enters its score and satisfies
$\|\Delta_T\|_{\max}=o_p(\lambda_S)$ by Assumption~\ref{ass:defactor}.
Theorem~\ref{thm:error} alone is insufficient because the first-step error
contains terms of orders $\lambda_R\sqrt r$ and
$\lambda_S\sqrt{s_N}$. Defactoring and the leakage condition reduce its
contribution to the sparse score below $\lambda_S$. The estimated-defactored
score satisfies
\[
  \|\nabla\mathcal L(S_0)+\Delta_T\|_\infty
  =
  O_p\!\left(\sqrt{\frac{(\log T)^2\log N}{T}}+\sqrt{\frac{\log N}{m}}\right)
  \le c\lambda_S
\]
with probability tending to one by Assumption~\ref{ass:defactor}. Since
$\|\Delta_T\|_{\max}=o_p(\lambda_S)$, the same event, after a harmless
adjustment of $c$, bounds the total loss gradient
$\|\nabla\mathcal L(S_0)\|_\infty$. The boundedness of the inverse adaptive
weights converts this into the corresponding weighted score bound. Uniformly
over $i$, the empirical weighted irrepresentability condition gives
\[
  \left\|D_{i,\A_i^c}^{-1}
  \widehat{\widetilde\Gamma}_{\A_i^c\A_i}
  \widehat{\widetilde\Gamma}_{\A_i\A_i}^{-1}D_{i,\A_i}
  \right\|_{\op,\infty}\le 1-\varepsilon/2.
\]
Combining the last three displays,
\[
  \|Z_{i,\A_i^c}\|_\infty
  \le
  (1-\varepsilon/2)\|\operatorname{sign}(S_{0,i,\A_i})\|_\infty
  +
  O_p(\|\nabla\mathcal L(S_0)\|_\infty/\lambda_S)
  +
  o_p(1)
  \le 1-\varepsilon/4
\]
after choosing the penalty constant large enough relative to the score
constant. Strict dual feasibility implies that the witness is the unique
solution for response row $i$. The score, inverse-Gram, beta-min, and
irrepresentability events hold uniformly over $i$, so intersecting them gives
$\widehat S^{(2)}=\widetilde S$ for every row with probability tending to one.
Combining strict dual feasibility with the beta-min sign argument yields
\[
  \p\{\supp(\widehat S^{(2)})=\A,\ \operatorname{sign}(\widehat S^{(2)}_\A)=\operatorname{sign}(S_{0,\A})\}\to1.
\]
\end{proof}

\subsection{Proof of Theorem~\ref{thm:conn-inference}}

\begin{proof}
Let $\A_0$ be the true support. By Theorem~\ref{thm:oracle}, $\p(\widehat\A=\A_0)\to1$. On this event, the support-indexed initialization and the fixed measurable local-solution rule are the same as those used when $\A_0$ is known. Hence $\widetilde\eta_{\widehat\A}=\mathcal S_T(\A_0,\bar\eta_{\A_0})=\widetilde\eta_{\A_0}$ on the support-recovery event; this identity concerns the same local constrained-KKT selection and does not invoke global optimality.

For the fixed-dimensional local oracle KKT refit, use the identifiable projected-rank chart \eqref{eq:rank-chart}. Let $s_{\A_0,t}=2J_{\A_0,t}^\top\widetilde u_t$ and let $H_{\A_0}$ be \eqref{eq:oracle-hessian}. The stability margin in Assumption~\ref{ass:oracle-functional} makes the cap in \eqref{eq:oracle-refit} locally inactive with probability tending to one, so the local KKT equation is \eqref{eq:oracle-kkt}. The consistent initialization, local identification, and nonsingular-Hessian conditions then give the usual first-order expansion
\[
  \sqrt T(\widetilde\eta_{\A_0}-\eta_{\A_0,0})
  =
  H_{\A_0}^{-1}T^{-1/2}\sum_{t=1}^T s_{\A_0,t}+o_p(1).
\]
For a differentiable connectedness functional,
\[
  \sqrt T\{g(\widetilde\eta_{\A_0})-g(\eta_{\A_0,0})\}
  =
  \nabla g(\eta_{\A_0,0})^\top
  H_{\A_0}^{-1}T^{-1/2}\sum_{t=1}^T s_{\A_0,t}+o_p(1).
\]
Under the score linearization and central limit theorem maintained in Assumption~\ref{ass:oracle-functional}, the fixed-dimensional influence function is therefore
\[
  \psi_{g,t}=\nabla g(\eta_{\A_0,0})^\top H_{\A_0}^{-1}s_{\A_0,t}.
\]
When the selected support or projected-rank nuisance dimension grows, this low-dimensional representation is imposed rather than derived.
Assumption~\ref{ass:oracle-functional} then gives
\[
  \sqrt T\{g(\widetilde\eta_{\A_0})-g(\eta_{\A_0,0})\}
  =
  T^{-1/2}\sum_{t=1}^T\psi_{g,t}+o_p(1)
  \xrightarrow{d}
  \mathcal N(0,V_g).
\]
For any bounded Lipschitz test function $\varphi$,
\[
  \left|
  \E\varphi\!\left[\sqrt T\{g(\widetilde\eta_{\widehat\A})-g(\eta_{\A_0,0})\}\right]
  -
  \E\varphi\!\left[\sqrt T\{g(\widetilde\eta_{\A_0})-g(\eta_{\A_0,0})\}\right]
  \right|
  \le
  2\|\varphi\|_\infty\,\p(\widehat\A\ne\A_0),
\]
so replacing $\widetilde\eta_{\A_0}$ by $\widetilde\eta_{\widehat\A}$
does not change the limit law.

It remains to record the differentiability of the connectedness map used in Assumption~\ref{ass:oracle-functional}. Let
$A(\eta)=(\gamma+\beta)I_N+M(\eta)$ and
$D_A(\eta)=I_N-A(\eta)$. On the stability region,
\[
  \Phi(\eta)=D_A(\eta)^{-1}M(\eta)
\]
is continuously differentiable, with differential
\[
  d\Phi
  =
  D_A^{-1}(dA)D_A^{-1}M+D_A^{-1}dM.
\]
If a connectedness functional uses absolute response masses, the local
non-kink condition stated in the theorem requires each relevant response entry
to be bounded away from zero or locally fixed at zero by a maintained
restriction. The latter entries are omitted from the active derivative.
The denominator bounded-away-from-zero condition then makes the quotient map
locally continuously differentiable. Consistency of the variance estimator follows from consistency of the oracle influence-function sequence: its sample second moment is used in the martingale-difference case and a standard HAC long-run variance estimator is used in the mixing case.
\end{proof}

\subsection{Proof of Corollary~\ref{cor:net-tests}}

\begin{proof}
The null $H_0:C_i^{\mathrm{net}}=0$ is the scalar smooth restriction
$g_i(\eta_{\A,0})=0$, where $g_i(\eta)=C_i^{\mathrm{net}}(\eta)$. Under the
regularity conditions of Theorem~\ref{thm:conn-inference},
\[
  \sqrt T\{g_i(\widetilde\eta_{\widehat\A})-g_i(\eta_{\A,0})\}
  \xrightarrow{d}
  \mathcal N(0,V_i).
\]
The plug-in estimator $\widehat V_i$ is consistent by
Theorem~\ref{thm:conn-inference}. Hence, under $H_0$,
\[
  W_i
  =
  T\,g_i(\widetilde\eta_{\widehat\A})^2/\widehat V_i
  \xrightarrow{d}\chi^2_1,
\]
which gives the stated Wald test. For simultaneous node-level bands, collect the relevant
node connectedness functionals in a vector $g(\eta)$. The multivariate version of
Theorem~\ref{thm:conn-inference} and the continuous mapping theorem give the
joint Gaussian limit with covariance given by the corresponding influence-function covariance. Critical values from this
estimated Gaussian law produce asymptotically valid simultaneous bands. Applying the scalar version directly to $C^{\mathrm{total}}$ gives its pointwise confidence interval.
\end{proof}

\bibliographystyle{apalike}
\bibliography{references}

\end{document}


\maketitle

This Online Appendix reports the primitive Gaussian example, additional
simulation results, data-retention details, tuning diagnostics, and empirical
robustness checks.

\section{A Primitive Gaussian Verification}

This section verifies the restricted-curvature and network-score conditions
under a Gaussian factor model. The result is a sufficient example rather than
a primitive characterization of the full semimartingale model.

\begin{proposition}
\label{prop:oa-gaussian}
Let the latent filtered design have the factor representation
\[
  x_t=\Lambda f_t+e_t,
\]
where the factor dimension is fixed, the rows of $\Lambda$ have uniformly
bounded Euclidean norm, and the eigenvalues of
$N^{-1}\Lambda^\top\Lambda$ are bounded above and away from zero. Suppose
$\{(f_t,e_t)\}$ is a centered stationary Gaussian triangular array with
geometric beta-mixing coefficients bounded by $C\exp(-ck)$, $f_t$ and $e_t$
are independent at every lead and lag,
\[
  c_f I\preceq\Var(f_t)\preceq C_f I,
  \qquad
  \phi_e I_N\preceq\Var(e_t)\preceq C_e I_N,
\]
and their coordinatewise sub-Gaussian norms are uniformly bounded.
Assume additionally that the factor and idiosyncratic autocovariance
operators are uniformly absolutely summable:
\[
  \sum_{h\in\mathbb Z}\|\operatorname{Cov}(f_t,f_{t-h})\|_2\le C,
  \qquad
  \sup_N\sum_{h\in\mathbb Z}
  \|\operatorname{Cov}(e_t,e_{t-h})\|_2\le C.
\]

Let $\widetilde u_t$ be an i.i.d.\ centered Gaussian $N$-vector, independent
of the complete design path, with covariance eigenvalues bounded above.
Model the feasible high-frequency perturbation by
\[
  \widehat x_t=x_t+v_t,
  \qquad
  \widehat\varepsilon_t^{\,\mathrm f}
  =\widetilde u_t+r_t,
\]
where $\{(v_t,r_t)\}$ is independent of
$\{(x_t,\widetilde u_t)\}$, is Gaussian and geometrically beta-mixing, and
satisfies
\[
  \max_j\{\|v_{j,t}\|_{\psi_2},\|r_{j,t}\|_{\psi_2}\}
  \le C m^{-1/2},
  \qquad
  \max\{\|\Var(v_t)\|_2,\|\Var(r_t)\|_2\}\le C/m.
\]
Assume $\log N=O(\log T)$ and
\[
  \frac{N^2(\log T)^2}{T}+\frac{N^2}{m}=o(1).
\]
Then, with probability tending to one,
\[
  \inf_{\Delta\ne0}
  \frac{T^{-1}\sum_{t=1}^T\|\Delta\widehat x_t\|_2^2}
       {\|\Delta\|_F^2}
  \ge \frac{\phi_e}{2}.
\]
Thus the restricted-curvature condition of the main paper holds on its joint
cone. Moreover, for
\[
  G_T^{\mathrm f}
  =T^{-1}\sum_{t=1}^T
  \widehat\varepsilon_t^{\,\mathrm f}\widehat x_t^\top,
\]
\[
  \|G_T^{\mathrm f}\|_{\max}
  =
  O_p\!\left(
    \log T\sqrt{\frac{\log N}{T}}
    +\sqrt{\frac{\log N}{m}}
  \right),
\]
and
\[
  \|P_{\mathrm{off}}(G_T^{\mathrm f})\|_2
  =
  O_p\!\left(
    \frac{N\log T}{\sqrt T}+\frac{N}{\sqrt m}
  \right).
\]
Because the adaptive weights are bounded above and away from zero and
$\Omega_*^\circ(G)=\|P_{\mathrm{off}}(G)\|_2$, these bounds imply the two
score orders in the main paper up to fixed constants.
\end{proposition}

\begin{proof}
Write $F=(f_1,\ldots,f_T)^\top$ and
$E=(e_1,\ldots,e_T)^\top$. The population Gram is
\[
  \Gamma_x
  =\Lambda\Var(f_t)\Lambda^\top+\Var(e_t)
  \succeq\phi_e I_N.
\]
Geometric blocking with blocks of length $C\log T$, followed by Gaussian
quadratic-form concentration on the approximately independent blocks, gives
\[
  \left\|T^{-1}F^\top F-\Var(f_t)\right\|_2
  =O_p(\log T/\sqrt T),
\]
\[
  \|T^{-1}F^\top E\|_2
  =O_p\!\left(\log T\sqrt{N/T}\right),
\]
and
\[
  \left\|T^{-1}E^\top E-\Var(e_t)\right\|_2
  =
  O_p\!\left(
    \log T\sqrt{N/T}+N\log T/T
  \right).
\]
Since $\|\Lambda\|_2=O(\sqrt N)$, expansion of
$T^{-1}\sum_tx_tx_t^\top$ around $\Gamma_x$ therefore yields
\[
  \left\|T^{-1}\sum_tx_tx_t^\top-\Gamma_x\right\|_2
  =O_p(N\log T/\sqrt T)=o_p(1).
\]
The same calculation and the assumed covariance scale of $v_t$ give
\[
  \left\|T^{-1}\sum_t
  \widehat x_t\widehat x_t^\top
  -T^{-1}\sum_tx_tx_t^\top\right\|_2
  =O_p(N/\sqrt m)=o_p(1).
\]
Weyl's inequality now implies
$\lambda_{\min}(T^{-1}\sum_t\widehat x_t\widehat x_t^\top)
\ge\phi_e/2$. The displayed curvature ratio is the quadratic form of this
sample Gram, so the claim holds for every matrix $\Delta$ and hence on the
smaller joint cone.

For the entrywise score, each coordinate of
$\widetilde u_tx_t^\top$ is a centered product of Gaussian variables with a
uniform sub-exponential norm. Applying the same geometric blocking argument
and taking a union bound over at most $N^2$ entries gives
\[
  \left\|T^{-1}\sum_t\widetilde u_tx_t^\top\right\|_{\max}
  =O_p\!\left(\log T\sqrt{\log N/T}\right).
\]
The terms containing $v_t$ or $r_t$ have the same product structure with one
factor of scale $m^{-1/2}$. Gaussian maximal inequalities and
$N^2/m=o(1)$ give
\[
  \left\|
    T^{-1}\sum_t
    \{\widehat\varepsilon_t^{\,\mathrm f}\widehat x_t^\top
      -\widetilde u_tx_t^\top\}
  \right\|_{\max}
  =
  O_p\!\left(\sqrt{\frac{\log N}{m}}\right).
\]

For the spectral score, decompose
\[
  T^{-1}\sum_t\widetilde u_tx_t^\top
  =
  (T^{-1}U^\top F)\Lambda^\top+T^{-1}U^\top E,
\]
where $U=(\widetilde u_1,\ldots,\widetilde u_T)^\top$. Gaussian operator-norm
concentration, again with the logarithmic blocking factor, gives
\[
  \|T^{-1}U^\top F\|_2
  =O_p(\log T\sqrt{N/T}),
\]
while the idiosyncratic term is of smaller order under the displayed growth
condition. Multiplication by $\|\Lambda\|_2=O(\sqrt N)$ gives
$O_p(N\log T/\sqrt T)$. Terms involving $v_t$ and $r_t$ are bounded
conservatively by $O_p(N/\sqrt m)$. Deleting the diagonal cannot increase
these orders by more than a fixed factor because
$\|P_{\mathrm{off}}(A)\|_2\le\|A\|_2+\|\diag(A)\|_2
\le2\|A\|_2$. This proves the projected-dual score bound.
\end{proof}

\section{Additional Simulation Results}

\subsection{Algorithm alignment}

Two hundred independent replications compare the completion-variable
projected-nuclear solution used in the paper with the earlier zero-diagonal
singular-value-thresholding surrogate. The completion implementation lowers
the penalized objective by 0.00032 on average, has mean KKT residual
$2.25\times10^{-6}$, satisfies the strict convergence criterion in 98 percent
of replications, and produces a stable recursion in every replication. The
final estimator retains the penalized projected component and does not impose
a selected completion rank; reported ranks are numerical diagnostics.

\subsection{Complete recovery and forecast diagnostics}

Table~\ref{tab:oa-sim-v5-headline} reports the complete decomposition, support,
sign, and rank diagnostics. Penalties are selected by replication-specific
validation QLIKE subject to $c_S/c_R=1$, followed by row-wise GIC support
selection with multiplier $0.25$; neither step uses the true support.
$\widehat r_{\mathrm{diag}}$ is a thresholded completion-rank diagnostic, not
a constraint on the final estimator, and
$\kappa_{\mathrm{sep}}=\|P_\Omega P_T\|$ is computed from the DGP tangent and
support spaces. Sign rate is the fraction of true edges with the correct
estimated sign. Exact support and exact sign refer to the full matrix, whereas
exact rows is the fraction of rows with exactly recovered support. Full-matrix
exact recovery is rare in the moderate-signal designs, so the row-wise and F1
measures are more informative.

\input{simulation_tables/table_v5_headline}

Table~\ref{tab:sim-v5-forecast} reports all component ablations and external
forecast benchmarks. Entries are Monte Carlo means, with Monte Carlo standard
errors in parentheses. Positive loss differences favor the full model.
Component ablations hold the fitted dynamic coefficients fixed and remove the
indicated network component; they are not separately retuned estimators.
Tuning, benchmark estimation, and smearing corrections use only the training
sample.

\input{simulation_tables/table_v5_forecast}

\subsection{Selected-estimator coverage and nonregular designs}

Table~\ref{tab:sim-v5-selected-coverage-full} reports unconditional and
target-selected coverage for the complete selected estimator. It also includes
the stronger-signal, beta-min, and near-zero response designs. The last design
is outside the smooth-functional region used for the headline connectedness
inference and records the resulting finite-sample failure rather than treating
it as a regular case. Entries are Monte Carlo means with standard errors in
parentheses. Each bootstrap draw repeats projected-component estimation,
predictor-rank selection, sparse-support selection, and dynamic estimation.
Coverage U is unconditional; coverage C conditions only on selection of the
reported target edge, rather than on recovery of the full support. Normal and
percentile denote the corresponding bootstrap intervals. Net connectedness is
evaluated at the DGP node with the largest absolute net response. This
experiment evaluates the practical selected estimator and is distinct from
the fixed-support, fixed-rank experiment reported in the main paper.

\input{simulation_tables/table_v5_coverage}

\subsection{All price-level volatility proxies}

Table~\ref{tab:oa-sim-v5-price} reports relative integrated-variance RMSE for
every candidate realized-volatility proxy, together with network recovery and
forecast performance. Parentheses following the selected proxy give its
selection frequency. The pre-specified rule uses observed lag-one return
autocovariance and threshold exceedances, but not the DGP jump or noise labels.
Network penalties use validation QLIKE and row-wise GIC with multiplier $0.25$
without knowledge of the true support.

\input{simulation_tables/table_v5_price}

\section{Additional Empirical Results}

\subsection{Data retention}

The one-minute and five-minute filters are applied separately. A date is
retained only when every sector ETF and the auxiliary SPY series satisfies the
common regular-session coverage rule. Duplicate timestamps use the last
observation, isolated gaps use previous-tick interpolation, and a leading gap
is initialized from the first observed bar's opening price. Early closes and
days with more extensive missingness are excluded. Annual retention ratios use
the number of source trading dates in the corresponding calendar year as the
denominator.

The Parzen realized kernel uses the fixed bandwidth $H=20$, with weight
$k\{h/(H+1)\}$ at lag $h$, where $k(\cdot)$ is the Parzen kernel. Its
two-percent threshold-RV floor is a numerical safeguard and does not bind on
any retained sector-day. The jump threshold uses daily bipower variation
without a separate intraday-periodicity rescaling. Thus the realized-kernel and
five-minute results vary the noise correction and sampling frequency while
retaining the same daily truncation convention.

\begin{table}[htbp]
\centering
\caption{Annual balanced-panel retention}
\label{tab:oa-emp-retention}
\small
\input{generated_empirical/table_retained_days.tex}
\end{table}

\subsection{Complete forecast results}

For forecasting, the candidate sparse coefficients in each row are ordered by
absolute magnitude. For row degree $d=0,\ldots,3$, the criterion is
\[
  \operatorname{GIC}_i(d)
  =
  T_i\log\{\operatorname{RSS}_i(d)/T_i\}
  +0.25\,d\log T_i .
\]
The selected degree minimizes this criterion. The inference support does not
use this GIC: it is selected by the second-step weighted LASSO analyzed in the
main paper. Diebold--Mariano standard errors use Bartlett weights and bandwidth
$\lfloor T_{\mathrm{eval}}^{1/3}\rfloor$; the baseline bandwidth is 13. In the
tables below, a forecast gain is benchmark QLIKE minus proposed-model QLIKE,
so a positive value favors the proposed model. Component removals hold the
remaining full-fit parameters fixed and are not independently retuned nested
estimators. The asset-level table uses HAR as the benchmark.

\begin{table}[htbp]
\centering
\caption{Complete recursive forecast comparison}
\label{tab:oa-emp-forecast}
\small
\resizebox{\textwidth}{!}{%
\input{generated_empirical/table_forecast_summary.tex}}
\end{table}

\begin{table}[htbp]
\centering
\caption{Forecast comparison by sector ETF}
\label{tab:oa-emp-assets}
\small
\input{generated_empirical/table_forecast_by_asset.tex}
\end{table}

\subsection{Tuning and cross-proxy network diagnostics}

Table~\ref{tab:oa-emp-penalty} summarizes the pointwise validation-QLIKE
choices over the nine expanding annual windows. The sparsity-first
one-standard-error rule is reported only as a diagnostic and is not used for
the headline estimates. In Table~\ref{tab:oa-emp-network-robustness}, a stable
edge must be selected in at least two thirds of the annual fits, attain sign
stability of at least 80 percent, and have the same mean sign under one-minute
TSRV, one-minute realized kernel, and five-minute threshold RV. Columns
transmit and rows receive. Completion ranks are diagnostics for the penalized
fit and are not imposed on the forecast refit.

\begin{table}[htbp]
\centering
\caption{Sparse-penalty grid diagnostics}
\label{tab:oa-emp-penalty}
\small
\resizebox{\textwidth}{!}{%
\input{generated_empirical/table_penalty_diagnostics.tex}}
\end{table}

\begin{table}[htbp]
\centering
\caption{Cross-proxy network robustness}
\label{tab:oa-emp-network-robustness}
\small
\textit{Panel A: sparse edges stable under all specifications}\\[2pt]
\input{generated_empirical/table_cross_proxy_edges.tex}
\par\smallskip
\textit{Panel B: annual connectedness rankings and rank diagnostics}\\[2pt]
\resizebox{\textwidth}{!}{%
\input{generated_empirical/table_cross_proxy_rankings.tex}}
\end{table}

\subsection{Connectedness-inference sensitivity}

The released inference run uses 499 circular block bootstrap draws. All draws
pass the common stability and constrained-KKT screen. The largest accepted
scaled stationarity residual is $5.0\times10^{-4}$, and the smallest stability
margin is 0.0218. For the full-sample fit, the joint raw residual is
$2.25\times10^{-4}$; the raw $\gamma$ residual is $1.21\times10^{-4}$ with
scaling factor 1.487, the completion residual is $2.25\times10^{-4}$, and the
remaining parameter-block residuals are below $1.4\times10^{-13}$. The
stability cap is inactive for the full-sample fit, while any active-cap
bootstrap solution must satisfy the constrained normal-cone KKT condition.
The rank-one completion and empty weighted-LASSO support are imposed exactly
in every retained draw. The largest raw joint residual among retained draws is
$7.53\times10^{-4}$.
The block length is 10 trading days and is fixed before examining the
intervals. It retains two trading weeks of short-run dependence. The
circular block bootstrap is used as a finite-sample estimate of the
influence-function law maintained in the main paper; its standard asymptotic
justification lets the block length diverge while remaining $o(T)$.
Table~\ref{tab:oa-emp-rank-sensitivity} re-estimates the full-sample local
refit at each indicated completion rank. The weighted-LASSO support is empty,
and rank zero is the no-network benchmark.
Table~\ref{tab:oa-emp-horizon-sensitivity} uses the rank-one full-sample refit;
its infinite-horizon row is based on the stable cumulative response matrix.

\begin{table}[htbp]
\centering
\caption{Completion-rank sensitivity}
\label{tab:oa-emp-rank-sensitivity}
\small
\input{generated_empirical/table_rank_sensitivity.tex}
\end{table}

\begin{table}[htbp]
\centering
\caption{Connectedness by response horizon}
\label{tab:oa-emp-horizon-sensitivity}
\small
\input{generated_empirical/table_horizon_sensitivity.tex}
\end{table}

\subsection{Regime diagnostics}

Each period and sampling frequency in
Table~\ref{tab:oa-emp-regimes} is filtered, estimated, and tuned separately.
Total connectedness is the off-diagonal share of absolute cumulative-response
mass. The component norms are Frobenius norms of the projected-factor and
selected sparse estimates.

\begin{table}[htbp]
\centering
\caption{Network composition across market regimes and sampling frequencies}
\label{tab:oa-emp-regimes}
\small
\resizebox{\textwidth}{!}{%
\input{generated_empirical/table_regime_robustness.tex}}
\end{table}

Figure~\ref{fig:oa-emp-sparse-regimes} displays the one-minute selected sparse
component in each regime. Columns transmit and rows receive; red and blue
denote positive and negative coefficients. The panels use a common symmetric
color scale. They describe local selected channels rather than the full
response network, which also includes the projected-factor component and
dynamic feedback.

\begin{figure}[htbp]
\centering
\begin{subfigure}[t]{0.47\textwidth}
\centering
\includegraphics[width=\textwidth]{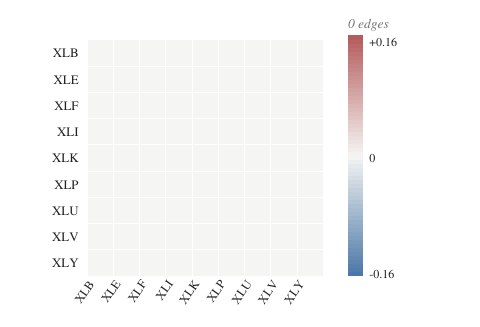}
\caption{Global financial crisis.}
\end{subfigure}
\hfill
\begin{subfigure}[t]{0.47\textwidth}
\centering
\includegraphics[width=\textwidth]{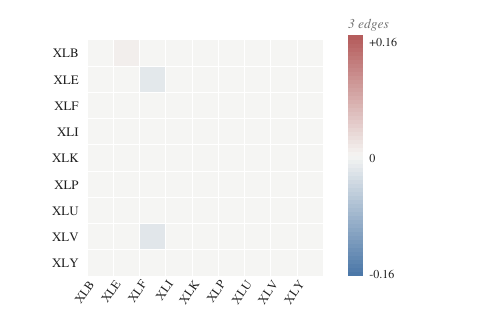}
\caption{Pre-pandemic period.}
\end{subfigure}

\medskip

\begin{subfigure}[t]{0.47\textwidth}
\centering
\includegraphics[width=\textwidth]{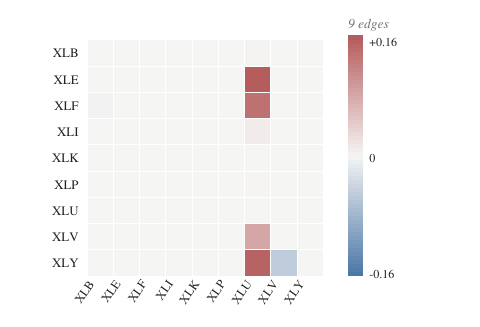}
\caption{COVID period.}
\end{subfigure}
\hfill
\begin{subfigure}[t]{0.47\textwidth}
\centering
\includegraphics[width=\textwidth]{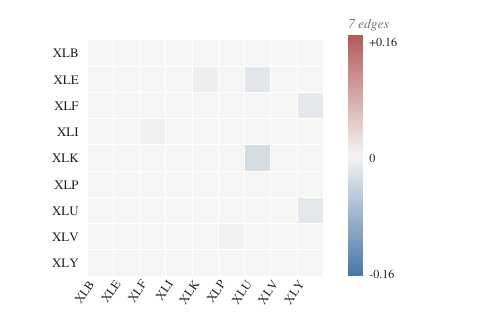}
\caption{Monetary tightening.}
\end{subfigure}

\medskip

\begin{subfigure}[t]{0.47\textwidth}
\centering
\includegraphics[width=\textwidth]{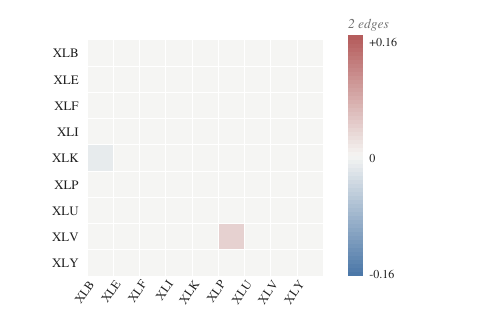}
\caption{Recent period.}
\end{subfigure}
\par\medskip
\caption{Selected sparse channels across market regimes}
\label{fig:oa-emp-sparse-regimes}
\end{figure}

%% file: simulation_section.tex
We conduct Monte Carlo experiments to examine the finite-sample performance of
the proposed procedure. The experiments consider the generated-regressor
approximation under joint in-fill asymptotics, network recovery and forecasting
under feasible tuning, and connectedness inference based on the fixed-support,
fixed-rank refit. Calibration and evaluation use independent seeds. We use 300
replications for the generated-regressor designs and 500 for the
network-recovery designs and the experiments that simulate intraday prices.
The Online Appendix reports algorithm diagnostics. Reported entries are Monte
Carlo means; parentheses give Monte Carlo standard errors where shown.

The baseline data-generating process is
\[
 q_t=\omega+\gamma q_{t-1}+(\beta I_N+M_0)z_{t-1},
 \qquad
 z_t=q_t+\widetilde u_t,
 \qquad
 M_0=R_0+S_0,\quad R_0=P_{\mathrm{off}}(L_0).
\]
The simulated log innovation is centered, so no Jensen recentering is needed
and $\omega$ here coincides with $\widetilde\omega$ in
\eqref{eq:q-recursion}.
Unless stated otherwise, $N=10$, $T=1000$, $m=390$,
$(\gamma,\beta)=(0.25,0.45)$, $\rank(L_0)=1$, and each row of $S_0$
has two signed edges of magnitude $0.08$. The innovation contains a pervasive
rank-one factor and idiosyncratic noise, and the transition matrix is scaled to
have spectral radius below $0.92$. Seventy percent of each sample is used for
estimation and 30 percent for forecast evaluation. In the generated-regressor
designs,
\[
 \widehat z_{i,t}
 =z_{i,t}+\log(\chi_m^2/m)+m^{-1},
\]
where $m^{-1}$ is the first-order log-bias correction.

All reported estimators are feasible. A preliminary recursion estimates
$(c_q,\gamma,\beta)$. The network step optimizes over a completion $L$, sets
$R=P_{\mathrm{off}}(L)$, and penalizes $\|L\|_*$. Validation QLIKE selects the
penalties in each replication. After defactoring, a row-wise information
criterion selects the sparse support, and the active coefficients and dynamic
parameters are refitted without shrinkage. The information criterion serves as
a feasible finite-sample proxy for the sparse-penalty choice in
Theorem~\ref{thm:oracle}. The true support and ranks are used only for
evaluation. The Online Appendix compares the completion algorithm with
zero-diagonal SVT and reports convergence diagnostics.

Table~\ref{tab:sim-v5-generated} compares the complete feasible fit with the
same estimator applied to latent integrated volatility, holding the feasible
penalty fixed in both fits. The $R$, $S$, and $M$ gaps are Frobenius distances
between the feasible and latent-IV estimates, and the QLIKE gap is feasible
minus latent. The feasible--latent $M$ and QLIKE gaps are large when the
intraday mesh is coarse. They fall to $0.076$ and $0.00007$ at
$(T,m)=(2500,1950)$ and to $0.041$ and $0.00004$ at $(5000,3900)$.
Increasing $T$ with $m=390$ improves parameter estimation but leaves the QLIKE
gap nearly unchanged. The result is consistent with the joint in-fill mechanism
in Theorem~\ref{thm:asym-normal}. It is not a direct simulation of that theorem
because the experiment re-estimates the complete network.

\input{simulation_tables/table_v5_generated}

Table~\ref{tab:sim-v5-headline} summarizes end-to-end network recovery and
forecasting. Exact rows is the fraction of rows whose sparse support is
recovered exactly, and the final column reports HAR QLIKE minus full-model
QLIKE. The HAR benchmark is estimated separately for each asset in log realized
volatility using its daily value and its 5-day and 22-day own-history averages.
As $T$ rises from 1000 to 2500, F1 increases from $0.610$ to $0.750$ and
$\|\widehat M-M_0\|_F$ falls from $0.314$ to $0.197$.
Recovery is weaker at larger $N$ and with a rank-2 coefficient component.
The full model has lower QLIKE than HAR in each design. Additional recovery
and forecast diagnostics are reported in the Online Appendix.

\input{simulation_tables/table_v5_headline_main}

Table~\ref{tab:sim-v6-oracle-inference} examines the local refit used for
connectedness inference. The design sets $N=5$, $T=5000$, $m=3900$, the
completion rank to one, and the sparse row degree to one. The support and rank
are fixed at their DGP values, while the preliminary parameters are estimated
and $(\gamma,\beta,R,S)$ are jointly refitted. Each replication uses 99
bootstrap draws and applies the same hard stability constraint and local-KKT
screen as the empirical analysis. All 500 replications and all 49,500
bootstrap draws pass the prespecified numerical gate. Bootstrap-normal
95 percent coverage is $0.944$ for the combined coefficient $M_{ij}$ at a
randomly selected true sparse edge and $0.956$ for net connectedness of the
DGP node with the largest absolute net response.
These results evaluate the fixed-support, fixed-rank conditioning experiment
in Theorem~\ref{thm:conn-inference}; they are not unconditional post-selection
coverage probabilities.

\input{simulation_tables/table_v6_oracle_inference_main}

Table~\ref{tab:sim-v5-price} replaces proxy-level measurement error
with an intraday price experiment. We simulate diffusion returns,
finite-activity jumps, additive market microstructure noise, and their
combination, then reconstruct integrated volatility using a pre-specified
choice among realized variance, truncated realized variance, a realized
kernel, and two-scale realized variance. Proxy IV RMSE is relative
integrated-variance RMSE, and the number following the selected proxy is its
selection frequency. The rule selects RV under diffusion, TRV under jumps, and
TSRV under noise in every replication. Jumps alone leave network recovery
nearly unchanged after truncation. Under noise, proxy RMSE rises from $0.072$
to $0.270$ and support F1 falls from about $0.61$ to $0.46$.

\input{simulation_tables/table_v5_price_main}

The simulations show that feasible and latent-IV estimation approach each
other under joint in-fill growth, feasible tuning recovers part of the sparse
network without using the true support or rank, and the fixed-support,
fixed-rank refit has coverage close to the nominal level for the combined
network and connectedness targets. Microstructure noise reduces
network-recovery accuracy.

%% file: simulation_tables/table_v5_generated.tex
\begin{table}[!htbp]
\centering
\caption{Feasible and latent-IV estimators}
\label{tab:sim-v5-generated}
\resizebox{\textwidth}{!}{%
\begin{tabular}{lrrrrrr}
\toprule
Design & $T$ & $m$ & $R$ gap & $S$ gap & $M$ gap & QLIKE gap \\
\midrule
Baseline & 1000 & 390 & 0.037 (0.002) & 0.164 (0.002) & 0.162 (0.002) & 0.00059 (0.00001) \\
Long $T$ & 2500 & 390 & 0.032 (0.001) & 0.100 (0.001) & 0.099 (0.001) & 0.00058 (0.00001) \\
Low mesh & 1000 & 78 & 0.056 (0.002) & 0.256 (0.002) & 0.259 (0.002) & 0.00250 (0.00003) \\
Joint infill I & 2500 & 1950 & 0.024 (0.001) & 0.077 (0.002) & 0.076 (0.002) & 0.00007 (0.00000) \\
Joint infill II & 5000 & 3900 & 0.015 (0.000) & 0.042 (0.002) & 0.041 (0.001) & 0.00004 (0.00000) \\
\bottomrule
\end{tabular}%
}
\end{table}

%% file: simulation_tables/table_v5_headline_main.tex
\begin{table}[!htbp]
\centering
\caption{Feasible network recovery and forecasting}
\label{tab:sim-v5-headline}
\small
\begin{tabular}{lrrrrrrr}
\toprule
Design & $N$ & $T$ & $r_0$ & $\|\widehat M-M_0\|_F$ &
F1 & Exact rows & HAR $-$ full \\
\midrule
Baseline & 10 & 1000 & 1 & 0.314 & 0.610 & 0.148 &
\shortstack{0.00111\\(0.00002)} \\
Large $N$ & 15 & 1000 & 1 & 0.440 & 0.556 & 0.079 &
\shortstack{0.00085\\(0.00002)} \\
Long $T$ & 10 & 2500 & 1 & 0.197 & 0.750 & 0.294 &
\shortstack{0.00132\\(0.00002)} \\
Rank 2 & 10 & 1000 & 2 & 0.321 & 0.586 & 0.120 &
\shortstack{0.00119\\(0.00003)} \\
\bottomrule
\end{tabular}
\end{table}

%% file: simulation_tables/table_v6_oracle_inference_main.tex
\begin{table}[!htbp]
\centering
\caption{Fixed-support, fixed-rank inference}
\label{tab:sim-v6-oracle-inference}
\small
\begin{tabular}{lrrrr}
\toprule
Target & MC & Coverage & MCSE & Length \\
\midrule
$M_{ij}$ & 500 & 0.944 & 0.010 & 0.068 \\
Net connectedness & 500 & 0.956 & 0.009 & 0.147 \\
\bottomrule
\end{tabular}
\end{table}

%% file: simulation_tables/table_v5_price_main.tex
\begin{table}[!htbp]
\centering
\caption{Price-level high-frequency stress designs}
\label{tab:sim-v5-price}
\small
\begin{tabular}{llrrrr}
\toprule
Design & Selected proxy & Proxy IV RMSE & $\|\widehat M-M_0\|_F$ & F1 & QLIKE \\
\midrule
Diffusion & RV (1.00) & 0.072 & 0.310 & 0.609 & 0.0403 \\
Jumps & TRV (1.00) & 0.072 & 0.312 & 0.614 & 0.0404 \\
Jumps and noise & TSRV (1.00) & 0.270 & 0.367 & 0.463 & 0.0467 \\
Noise & TSRV (1.00) & 0.270 & 0.368 & 0.463 & 0.0469 \\
\bottomrule
\end{tabular}
\end{table}

%% file: generated_empirical/table_forecast_main.tex
\begin{tabular}{lrrrrr}
\toprule
Model & QLIKE & Gain (\%) & $t_{\mathrm{HAC}}$ & $p$-value & Years won \\
\midrule
Proposed network & 0.1770 & 0.0000 & -- & -- & -- \\
Daily-return proxy & 0.3066 & 42.2752 & 9.4169 & $<0.0001$ & 9 \\
Diagonal AR & 0.2014 & 12.1425 & 7.4637 & $<0.0001$ & 9 \\
HAR & 0.1796 & 1.4538 & 1.1467 & 0.2515 & 4 \\
VAR & 0.1935 & 8.5273 & 5.3180 & $<0.0001$ & 9 \\
\bottomrule
\end{tabular}

%% file: generated_empirical/table_sampling_proxy_robustness.tex
\begin{tabular}{lrrrrrrrr}
\toprule
Specification & Days & Proposed & Gain vs. HAR & $t$ vs. HAR & $p$ vs. HAR & Gain vs. VAR & $t$ vs. VAR & $p$ vs. VAR \\
\midrule
1-minute TSRV & 2239 & 0.1770 & 0.0026 & 1.1467 & 0.2515 & 0.0165 & 5.3180 & $<0.0001$ \\
1-minute realized kernel & 2239 & 0.1574 & 0.0024 & 1.0598 & 0.2892 & 0.0141 & 5.3259 & $<0.0001$ \\
5-minute threshold RV & 2236 & 0.1265 & 0.0014 & 0.7687 & 0.4421 & 0.0108 & 5.5206 & $<0.0001$ \\
\bottomrule
\end{tabular}

%% file: generated_empirical/table_connectedness_main.tex
\begin{tabular}{lrrlr}
\toprule
Object & Estimate & SE & 95\% interval & Adjusted $p$ \\
\midrule
\multicolumn{5}{l}{\emph{Panel A: system connectedness (pointwise)}} \\
Total connectedness & 0.9126 & 0.0035 & $[0.9058,\,0.9194]$ & -- \\
\addlinespace
\multicolumn{5}{l}{\emph{Panel B: node net connectedness (simultaneous max-$t$)}} \\
XLE (Energy) & 0.1686 & 0.0499 & $[0.0341,\,0.3031]$ & 0.0080 \\
XLY (Consumer Discretionary) & 0.0792 & 0.0590 & $[-0.0797,\,0.2382]$ & 0.8060 \\
XLI (Industrials) & 0.0296 & 0.0633 & $[-0.1411,\,0.2003]$ & 1.0000 \\
XLP (Consumer Staples) & -0.0139 & 0.0602 & $[-0.1764,\,0.1485]$ & 1.0000 \\
XLK (Technology) & -0.0231 & 0.0538 & $[-0.1682,\,0.1220]$ & 1.0000 \\
XLB (Materials) & -0.0293 & 0.0649 & $[-0.2043,\,0.1457]$ & 1.0000 \\
XLV (Health Care) & -0.0400 & 0.0617 & $[-0.2064,\,0.1265]$ & 0.9960 \\
XLF (Financials) & -0.0728 & 0.0469 & $[-0.1993,\,0.0536]$ & 0.6500 \\
XLU (Utilities) & -0.0982 & 0.0447 & $[-0.2189,\,0.0224]$ & 0.1840 \\
\bottomrule
\end{tabular}

%% file: simulation_tables/table_v5_headline.tex
\begin{table}[!htbp]
\centering
\caption{Complete network-recovery results}
\label{tab:oa-sim-v5-headline}
\resizebox{\textwidth}{!}{%
\begin{tabular}{lrrrrrrrrr}
\toprule
\multicolumn{10}{l}{Panel A: decomposition recovery} \\
\addlinespace
Design & $N$ & $T$ & $m$ & $r_0$ & $\widehat r_{\mathrm{diag}}$ &
$\kappa_{\mathrm{sep}}$ & $\|\widehat R-R_0\|_F$ &
$\|\widehat S-S_0\|_F$ & $\|\widehat M-M_0\|_F$ \\
\midrule
Baseline & 10 & 1000 & 390 & 1 & 1.560 & 0.838 & 0.195 & 0.319 & 0.314 \\
Large $N$ & 15 & 1000 & 390 & 1 & 1.272 & 0.751 & 0.217 & 0.416 & 0.440 \\
Long $T$ & 10 & 2500 & 390 & 1 & 2.116 & 0.838 & 0.174 & 0.220 & 0.197 \\
Rank 2 & 10 & 1000 & 390 & 2 & 1.548 & 0.924 & 0.226 & 0.338 & 0.321 \\
Rank 0 & 10 & 1000 & 390 & 0 & 1.038 & 0.000 & 0.062 & 0.281 & 0.293 \\
\midrule
\multicolumn{10}{l}{Panel B: sparse-support and sign recovery} \\
\addlinespace
Design & Precision & Recall & F1 & FP & FN & Sign rate &
Exact support & Exact sign & Exact rows \\
\midrule
Baseline & 0.557 & 0.684 & 0.610 & 11.108 & 6.318 & 0.682 & 0.000 & 0.000 & 0.148 \\
Large $N$ & 0.468 & 0.692 & 0.556 & 24.010 & 9.254 & 0.690 & 0.000 & 0.000 & 0.079 \\
Long $T$ & 0.652 & 0.890 & 0.750 & 9.758 & 2.192 & 0.890 & 0.000 & 0.000 & 0.294 \\
Rank 2 & 0.529 & 0.666 & 0.586 & 12.164 & 6.672 & 0.663 & 0.000 & 0.000 & 0.120 \\
Rank 0 & 0.618 & 0.733 & 0.666 & 9.450 & 5.334 & 0.732 & 0.000 & 0.000 & 0.199 \\
\bottomrule
\end{tabular}%
}
\end{table}

%% file: simulation_tables/table_v5_forecast.tex
\begin{table}[!htbp]
\centering
\caption{Feasible out-of-sample forecast comparison}
\label{tab:sim-v5-forecast}
\begin{tabular}{lrrrr}
\toprule
\multicolumn{5}{l}{Panel A: component ablations} \\
\addlinespace
Design & Full QLIKE & No factor $-$ full & No sparse $-$ full &
Own only $-$ full \\
\midrule
Baseline & 0.0405 (0.0001) & 0.0001 (0.0000) & 0.0006 (0.0000) & 0.0010 (0.0000) \\
Large $N$ & 0.0406 (0.0001) & 0.0001 (0.0000) & 0.0004 (0.0000) & 0.0007 (0.0000) \\
Long $T$ & 0.0401 (0.0000) & 0.0001 (0.0000) & 0.0010 (0.0000) & 0.0013 (0.0000) \\
Rank 2 & 0.0405 (0.0001) & 0.0001 (0.0000) & 0.0007 (0.0000) & 0.0011 (0.0000) \\
Rank 0 & 0.0403 (0.0001) & 0.0000 (0.0000) & 0.0005 (0.0000) & 0.0007 (0.0000) \\
\midrule
\multicolumn{5}{l}{Panel B: benchmark loss differences} \\
\addlinespace
Design & Diagonal $-$ full & HAR $-$ full &
\multicolumn{2}{c}{VAR $-$ full} \\
\midrule
Baseline & 0.0015 (0.0000) & 0.0011 (0.0000) & \multicolumn{2}{c}{0.0010 (0.0000)} \\
Large $N$ & 0.0013 (0.0000) & 0.0008 (0.0000) & \multicolumn{2}{c}{0.0011 (0.0000)} \\
Long $T$ & 0.0018 (0.0000) & 0.0013 (0.0000) & \multicolumn{2}{c}{0.0009 (0.0000)} \\
Rank 2 & 0.0016 (0.0000) & 0.0012 (0.0000) & \multicolumn{2}{c}{0.0009 (0.0000)} \\
Rank 0 & 0.0013 (0.0001) & 0.0009 (0.0000) & \multicolumn{2}{c}{0.0010 (0.0000)} \\
\bottomrule
\end{tabular}
\end{table}

%% file: simulation_tables/table_v5_coverage.tex
\begin{table}[!htbp]
\centering
\caption{Selected-estimator coverage}
\label{tab:sim-v5-selected-coverage-full}
\small
\begin{tabular}{lllrrr}
\toprule
\multicolumn{6}{l}{Panel A: design and selection} \\
\addlinespace
Design & MC & Selected & $\min|\Phi_{ij}|$ &
\multicolumn{2}{c}{$\Pr(e\in\widehat{\mathcal A})$} \\
\midrule
Smooth margin & 500 & 406 & 0.09316 & \multicolumn{2}{c}{0.812 (0.017)} \\
Strong margin & 200 & 187 & 0.08235 & \multicolumn{2}{c}{0.935 (0.017)} \\
Near-zero response & 200 & 188 & 0.00110 & \multicolumn{2}{c}{0.940 (0.017)} \\
\midrule
\multicolumn{6}{l}{Panel B: interval performance} \\
\addlinespace
Design & Target & Interval & Cov. (U) & Cov. (C) & Length \\
\midrule
Smooth margin & $S_{ij}$ & Normal & 0.808 (0.018) & 0.968 (0.009) & 0.126 (0.001) \\
Smooth margin & $M_{ij}$ & Normal & 0.850 (0.016) & 0.970 (0.008) & 0.123 (0.001) \\
Smooth margin & Net connectedness & Percentile & 0.960 (0.009) & -- & 0.136 (0.001) \\
Smooth margin & Net connectedness & Normal & 0.950 (0.010) & -- & 0.141 (0.001) \\
Strong margin & $S_{ij}$ & Normal & 0.910 (0.020) & 0.973 (0.012) & 0.130 (0.002) \\
Strong margin & $M_{ij}$ & Normal & 0.905 (0.021) & 0.968 (0.013) & 0.125 (0.001) \\
Strong margin & Net connectedness & Percentile & 0.955 (0.015) & -- & 0.126 (0.002) \\
Strong margin & Net connectedness & Normal & 0.955 (0.015) & -- & 0.136 (0.002) \\
Near-zero response & $S_{ij}$ & Normal & 0.870 (0.024) & 0.926 (0.019) & 0.123 (0.002) \\
Near-zero response & $M_{ij}$ & Normal & 0.935 (0.017) & 0.968 (0.013) & 0.116 (0.001) \\
Near-zero response & Net connectedness & Percentile & 0.750 (0.031) & -- & 0.102 (0.001) \\
Near-zero response & Net connectedness & Normal & 0.885 (0.023) & -- & 0.111 (0.002) \\
\bottomrule
\end{tabular}
\end{table}

%% file: simulation_tables/table_v5_price.tex
\begin{table}[!htbp]
\centering
\caption{Price-level high-frequency stress designs}
\label{tab:oa-sim-v5-price}
\resizebox{\textwidth}{!}{%
\begin{tabular}{llrrrrrrrr}
\toprule
Design & Selected proxy & RV & TRV & RK & TSRV & Used proxy & $M$ error & F1 & QLIKE \\
\midrule
Diffusion & RV (1.00) & 0.072 & 0.072 & 0.240 & 0.270 & 0.072 & 0.310 & 0.609 & 0.0403 \\
Jumps & TRV (1.00) & 0.104 & 0.072 & 0.240 & 0.270 & 0.072 & 0.312 & 0.614 & 0.0404 \\
Jumps and noise & TSRV (1.00) & 0.327 & 0.311 & 0.241 & 0.270 & 0.270 & 0.367 & 0.463 & 0.0467 \\
Noise & TSRV (1.00) & 0.309 & 0.309 & 0.240 & 0.270 & 0.270 & 0.368 & 0.463 & 0.0469 \\
\bottomrule
\end{tabular}%
}
\end{table}

%% file: generated_empirical/table_retained_days.tex
\begin{tabular}{rrrrr}
\toprule
Year & 1-min days & 1-min ratio & 5-min days & 5-min ratio \\
\midrule
2007 & 16 & 0.0637 & 149 & 0.5936 \\
2008 & 141 & 0.5573 & 246 & 0.9723 \\
2009 & 187 & 0.7421 & 249 & 0.9881 \\
2010 & 187 & 0.7421 & 251 & 0.9960 \\
2011 & 214 & 0.8492 & 251 & 0.9960 \\
2012 & 224 & 0.8960 & 247 & 0.9880 \\
2013 & 245 & 0.9722 & 249 & 0.9881 \\
2014 & 232 & 0.9206 & 248 & 0.9841 \\
2015 & 246 & 0.9762 & 249 & 0.9881 \\
2016 & 251 & 0.9960 & 251 & 0.9960 \\
2017 & 249 & 0.9920 & 249 & 0.9920 \\
2018 & 248 & 0.9880 & 247 & 0.9841 \\
2019 & 248 & 0.9841 & 249 & 0.9881 \\
2020 & 250 & 0.9881 & 247 & 0.9763 \\
2021 & 251 & 0.9960 & 251 & 0.9960 \\
2022 & 250 & 0.9960 & 250 & 0.9960 \\
2023 & 248 & 0.9920 & 248 & 0.9920 \\
2024 & 248 & 0.9841 & 248 & 0.9841 \\
2025 & 247 & 0.9880 & 247 & 0.9880 \\
\bottomrule
\end{tabular}

%% file: generated_empirical/table_forecast_summary.tex
\begin{tabular}{lrrrrrr}
\toprule
Model & QLIKE & Gain & Gain (\%) & $t_{\mathrm{HAC}}$ & $p$-value & Years won \\
\midrule
Proposed network & 0.1770 & 0.0000 & 0.0000 & -- & -- & -- \\
Remove sparse component & 0.1771 & 0.0001 & 0.0587 & 0.7372 & 0.4610 & 4 \\
Remove factor component & 0.1943 & 0.0174 & 8.9399 & 3.7175 & 0.0002 & 6 \\
Remove cross-network & 0.1945 & 0.0175 & 9.0176 & 3.6294 & 0.0003 & 6 \\
Daily-return proxy & 0.3066 & 0.1296 & 42.2752 & 9.4169 & $<0.0001$ & 9 \\
Diagonal AR & 0.2014 & 0.0245 & 12.1425 & 7.4637 & $<0.0001$ & 9 \\
HAR & 0.1796 & 0.0026 & 1.4538 & 1.1467 & 0.2515 & 4 \\
VAR & 0.1935 & 0.0165 & 8.5273 & 5.3180 & $<0.0001$ & 9 \\
\bottomrule
\end{tabular}

%% file: generated_empirical/table_forecast_by_asset.tex
\begin{tabular}{lrrr}
\toprule
ETF & Proposed & HAR & Gain vs. HAR \\
\midrule
XLU & 0.1336 & 0.1423 & 0.0087 \\
XLP & 0.1536 & 0.1605 & 0.0069 \\
XLB & 0.1695 & 0.1727 & 0.0033 \\
XLV & 0.1670 & 0.1700 & 0.0031 \\
XLI & 0.1891 & 0.1912 & 0.0021 \\
XLY & 0.2103 & 0.2109 & 0.0007 \\
XLE & 0.1438 & 0.1439 & 0.0001 \\
XLF & 0.1871 & 0.1866 & -0.0005 \\
XLK & 0.2388 & 0.2380 & -0.0008 \\
\bottomrule
\end{tabular}

%% file: generated_empirical/table_penalty_diagnostics.tex
\begin{tabular}{lrrrrrr}
\toprule
Specification & Min $c_S$ & Median $c_S$ & Max $c_S$ & Grid upper & Pointwise upper hits & One-SE upper hits \\
\midrule
1-minute TSRV & 0.2000 & 1.6000 & 12.8000 & 12.8000 & 2 & 8 \\
1-minute realized kernel & 0.2000 & 1.6000 & 12.8000 & 12.8000 & 1 & 9 \\
5-minute threshold RV & 0.2000 & 0.8000 & 12.8000 & 12.8000 & 1 & 9 \\
\bottomrule
\end{tabular}

%% file: generated_empirical/table_cross_proxy_edges.tex
\begin{tabular}{lrrr}
\toprule
Directed edge & TSRV & RK & 5-min \\
\midrule
\multicolumn{4}{c}{No directed edge satisfies all cross-proxy stability criteria} \\
\bottomrule
\end{tabular}

%% file: generated_empirical/table_cross_proxy_rankings.tex
\begin{tabular}{llrlrrrr}
\toprule
Specification & Top transmitter & Tx net & Top receiver & Rx net & Min $\widehat r_{\rm diag}$ & Median $\widehat r_{\rm diag}$ & Max $\widehat r_{\rm diag}$ \\
\midrule
1-minute TSRV & XLY & 0.0420 & XLU & -0.0349 & 1 & 2.0000 & 7 \\
1-minute realized kernel & XLE & 0.0292 & XLU & -0.0355 & 1 & 1.0000 & 7 \\
5-minute threshold RV & XLB & 0.0521 & XLV & -0.0423 & 1 & 1.0000 & 2 \\
\bottomrule
\end{tabular}

%% file: generated_empirical/table_rank_sensitivity.tex
\begin{tabular}{rrrlrlr}
\toprule
Rank & Total conn. & Spectral radius & Top transmitter & Tx net & Top receiver & Rx net \\
\midrule
0 & 0.0000 & 0.9623 & -- & 0.0000 & -- & 0.0000 \\
1 & 0.9126 & 0.9668 & XLE & 0.1686 & XLU & -0.0982 \\
2 & 0.8920 & 0.9668 & XLE & 0.1371 & XLU & -0.0737 \\
\bottomrule
\end{tabular}

%% file: generated_empirical/table_horizon_sensitivity.tex
\begin{tabular}{lrr}
\toprule
horizon & Total conn. & Spectral radius \\
\midrule
5 & 0.9781 & 0.9668 \\
22 & 0.9357 & 0.9668 \\
60 & 0.9168 & 0.9668 \\
infinite & 0.9126 & 0.9668 \\
\bottomrule
\end{tabular}

%% file: generated_empirical/table_regime_robustness.tex
\begin{tabular}{lrrrrrrrr}
\toprule
Period & 1-min days & 5-min days & $\|R\|_F$, 1-min & $\|S\|_F$, 1-min & $\|R\|_F$, 5-min & $\|S\|_F$, 5-min & Total, 1-min & Total, 5-min \\
\midrule
GFC & 344 & 585 & 0.2105 & 0.0000 & 0.2367 & 0.2957 & 0.8954 & 0.9034 \\
Calm & 1719 & 1742 & 0.1604 & 0.0257 & 0.1594 & 0.0000 & 0.9229 & 0.9225 \\
COVID & 501 & 498 & 0.4008 & 0.2708 & 0.7044 & 0.0000 & 0.9059 & 0.9108 \\
Tightening & 498 & 498 & 0.2643 & 0.0421 & 0.3769 & 0.2217 & 0.9105 & 0.9146 \\
Recent & 495 & 495 & 0.2462 & 0.0386 & 0.2115 & 0.1920 & 0.9224 & 0.9312 \\
\bottomrule
\end{tabular}